\documentclass[12pt]{article}
\usepackage{latexsym} \usepackage{epsf}
\usepackage{a4}\usepackage{amsfonts}\usepackage{amsmath}
\usepackage{amsthm}\usepackage{amssymb}\usepackage{cite}

\renewcommand{\theequation}{\thesection.\arabic{equation}}
\newcounter{subequation}[equation]
\makeatletter

\expandafter\let\expandafter\reset@font\csname reset@font\endcsname

\def\subeqnarray{\arraycolsep1pt
    \def\@eqnnum\stepcounter##1{\stepcounter{subequation}%
        {\reset@font\rm(\theequation\alph{subequation})}}
\jot5mm     \eqnarray}

\makeatother

\def\be{\begin{equation}}
\def\ee{\end{equation}}
\def\bea{\begin{eqnarray}}
\def\eea{\end{eqnarray}}
\def\ba{\begin{array}}
\def\ea{\end{array}}
\def\dd{\partial}

\def\one#1{#1^{\raise5pt\hbox{$\scriptstyle\!\!\!\!1$}}\,{}}
\def\two#1{#1^{\raise5pt\hbox{$\scriptstyle\!\!\!\!2$}}\,{}}

\def\tilde{\widetilde}
\def\II{\hbox{{1}\kern-.25em\hbox{l}}}
\makeatletter
\def\binrel@#1{\begingroup
  \setboxz@h{\thinmuskip0mu
    \medmuskip\m@ne mu\thickmuskip\@ne mu
    \setbox\tw@\hbox{$#1\m@th$}\kern-\wd\tw@
    ${}#1{}\m@th$}%
  \edef\@tempa{\endgroup\let\noexpand\binrel@@
    \ifdim\wdz@<\z@ \mathbin
    \else\ifdim\wdz@>\z@ \mathrel
    \else \relax\fi\fi}%
  \@tempa
}
\let\binrel@@\relax
\def\overset#1#2{\binrel@{#2}%
  \binrel@@{\mathop{\kern\z@#2}\limits^{#1}}}
\def\underset#1#2{\binrel@{#2}%
  \binrel@@{\mathop{\kern\z@#2}\limits_{#1}}}
\makeatother
\newfont{\bbd}{msbm10 scaled\magstep1}
\def\C{\hbox{\bbd C}}

\def\R{\hbox{\bbd R}}

\def\V{\hbox{\bbd V}}

\def\T{\hbox{\bbd T}}

\def\P{\hbox{\bbd P}}

\def\Z{\hbox{\bbd Z}}

\newtheorem{prop}{Proposition}

\newtheorem{theorem}[prop]{Theorem}

\def\stackreb#1#2{\ \mathrel{\mathop{#1}\limits_{#2}}}

\begin{document}

{\begin{center} {\Large \bf{Yang-Baxter equation, parameter permutations,\\
\medskip
and the elliptic beta integral}}%
\\ [8mm]
{\large \sf S. E. Derkachov$^{a}$\footnote{e-mail: derkach@pdmi.ras.ru},
V. P. Spiridonov$^b$\footnote{e-mail: spiridon@theor.jinr.ru } \\ [8mm] }
\end{center}

\begin{itemize}
\item[$^a$]
St. Petersburg Department of Steklov Mathematical Institute
of Russian Academy of Sciences, Fontanka 27, 191023 St. Petersburg, Russia.
\item[$^b$]
Bogoliubov Laboratory of Theoretical Physics, JINR,
 Dubna, Moscow reg. 141980, Russia and Max-Planck-Institut f\"ur Mathematik,
 Vivatsgasse 7, 53111, Bonn, Germany.

\end{itemize}

\vspace{1cm} \noindent {\bf Abstract}

We construct an infinite-dimensional solution of the Yang-Baxter equation (YBE) of rank 1
which is represented as an integral operator with an elliptic hypergeometric
kernel acting in the space of functions of two complex variables.
This R-operator intertwines the product of two standard L-operators
associated with the Sklyanin algebra, an elliptic deformation of $s\ell(2)$-algebra.
It is built from three basic operators
$\mathrm{S}_1, \mathrm{S}_2$, and $\mathrm{S}_3$ generating
the permutation group of four parameters $\mathfrak{S}_4$.
Validity of the key Coxeter relations (including the
star-triangle relation) is based on the elliptic
beta integral evaluation formula and the Bailey lemma
associated with an elliptic Fourier transformation.
The operators $\mathrm{S}_j$ are determined uniquely with the help
of the elliptic modular double.

{\small \tableofcontents}
\renewcommand{\refname}{References.}
\renewcommand{\thefootnote}{\arabic{footnote}}
\setcounter{footnote}{0}

\section{Introduction}
\setcounter{equation}{0}

The Yang-Baxter equation (YBE)
\begin{equation}\label{YB}
\mathbb{R}_{12} (u-v)\,\mathbb{R}_{13}(u)\, \mathbb{R}_{23}(v)
=\mathbb{R}_{23}(v)\,\mathbb{R}_{13}(u)\,\mathbb{R}_{12}(u-v)
\end{equation}
plays a key role in the theory of completely integrable quantum
systems~\cite{Baxter,FT,KS1,KS2,Jimbo,Faddeev,PA}. Its general solution
is described by the operators $\mathbb{R}_{ik}(u)$ acting in the
tensor product of three (in general different) spaces $\V_1\otimes\V_2\otimes\V_3$.
The indices $i$ and $k$ show that $\mathbb{R}_{ik}(u)$ acts nontrivially in the
tensor product $\V_i\otimes\V_k$ and it is the unity operator in the remaining
part of $\V_1\otimes\V_2\otimes\V_3$.
The operator $\mathbb{R}_{ik}(u)$ depend on the complex spectral parameter $u$
and is called the R-matrix (or R-operator).

For spaces $\V_i$ of finite dimension
there are three increasing levels of complexity of known
YBE solutions described by rational, trigonometric, and
elliptic functions. Investigation of the most complicated elliptic level
was initiated by Baxter \cite{Baxter1} for the case when all
$\V_i$-spaces are two-dimensional. In a more general setting,
when one of the spaces becomes infinite-dimensional, a major role is
played by the Sklyanin algebra \cite{Sklyanin1,Sklyanin2}.
Our main goal consists in the construction of a solution of the Yang-Baxter
equation when all spaces $\V_i$ are infinite-dimensional
and $\mathbb{R}_{ik}(u)$ are described by integral operators.
In this case the hierarchy of solutions of YBE
is attached to the plain hypergeometric, $q$-hypergeometric,
and elliptic hypergeometric functions \cite{spi:essays},
in the increasing order of complexity.
An explicit realization of the R-matrix as an integral operator
in the simplest situation of rank 1 symmetry algebra $s\ell(2)$
has been considered in detail in  \cite{DKK0}.
In the present work we discuss only rank 1
R-operators related to the most complicated elliptic level.

Elliptic hypergeometric integrals were introduced in \cite{spi:umn,S2}.
They define the general form of elliptic hypergeometric functions
which cannot be approached by infinite series \cite{FTur}
because of the convergence problems. The discovery of such functions
and various relations for them formed a breakthrough in the theory of
special functions. These functions generalize all previously constructed
functions of hypergeometric type and they still obey the properties
characteristic to classical special functions \cite{aar}.
In particular, elliptic beta integrals \cite{spi:umn,spi:essays}
form a new class of exactly computable integrals generalizing
the Euler, Selberg and other known beta integrals and their $q$-analogues.
A kind of elliptic Fourier transform was introduced in \cite{spi:bailey}
as an integral generalization of the well known Bailey chain transformation \cite{aar}.
An elliptic extension of Faddeev's modular double \cite{fad:mod}
was introduced in \cite{AA2008}. All these constructions will play a major role
in our consideration below.

We use the general strategy of building integral operator solutions of YBE
whose initial steps were discovered in \cite{SD}. Its formulation
was completed at the rational level
in \cite{DM0,DM}, where an $SL(N,\mathbb{C})$-invariant solution of YBE
related to $A_n$-root system has been constructed.
In \cite{DKK} this method was used also at the elliptic level
employing some formal infinite series.
Here we apply it for constructing solutions of YBE
related to elliptic hypergeometric integrals. First we define some useful formal operators
$\mathrm{S}_{1}\,,\mathrm{S}_{2}$, and $\mathrm{S}_{3}$ performing elementary
permutations of
parameters in the defining RLL-relation. One of these operators
is an intertwining operator of the Sklyanin algebra. Then we build these operators
explicitly as integral operators with elliptic hypergeometric kernels.
Finally we prove the Coxeter relations for these operators,
for which the elliptic beta integral and the related integral Bailey lemma
play a crucial role, and confirm that they indeed generate the group $\mathfrak{S}_4$.
The cubic Coxeter relation represents the star-triangle relation.
The operators $\mathrm{S}_{1}\,,\mathrm{S}_{2}$, and $\mathrm{S}_{3}$
are determined essentially uniquely, if one implements the elliptic modular
doubling principle. As discussed in the concluding section, our results
have applications to an interplay between integrable $2d$ spin systems
and $4d$ supersymmetric gauge field theories.

\section{Sklyanin algebra}
\label{Sklyanin} \setcounter{equation}{0}

In the simplest case of equation \eqref{YB} all $\V_i$-spaces are two-dimensional,
$\V_i =\mathbb{C}^2$. For this case, in solving the eight-vertex model Baxter
has found the following R-matrix ~\cite{Baxter,Baxter1,FT}
\begin{equation}\label{Baxter}
\mathbb{R}_{12}(u) = \sum_{a=0}^3 w_{a} (u)\,
\sigma_a \otimes\sigma_a \ \ ;\ \ w_{a}(u) = \frac{\theta_{a+1}
(u+\eta)}{\theta_{a+1}(\eta)}\,,
\end{equation}
where $\sigma_0=\II$ and $\sigma_\alpha, \alpha=1,2,3,$ are the standard Pauli
matrices. We use the shorthand notation $\theta_j(x)\equiv \theta_j(x|\tau)$
for the Jacobi theta-functions with modular parameter $\tau$.
The definitions and useful formulae for $\theta_j$-functions are listed
in the Appendix. This R-matrix depends on the spectral parameter $u\in\mathbb{C}$
and two additional free variables $\eta\in\mathbb{C}$, $\theta_j(\eta)\neq0, j=1,\ldots,4,$
and $\tau\in\mathbb{C}$, Im$(\tau)>0$. Another  $4\times 4$ matrix solution of YBE
has been found by Felderhof \cite{Fel}. As shown by Krichever \cite{Kri}, the
Baxter and Felderhof R-matrices represent all solutions of YBE for $\V_i =\mathbb{C}^2$.

At the next level of complexity of relation \eqref{YB} one of the spaces,
say $\V_3$, is arbitrary, and $\V_1, \V_2 =\mathbb{C}^2$. In this case the R-matrix
$\mathbb{R}_{13}(u)\equiv\mathrm{L}_{13}(u)$ (and $\mathbb{R}_{23}(u)
\equiv\mathrm{L}_{23}(u)$)
is known as the quantum L-operator or the Lax matrix.
It is a matrix with two rows and two columns acting in $\V_1$
\begin{equation}\makebox[-1em]{}
\mathrm{L}_{13}(u)=\mathrm{L}(u):=
\sum_{a=0}^3 w_{a} (u)\, \sigma_a \otimes
\mathbf{S}^a = \left(
\begin{array}{cc}
w_0(u)\,\mathbf{S}^0+w_3(u)\,\mathbf{S}^3 &
w_1(u)\,\mathbf{S}^1-\textup{i} w_2(u)\,\mathbf{S}^2 \\
w_1(u)\,\mathbf{S}^1+\textup{i} w_2(u)\,\mathbf{S}^2&
w_0(u)\,\mathbf{S}^0-w_3(u)\,\mathbf{S}^3
\end{array} \right),
\label{L_op}\end{equation}
where the matrix element entries $\mathbf{S}^a$ are some
operators acting in $\V_3$. The same $\mathbf{S}^a$-operators
enter $\mathrm{L}_{23}(u)$ which acts as a $2\times 2$ matrix in $\V_2$.
In this case the equation for L-operator is the Yang-Baxter
equation of the form
\begin{equation}\label{Lax}
\mathbb{R}_{12} (u-v)\,\mathrm{L}_{13}(u)\, \mathrm{L}_{23}(v)
=\mathrm{L}_{23}(v)\,\mathrm{L}_{13}(u)\,\mathbb{R}_{12}(u-v)\,,
\end{equation}
where $\mathbb{R}_{12}(u)$ is Baxter's R-matrix~(\ref{Baxter}).
This equation is equivalent to the following set of relations for
four operators $\mathbf{S}^0, \mathbf{S}^1,\mathbf{S}^2,\mathbf{S}^3$
forming the Sklyanin algebra \cite{Sklyanin1,Sklyanin2}:
$$
\mathbf{S}^\alpha\,\mathbf{S}^\beta - \mathbf{S}^\beta\,\mathbf{S}^\alpha =
\textup{i}\cdot\left(\mathbf{S}^0\,\mathbf{S}^\gamma +\mathbf{S}^\gamma\,\mathbf{S}^0\right)\,,
$$
$$
\mathbf{S}^0\,\mathbf{S}^\alpha - \mathbf{S}^\alpha\,\mathbf{S}^0 =
\textup{i}\,\mathbf{J}_{\beta \gamma}\cdot\left(\mathbf{S}^\beta\,\mathbf{S}^\gamma +\mathbf{S}^\gamma\,\mathbf{S}^\beta\right)\,,
$$
where the triplet $(\alpha,\beta,\gamma)$ is an arbitrary cyclic permutation
of $(1,2,3)$ and the structure constants $\mathbf{J}_{\beta \gamma}$ are
parameterized in terms of theta functions as
\begin{equation}\label{Jik}
\mathbf{J}_{12}=\frac{\theta_1^2(\eta)\theta_4^2(\eta)}
{\theta_2^2(\eta)\theta_3^2(\eta)}\ ;\quad
\mathbf{J}_{23}=\frac{\theta_1^2(\eta)\theta_2^2(\eta)}
{\theta_3^2(\eta)\theta_4^2(\eta)}\ ;\quad
\mathbf{J}_{31}= -\frac{\theta_1^2(\eta)\theta_3^2(\eta)}
{\theta_2^2(\eta)\theta_4^2(\eta)}\,.
\end{equation}
One can write $\mathbf{J}_{\alpha\beta}=
\frac{\mathbf{J}_{\beta}-\mathbf{J}_{\alpha}}{\mathbf{J}_{\gamma}}$,
$\gamma\neq \alpha,\beta$,  where
$$
\mathbf{J}_{1}=\frac{\theta_2(2\eta)\theta_2(0)}
{\theta_2^2(\eta)}\ ;\quad
\mathbf{J}_{2}=\frac{\theta_3(2\eta)\theta_3(0)}
{\theta_3^2(\eta)}\ ;\quad
\mathbf{J}_{3}= \frac{\theta_4(2\eta)\theta_4(0)}
{\theta_4^2(\eta)}\,.
$$
There are two Casimir operators commuting with all generators:
$\left[\mathbf{K}_0 ,\mathbf{S}^a\right] =
\left[\mathbf{K}_2 ,\mathbf{S}^a\right] =0$,
$$
\mathbf{K}_0 = \sum_{a=0}^3\,\mathbf{S}^a\,\mathbf{S}^a\ ;\qquad
\mathbf{K}_2 = \sum_{\alpha=1}^3\,\mathbf{J}_\alpha\,\mathbf{S}^\alpha\,\mathbf{S}^\alpha\,.
$$

We shall use the explicit realization of generators as difference
operators found in \cite{Sklyanin2}
\begin{equation}\label{Sklyan}
\left[\mathbf{S}^a\,\Phi\right](z) =\frac{(\textup{i})^{\delta_{a,2}}
\theta_{a+1}(\eta)}{\theta_1(2 z) } \Bigl[\,\theta_{a+1} \left(2
z-2\eta\ell\right)\cdot \Phi(z+\eta) - \theta_{a+1}
\left(-2z-2\eta\ell\right)\cdot \Phi(z-\eta)\, \Bigl]\,,
\end{equation}
where $\Phi(z)$ are some (supposedly meromorphic) functions of $z\in\C$.
In this realization the Casimir operators reduce to the following
scalar expressions
$$
\mathbf{K}_0 = 4\,\theta_1^2\bigl((2\ell+1)\,\eta\bigr)\ ;\quad
\mathbf{K}_2 = 4\,\theta_1\bigl(2(\ell+1)\,\eta\bigr)\,
\theta_1(2\ell\,\eta)\,.
$$
The variable $\ell\in\mathbb{C}$ is called the {\em spin}.
It labels the Sklyanin algebra representations since it fixes
(together with $\eta$ and $\tau$) the Casimir operator values.
Note that $\mathrm{R}$-matrix \eqref{Baxter} is invariant under the
change of variables $u\to -u, \eta\to -\eta$.
However, for operators \eqref{Sklyan} the reflection $\eta\to -\eta$
changes their sign, $\mathbf{S}^a\to -\mathbf{S}^a$. Therefore,
the $\mathrm{L}$-operator changes the sign if one negates simultaneously
the spectral parameter $u$ and $\eta$.

For $\mathbf{S}^a$-operators \eqref{Sklyan} there exists a useful
factorized representation for the L-operator
$$
\mathrm{L}(u_1,u_2) = \frac{1}{\theta_1(2 z)} \left(
\begin{array}{cc}
\bar{\theta}_3\left(z - u_1\right) & -\bar{\theta}_3\left(z+u_1\right) \\
-\bar{\theta}_4\left(z - u_1\right) & \bar{\theta}_4\left(z+u_1\right)
\end{array} \right)
\left(
\begin{array}{cc}
\mathrm{e}^{\eta\dd_z} &0\\
0 & \mathrm{e}^{-\eta\dd_z }
\end{array} \right )
\left(
\begin{array}{cc}
\bar{\theta}_4\left(z+u_2\right) & \bar{\theta}_3\left(z+u_2\right) \\
\bar{\theta}_4\left(z - u_2\right) & \bar{\theta}_3\left(z - u_2\right)
\end{array} \right),
$$
where $\mathrm{e}^{\eta\dd_z}$ is a shift operator, $\mathrm{e}^{\eta\dd_z}f(z)=f(z+\eta)$.
New parameters $u_1 = \frac{u}{2} + \eta(\ell+\frac{1}{2})$ and
$u_2 = \frac{u}{2} - \eta(\ell +\frac{1}{2})$  are simple linear
combinations of the spectral parameter $u$ and the spin $\ell$.
Here we use notation
$\bar\theta_a(x)\equiv \theta_a(x|{\textstyle\frac{\tau}{2}})$ for
theta-functions with the modular parameter ${\textstyle\frac{\tau}{2}}$.

When $2\ell+1$ is a positive integer there exists
$(2\ell+1)$-dimensional space $\Theta^+_{4\ell}$ of even
theta-functions of order $4\ell$ (having $4\ell$ zeros in the
fundamental parallelogram of periods) which is invariant
under the action of generators $\mathbf{S}^a$. For $\ell=1/2$
one can choose the basis of $\Theta^+_2$-functions as $e_1=\bar\theta_4(x)$
and $e_2=\bar\theta_3(x)$. Then the Sklyanin algebra generators reduce
in this basis to sigma-matrices $\mathbf{S}^a=\theta_1(2\eta|\tau)\sigma_a$
and the $\mathrm{L}$-operator becomes proportional to the Baxter R-matrix \eqref{Baxter}.

The $\mathrm{L}$-operator \eqref{L_op} is not unique. For instance, the operator
$\sigma_\beta\,\mathrm{L}$, where $\sigma_\beta$ is
any Pauli matrix, is also a solution of equation
\eqref{Lax}. This follows from the fact that the matrix
$X_\beta:=\sigma_\beta\otimes\sigma_\beta$ obeys the properties
$X_\beta^2=1$ and $X_\beta\mathbb{R}_{12}(u)X_\beta=\mathbb{R}_{12}(u)$.
This freedom leads also to the existence of
nontrivial automorphisms of the Sklyanin algebra~\cite{Sklyanin2}.

The top level of complexity of the R-matrix corresponds to
the situation when all spaces $\V_i$ are infinite-dimensional.
In this case one deals with the most general solutions of the Yang-Baxter equation.

In the next section we explain step-by-step our
strategy for building this solution, which is essentially
the same as in~\cite{DKK} where the important role is
played by an intertwining operator. For $2\ell\in\Z_{\geq 0}$
such an intertwining operator was constructed first by
Zabrodin~\cite{Z1} as a finite sum of the powers of the finite-difference
operator $e^{-\eta \partial_z}$. Its straightforward extension to arbitrary
values of $\ell$ as an infinite series proposed in \cite{Z2}
has only a formal meaning due to the convergency problem.
Nevertheless, the needed Coxeter relations were verified in~\cite{DKK}
by checking the equality of coefficients in two such formal infinite series
with the help of the Frenkel-Turaev summation formula~\cite{FTur}.

Here we put the construction on a firm mathematical ground
by using a different general ansatz for the intertwining operator
which is more useful for practical applications.
The key ingredients needed for the completion of this program
is the elliptic beta integral \cite{spi:umn,spi:essays},
the most general known exact integration formula generalizing the Euler
beta integral, and the elliptic Fourier transformation
of \cite{spi:bailey} which is defined precisely with the help of
needed intertwining operator.

\section{General construction}

We solve YBE \eqref{YB}
when all spaces $\V_i$ are infinite-dimensional in two steps:
\begin{itemize}
\item on the first stage, we solve a defining
$\mathrm{RLL}$-relation using as elementary building blocks some
simple operators $\mathrm{S}_1$, $\mathrm{S}_2$, and $\mathrm{S}_3$.
The key structural entries at this step are Coxeter relations
for $\mathrm{S}_i$ validity of which is guaranteed by the
elliptic beta integral evaluation formula~\cite{spi:umn};

\item on the second stage, we prove that the
operator $\mathbb{R}_{12}(u)$ found from the
defining $\mathrm{RLL}$-relation obeys YBE \eqref{YB}.
\end{itemize}

Consider a realization of YBE different
from the previous ones, namely, when the spaces $\V_1$ and $\V_2$
are arbitrary and the space $\V_3$
is two-dimensional. Then equation \eqref{YB} is reduced to
the defining equation for an infinite-dimensional (unknown) R-matrix
called $\mathrm{RLL}$-relation~\cite{KRS,Sklyanin1}:
\begin{equation}\label{RLL0}
\mathbb{R}_{12}(u-v)\,\mathrm{L}_1(u)\,\mathrm{L}_2(v)=
\mathrm{L}_2(v)\,\mathrm{L}_1(u)\,\mathbb{R}_{12}(u-v)\,.
\end{equation}
Here we use compact notation: the index $k$ in
$\mathrm{L}_k$ indicates that the Sklyanin algebra generators $\mathbf{S}^a_k$
entering this matrix are the operators acting in the space $\V_k$, i.e.
$\mathbf{S}_k^a:\V_k\to\V_k$.
The operators $\mathbf{S}_1^a$ and $\mathbf{S}_2^b$ act
in different spaces and, evidently, commute,
$\left[\mathbf{S}_1^a, \mathbf{S}_2^b\right] = 0$.
Matrices $\mathrm{L}_k$ in equation~(\ref{RLL0}) are
multiplied as usual $2\times 2$ matrices acting in the space $\V_3=\C^2$.

Due to the non-uniqueness of representation of the $\mathrm{L}$-operator,
there are several possible forms of equation \eqref{RLL0}
with different $\mathbb{R}$-operators
labeled by two indices $a$ and $b$ enumerating
possible $\mathrm{L}$-operators
$$
\mathbb{R}^{ab}_{12}(u-v)\,\sigma_a\,\mathrm{L}_1(u)
\,\sigma_b\,\mathrm{L}_2(v)=
\sigma_b\,\mathrm{L}_2(v)\,\sigma_a
\,\mathrm{L}_1(u)\,\mathbb{R}^{ab}_{12}(u-v)\,.
$$
For a technical reason, which will be clear \emph{a posteriori},
we fix $a=b=3$ from the very beginning and denote
$\mathbb{R}_{jk}(u):=\mathbb{R}^{33}_{jk}(u)$.
In this case it is possible to cancel one $\sigma_3$ and
our main defining $\mathrm{RLL}$-relation takes the form
\begin{equation}\label{RLL}
\mathbb{R}_{12}(u-v)\,\mathrm{L}_1(u)
\,\sigma_3\,\mathrm{L}_2(v)=
\mathrm{L}_2(v)\,\sigma_3
\,\mathrm{L}_1(u)\,\mathbb{R}_{12}(u-v)\,.
\end{equation}
We assume that $\V_1$ is the space of
functions of a complex variable $z_1$ and $\V_2$ is the
space of functions of a complex variable $z_2$. Respectively,
the space $\V_1\otimes\V_2$, where $\mathbb{R}_{12}$
is acting, is the space of functions $\Phi(z_1,z_2)$ of two
independent variables $z_1$ and $z_2$.

It is convenient to extract from the R-matrix the permutation operator
$\mathbb{R}_{12}(u): = \mathbb{P}_{12}\,\mathrm{R}_{12}(u)$, where
the permutation operator interchanges arguments,
$\mathbb{P}_{12}\,\Phi(z_1,z_2)=\Phi(z_2,z_1)$.
Then the defining equation for the operator $\mathrm{R}_{12}$ has the following
explicit form
\begin{equation}\label{RLL'}
\mathrm{R}_{12}(u-v)\,\mathrm{L}_1(u_1,u_2)\,\sigma_3\,\mathrm{L}_2(v_1,v_2)=
\mathrm{L}_1(v_1,v_2)\,\sigma_3\,\mathrm{L}_2(u_1,u_2)\,\mathrm{R}_{12}(u-v)\,,
\end{equation}
where the operators $z,\dd_z$ and $\ell$  in the Sklyanin algebra generators
\eqref{Sklyan} entering $\mathrm{L}_1$ are replaced by $z_1, \dd_{z_1}$
and $\ell_1$, whereas in $\mathrm{L}_2$ they are replaced by $z_2, \dd_{z_2}$
and $\ell_2$. We use also the following notation for combinations of the
spectral parameters and spin variables
\begin{eqnarray}\nonumber &&
u_1 =
\frac{u}{2}+ \eta\,\left(\ell_1+\frac{1}{2}\right)\, ,\qquad
u_2 = \frac{u}{2}-\eta\,\left(\ell_1+\frac{1}{2}\right)\, ,\
\\ &&
v_1 = \frac{v}{2} + \eta\,\left(\ell_2+\frac{1}{2}\right)\, ,\qquad
v_2 = \frac{v}{2}-\eta\,\left(\ell_2+\frac{1}{2}\right)\, .
\label{pars}\end{eqnarray}
For a subsequent use it is convenient to assume that
these parameters do not depend on $\eta$ and $\tau$
(i.e., to assume that the spectral parameters $u,v$
and the variables $g_{1,2}:=\eta(2\ell_{1,2}+1)$ are independent
on $\eta$ and $\tau$).

Equation~(\ref{RLL'}) admits a natural interpretation: the
operator $\mathrm{R}_{12}$ interchanges the set of
parameters $(u_1,u_2)$ from the first $\mathrm{L}$-operator with
the set of parameters $(v_1,v_2)$ in the second
$\mathrm{L}$-operator. It is useful to
combine these four parameters  in one set in the natural order
$\mathbf{u}\equiv(u_1,u_2,v_1,v_2)$. Then the operator
$$
\mathrm{R}_{12}(u-v)\equiv \mathrm{R}_{12}(\mathbf{u})\equiv
\mathrm{R}_{12}(u_1,u_2|v_1,v_2)
$$
corresponds to a particular permutation $s$ in the group of
permutations of four parameters $\mathfrak{S}_4$:
$$
s \rightarrow \mathrm{R}_{12}(\mathbf{u})\ ;\
s\mathbf{u}\equiv s(u_1,u_2,v_1,v_2)=(v_1,v_2,u_1,u_2).
$$
Any permutation from the group $\mathfrak{S}_4$ can be
composed from the elementary transpositions $s_{1}$,
$s_{2}$, and $s_{3}$:
$$
s_{1}\mathbf{u} = (u_2,u_1,v_1,v_2)\ , \
s_{2}\mathbf{u}
 = (u_1,v_1,u_2,v_2) \ , \
s_{3}\mathbf{u} = (u_1,u_2,v_2,v_1),
$$
which interchange only two nearest neighboring elements in the
set $(u_1,u_2,v_1,v_2)$. For example, the permutation $s$ has the
following decomposition $s = s_2 s_1s_3s_2$. It is natural to
search for the operators $\mathrm{S}_i(u_1,u_2,v_1,v_2)\equiv
\mathrm{S}_i(\mathbf{u})$ representing these elementary
transpositions in L-operators
$$
(\overset{\mathrm{S}_1}{\overbrace{u_1\ ,\ u_2}},
\overset{\mathrm{S}_{3}}{\overbrace{v_1\ ,\ v_2}})\ ;\ (u_1\
,\overset{\mathrm{S}_2}{\overbrace{u_2,v_1}},\ v_2).
$$
Namely, we demand that $\mathrm{S}_i$ obey the following defining relations
\begin{eqnarray}\label{RLL13} &&
\mathrm{S}_1(\mathbf{u})\,\mathrm{L}_1(u_1,u_2) =
\mathrm{L}_1(u_2,u_1)\,\mathrm{S}_1(\mathbf{u})\ ;\quad
\mathrm{S}_3(\mathbf{u})\,\mathrm{L}_2(v_1,v_2) =
\mathrm{L}_2(v_2,v_1)\,\mathrm{S}_3(\mathbf{u}),
\\ \label{RLL2}  &&
\mathrm{S}_2(\mathbf{u})\,\mathrm{L}_1(u_1,u_2)\,\sigma_3\,\mathrm{L}_2(v_1,v_2)=
\mathrm{L}_1(u_1,v_1)\,\sigma_3\,\mathrm{L}_2(u_2,v_2)\,\mathrm{S}_2(\mathbf{u})\,.
\end{eqnarray}
Since $\mathrm{R}_{12}$-matrix acts in the space $\V_1\otimes\V_2$, operators $\mathrm{S}_i$
should be scalars with respect to $\V_3=\mathbb{C}^2$. Moreover, it is natural
to demand that $\mathrm{S}_1$ commutes with $\mathrm{L}_2$
and $\mathrm{S}_3$ commutes with  $\mathrm{L}_1$:
\begin{equation}
\mathrm{S}_1(\mathbf{u})\mathrm{L}_2(v_1,v_2)
=\mathrm{L}_2(v_1,v_2)\mathrm{S}_1(\mathbf{u}),\qquad
\mathrm{S}_3(\mathbf{u})\mathrm{L}_1(u_1,u_2)=
\mathrm{L}_1(u_1,u_2)\mathrm{S}_3(\mathbf{u}).
\label{RLL4}\end{equation}

Our first step consists in the direct construction of these operators
(see the next section).
Having these operators we can build the R-matrix.

\begin{theorem} Suppose that formal scalar operators $\mathrm{S}_i$
satisfy relations~(\ref{RLL13})--(\ref{RLL4}).
Then the composite operator $\mathrm{R}_{12}(\mathbf{u})$,
\be\label{R}
\mathrm{R}_{12}(\mathbf{u}) =
\mathrm{S}_2(s_1s_3s_2\mathbf{u})\,
\mathrm{S}_1(s_3s_2\mathbf{u})\,\mathrm{S}_3(s_2\mathbf{u})\,
\mathrm{S}_2(\mathbf{u}),
\ee
satisfies equation~(\ref{RLL'}).
\end{theorem}
\begin{proof}
The proof reduces to a direct check, which is quite simple.
Namely, using equation \eqref{RLL2} we have
$$
\mathrm{R}_{12}(\mathbf{u})\,\mathrm{L}_1(u_1,u_2)\,\sigma_3\,\mathrm{L}_2(v_1,v_2)
=\mathrm{S}_2(s_1s_3s_2\mathbf{u})\,
\mathrm{S}_1(s_3s_2\mathbf{u})\, \mathrm{S}_3(s_2\mathbf{u})\,
\mathrm{L}_1(u_1,v_1)\,\sigma_3\,\mathrm{L}_2(u_2,v_2)\mathrm{S}_2(\mathbf{u}).
$$
Using the commutativity of $\mathrm{S}_3(s_2\mathbf{u})$ with $\sigma_3$
and $\mathrm{L}_1(u_1,v_1)$ \eqref{RLL4} and the second equation in
\eqref{RLL13}, we can rewrite this expression as
$$
\mathrm{S}_2(s_1s_3s_2\mathbf{u})\, \mathrm{S}_1(s_3s_2\mathbf{u})\,
\mathrm{L}_1(u_1,v_1)\,\sigma_3\,\mathrm{L}_2(v_2,u_2)\mathrm{S}_3(s_2\mathbf{u})
\,\mathrm{S}_2(\mathbf{u}).
$$
Now we apply the first equation in \eqref{RLL13} and commutativity
of $\mathrm{S}_1(s_3s_2\mathbf{u})$ with $\sigma_3$ and $\mathrm{L}_2(v_2,u_2)$
\eqref{RLL4} and obtain the expression
$$
\mathrm{S}_2(s_1s_3s_2\mathbf{u})\,
\mathrm{L}_1(v_1,u_1)\,\sigma_3\,\mathrm{L}_2(v_2,u_2)
 \mathrm{S}_1(s_3s_2\mathbf{u})\, \mathrm{S}_3(s_2\mathbf{u})
\,\mathrm{S}_2(\mathbf{u}).
$$
Finally, applying equation \eqref{RLL2} with $(u_1,u_2,v_1,v_2)$ replaced by
$(v_1,u_1,v_2,u_2)$ we obtain the right-hand side of equation~(\ref{RLL'}).
\end{proof}

Expression \eqref{R} for the R-matrix corresponds to a special
decomposition of the permutation $s$: $s = s_2 s_1s_3s_2$.
We will see that operators $\mathrm{S}_i$ depend on their parameters
in a special way
\begin{equation}\label{depend}
\mathrm{S}_1(\mathbf{u}) = \mathrm{S}_1(u_1-u_2)\ ;\
\mathrm{S}_2(\mathbf{u}) = \mathrm{S}_2(u_2-v_1)\ ;\
\mathrm{S}_3(\mathbf{u}) = \mathrm{S}_3(v_1-v_2),
\end{equation}
so that the operator $\mathrm{R}_{12}(\mathbf{u})$ depends on
the difference of spectral parameters $u-v$ as it should,
\begin{equation}
\mathrm{R}_{12}(u_1,u_2|v_{1},v_2) =
\mathrm{S}_2(u_1-v_2)\,
\mathrm{S}_1(u_1-v_1)\,\mathrm{S}_3(u_2-v_2)\,
\mathrm{S}_2(u_2-v_1)\,.
\end{equation}
We have thus the following correspondence between permutations $s_i$ and our
operators $\mathrm{S}_i$:
\begin{equation}
s_i \longrightarrow \mathrm{S}_i(\mathbf{u})\ ;\quad s_i s_j \longrightarrow
\mathrm{S}_i(s_j\mathbf{u})\,\mathrm{S}_j(\mathbf{u}).
\end{equation}

In order to prove that we have a representation of the permutation
group $\mathfrak{S}_4$ it remains to prove the defining Coxeter
relations for the generators
\begin{eqnarray}\label{def1}
&& s_i^2 = \II \longrightarrow
\mathrm{S}_i(s_i\mathbf{u})\,\mathrm{S}_i(\mathbf{u})= \II\ ;\quad
s_1s_3 = s_3s_1 \longrightarrow
\mathrm{S}_1(s_3\mathbf{u})\,\mathrm{S}_3(\mathbf{u})=
\mathrm{S}_3(s_1\mathbf{u})\,\mathrm{S}_1(\mathbf{u})\,,
\\  \label{def2}  &&
s_1 s_2 s_1 = s_2 s_1 s_2 \longrightarrow
\mathrm{S}_1(s_2s_1\mathbf{u})\,\mathrm{S}_2(s_1\mathbf{u})\,
\mathrm{S}_1(\mathbf{u})=
\mathrm{S}_2(s_1s_2\mathbf{u})\,\mathrm{S}_1(s_2\mathbf{u})\,
\mathrm{S}_2(\mathbf{u})\,,
\\ \label{def3} &&
s_2 s_3 s_2 = s_3 s_2 s_3 \longrightarrow
\mathrm{S}_2(s_3s_2\mathbf{u})\,\mathrm{S}_3(s_2\mathbf{u})\,
\mathrm{S}_2(\mathbf{u})=
\mathrm{S}_3(s_2s_3\mathbf{u})\,\mathrm{S}_2(s_3\mathbf{u})\,
\mathrm{S}_3(\mathbf{u})\,.
\end{eqnarray}
One can try to work with the equivalent power form of these relations connected
to reflection groups
$$
(s_is_j)^{m_{ij}}=1,\quad
m_{ii}=1,\quad m_{ij}=2,\; |i-j|>1,\quad m_{i,i\pm 1}=3,
$$
but it is much less efficient.
The explicit form of the operators $\mathrm{S}_i$ will be determined
in the next section. The proof of relations \eqref{def1}-\eqref{def3} will be given
in Sect.~\ref{Coxeter}.

Consider now the space $\V_1\otimes \V_2\otimes\V_3\otimes \mathbb{C}^2$,
where $\V_3$ is a new infinite-dimensional space of functions depending
on $z_3\in\mathbb C$. Introduce the third L-matrix $\mathrm{L}_3(w_1,w_2)$,
$w_{1,2}=\frac{w}{2}\pm \eta(\ell_3+\frac{1}{2})$, where $w$ is a new
spectral parameter and $\ell_3$ is a new spin variable in the Sklyanin algebra
generators \eqref{Sklyan} with $z,\dd_z$ replaced by $z_3, \dd_{z_3}$.

It is natural to form the set ${\bf u}=(u_1,u_2,v_1,v_2,w_1,w_2)$
and consider the group of permutations of six parameters $\mathfrak{S}_6$.
In addition to the previous case we have two more elementary permutation
generators $s_4{\bf u}=(u_1,u_2,v_1,w_1,v_2,w_2)$ and
$s_5{\bf u}=(u_1,u_2,v_1,v_2,w_2,w_1)$. We define operators $\mathrm{S}_4$
and $\mathrm{S}_5$ such that the triple $\{\mathrm{S_3},\mathrm{S}_4,
\mathrm{S}_5\}$ has the same properties as the triple $\{\mathrm{S}_1,\mathrm{S}_2,
\mathrm{S}_3\}$ after the replacement of parameters $(u_1,u_2,v_1,v_2)$
by $(v_1,v_2,w_1,w_2)$. More precisely, we demand that
\begin{eqnarray}\nonumber &&
\mathrm{S}_5(\mathbf{u})\,\mathrm{L}_3(w_1,w_2) =
\mathrm{L}_3(w_2,w_1)\,\mathrm{S}_5(\mathbf{u})\ ,\
\\ &&
\mathrm{S}_4(\mathbf{u})\,\mathrm{L}_2(v_1,v_2)\,\sigma_3\,\mathrm{L}_3(w_1,w_2)=
\mathrm{L}_2(v_1,w_1)\,\sigma_3\,\mathrm{L}_3(v_2,w_2)\,\mathrm{S}_4(\mathbf{u})\, ,
\label{RLL5}\end{eqnarray}
and that $\mathrm{S}_5$ commutes with $\mathrm{S}_{1,2,3}$ and
$\mathrm{S}_4$ commutes with $\mathrm{S}_{1,2}$.

Introduce the composite operator similar to $\mathrm{R}_{12}({\bf u})$,
\begin{eqnarray}\label{R23} &&
\mathrm{R}_{23}({\bf u})\equiv\mathrm{R}_{23}(v_1,v_2|w_1,w_2) =
\mathrm{S}_4(s_3s_5s_4\mathbf{u})\,
\mathrm{S}_3(s_5s_4\mathbf{u})\,\mathrm{S}_5(s_4\mathbf{u})\,
\mathrm{S}_4(\mathbf{u}) \qquad
\\ && \makebox[4em]{}
=\mathrm{S}_4(v_1-w_2)\,
\mathrm{S}_3(v_1-w_1)\,\mathrm{S}_5(v_2-w_2)\,
\mathrm{S}_4(v_2-w_1).
\nonumber\end{eqnarray}

To define the R-matrix $\mathrm{R}_{13}$ we consider the action of permutation
operators $\P_{jk}$ on $\mathrm{S}_i({\bf u})$.
Conjugating relations \eqref{RLL13} by $\P_{12}$, one can see that
$\P_{12}\mathrm{S}_3\P_{12}$ should be identified with $\mathrm{S}_1$
having the same argument. Namely,
$$
\P_{12}\mathrm{S}_1(u_1-u_2)=\mathrm{S}_3(u_1-u_2)\P_{12},\qquad
\P_{12}\mathrm{S}_3(v_1-v_2)=\mathrm{S}_1(v_1-v_2)\P_{12}.
$$

Conjugating similarly \eqref{RLL2}, one cannot directly deduce properties
of $\mathrm{S}_2$. As we will see from the explicit form of this operator
derived later, $\P_{12}\mathrm{S}_2({\bf u})=\mathrm{S}_2({\bf u})\P_{12}$.
Relations $\P_{12}\mathrm{S}_{4,5}({\bf u})=\mathrm{S}_{4,5}({\bf u})\P_{12}$
are evident. Analogous considerations yield nontrivial commutation relations
$$
\P_{13}\mathrm{S}_2(u_2-v_1)=\mathrm{S}_4(u_2-v_1)\P_{13},\qquad
\P_{23}\mathrm{S}_5(w_1-w_2)=\mathrm{S}_3(w_1-w_2)\P_{23},
$$
etc. However, the operator $\P_{12}\mathrm{S}_4\P_{12}=\P_{23}\mathrm{S}_2\P_{23}$
cannot be expressed in terms of $\mathrm{S}_i({\bf u})$-operators.
Now we define
\begin{eqnarray} &&
\mathrm{R}_{13}({\bf u})\equiv \mathrm{R}_{13}(u_1,u_2|w_1,w_2)
= \P_{12}\mathrm{R}_{23}(u_1,u_2|w_1,w_2)\P_{12}
\nonumber \\ &&
=\P_{12}\mathrm{S}_4(u_1-w_2)\,
\mathrm{S}_3(u_1-w_1)\,\mathrm{S}_5(u_2-w_2)\,
\mathrm{S}_4(u_2-w_1)\P_{12}
\nonumber \\&&
=\P_{12}\mathrm{S}_4(u_1-w_2)\P_{12}\,
\mathrm{S}_1(u_1-w_1)\,\mathrm{S}_5(u_2-w_2)\,
\P_{12}\mathrm{S}_4(u_2-w_1)\P_{12}.
\label{R23'}\end{eqnarray}
Analogously,
$$
\mathrm{R}_{13}({\bf u})
= \P_{23}\mathrm{R}_{12}(u_1,u_2|w_1,w_2)\P_{23}
=\P_{23}\mathrm{S}_2(u_1-w_2)\P_{23}\,
\mathrm{S}_1(u_1-w_1)\,\mathrm{S}_5(u_2-w_2)\,
\P_{23}\mathrm{S}_2(u_2-w_1)\P_{23}.
$$
We thus see that the operator $\mathrm{R}_{13}({\bf u})$ cannot be
factorized purely in terms of the operators $\mathrm{S}_i$.

\begin{theorem}
Suppose we have a set of well-defined operators $\mathrm{S}_i({\bf u}),
i=1,\ldots, 5,$ satisfying intertwining relations \eqref{RLL13}-\eqref{RLL4},
\eqref{RLL5} and the $\mathfrak{B}_6$-braid group generating relations
\begin{equation}
\mathrm{S}_j\mathrm{S}_k=\mathrm{S}_k\mathrm{S}_j,
\  |j-k|>1, \qquad \mathrm{S}_j\mathrm{S}_{j+1}\mathrm{S}_j
=\mathrm{S}_{j+1}\mathrm{S}_j\mathrm{S}_{j+1}.
\label{braid}\end{equation}
Then the $\mathrm{R}$-matrices
$$
\mathbb{R}_{12} (u-v)=\mathbb{P}_{12}\mathrm{R}_{12}(\mathbf{u}), \quad
\mathbb{R}_{23} (v-w)=\mathbb{P}_{23}\mathrm{R}_{23}(\mathbf{u}), \quad
\mathbb{R}_{13} (u-w)=\mathbb{P}_{13}\mathrm{R}_{13}(\mathbf{u}), \quad
$$
where operators $\mathrm{R}_{ij}(\mathbf{u})$ are fixed in \eqref{R},
\eqref{R23}, and \eqref{R23'}, satisfy the Yang-Baxter equation
\be
\mathbb{R}_{12} (u-v)\,\mathbb{R}_{13}(u-w)\, \mathbb{R}_{23}(v-w)
=\mathbb{R}_{23}(v-w)\,\mathbb{R}_{13}(u-w)\,\mathbb{R}_{12}(u-v).
\label{YBEfull}\ee
\end{theorem}
\begin{proof}
Consider the following permutation of parameters in
the product of three $\mathrm{L}$-operators:
$$
\mathrm{L}_1 (u_1,u_2)\,\sigma_3\, \mathrm{L}_2(v_1,v_2)\,\sigma_3\,
\mathrm{L}_3(w_1,w_2) \to \mathrm{L}_1 (w_1,w_2)\,\sigma_3\,
\mathrm{L}_2(v_1,v_2) \,\sigma_3\,\mathrm{L}_3(u_1,u_2).
$$
It can be realized in two different ways as shown on the figure below

\medskip

\vspace{5mm} \unitlength 0.8mm \linethickness{0.7pt}
\begin{picture}(152.67,83.33)
\put(30.00,50.00){\makebox(0,0)[cc]
{$\mathrm{L}_1(u_{1},u_2)\sigma_3\mathrm{L}_2(v_{1},v_2)\sigma_3
\mathrm{L}_3(w_{1},w_2)$}}
\put(32.00,55.00){\vector(0,1){20.00}}
\put(57.00,65.00){\makebox(0,0)[cc]
{$\mathrm{R}_{12}(u_{1},u_2|v_{1},v_2)$}}
\put(30.00,80.00){\makebox(0,0)[cc]
{$\mathrm{L}_1(v_{1},v_2)\sigma_3\mathrm{L}_2(u_{1},u_2)\sigma_3
\mathrm{L}_3(w_{1},w_2)$}}
\put(70.00,80.00){\vector(1,0){30.00}}
\put(85.00,85.00){\makebox(0,0)[cc]
{$\mathrm{R}_{23}(u_{1},u_2|w_1,w_2)$}}
\put(140.00,80.00){\makebox(0,0)[cc]
{$\mathrm{L}_1(v_{1},v_2)\sigma_3\mathrm{L}_2(w_{1},w_2)\sigma_3
\mathrm{L}_3(u_{1},u_2)$}}
\put(140.00,75.00){\vector(0,-1){20.00}}
\put(117.00,65.00){\makebox(0,0)[cc]
{$\mathrm{R}_{12}(v_1,v_2|w_{1},w_2)$}}
\put(140.00,50.00){\makebox(0,0)[cc]
{$\mathrm{L}_1(w_{1},w_2)\sigma_3\mathrm{L}_2(v_{1},v_2)\sigma_3
\mathrm{L}_3(u_{1},u_2)$}}
\put(32.00,45.00){\vector(0,-1){20.00}}
\put(57.00,35.00){\makebox(0,0)[cc]
{$\mathrm{R}_{23}(v_{1},v_{2}|w_1,w_2)$}}
\put(30.00,20.00){\makebox(0,0)[cc]
{$\mathrm{L}_1(u_{1},u_2)\sigma_3\mathrm{L}_2(w_{1},w_2)\sigma_3
\mathrm{L}_3(v_{1},v_2)$}}
\put(70.00,20.00){\vector(1,0){30.00}}
\put(85.00,15.00){\makebox(0,0)[cc]
{$\mathrm{R}_{12}(u_1,u_{2}|w_1,w_{2})$}}
\put(140.00,20.00){\makebox(0,0)[cc]
{$\mathrm{L}_1(w_{1},w_2)\sigma_3\mathrm{L}_2(u_{1},u_2)\sigma_3
\mathrm{L}_3(v_{1},v_2)$}}
\put(140.00,25.00){\vector(0,1){20.00}}
\put(117.00,35.00){\makebox(0,0)[cc]
{$\mathrm{R}_{23}(u_{1},u_2|v_{1},v_2)$}}
\end{picture}

The condition  of commutativity of this diagram indicates that
\begin{eqnarray} \nonumber &&
\mathrm{R}_{23}(u_{1},u_{2}|v_1,v_2)\,\mathrm{R}_{12}(u_1,u_{2}|w_1,w_{2})\,
\mathrm{R}_{23}(v_{1},v_2|w_{1},w_2)\,
\\ &&
=
\mathrm{R}_{12}(v_{1},v_2|w_{1},w_2)\,\mathrm{R}_{23}(u_{1},u_2|w_1,w_2)\,
\mathrm{R}_{12}(u_1,u_2|v_{1},v_2).
\label{RRR}\end{eqnarray}

Let us prove this equality using the braid group generating relations \eqref{braid}
for operators $\mathrm{S}_j$. We start from the identity
$$
\mathrm{S}_2\mathrm{S}_3\mathrm{S}_1\mathrm{S}_2\cdot
\mathrm{S}_4\mathrm{S}_3\mathrm{S}_5\mathrm{S}_4\cdot
\mathrm{S}_2\mathrm{S}_1\mathrm{S}_3\mathrm{S}_2
=
\mathrm{S}_2\mathrm{S}_3\mathrm{S}_4\mathrm{S}_1\cdot
\mathrm{S}_3\mathrm{S}_2\mathrm{S}_3\cdot\mathrm{S}_1
\mathrm{S}_5\mathrm{S}_4\mathrm{S}_3\mathrm{S}_2.
$$
The left-hand side is equal to the product of R-matrices
in the left-hand side of \eqref{RRR} under the taken convention
$\mathrm{S}_j\mathrm{S}_k:=\mathrm{S}_j(s_k{\bf u})\mathrm{S}_k({\bf u})$.
The right-hand side is obtained by permuting $\mathrm{S}_1\mathrm{S}_2$
with neighboring $\mathrm{S}_4$,  $\mathrm{S}_5\mathrm{S}_4$
with neighboring $\mathrm{S}_2\mathrm{S}_1$, and application
of the cubic relation from  \eqref{braid} to
the emerging product $\mathrm{S}_2\mathrm{S}_3\mathrm{S}_2$.
Now we replace $\mathrm{S}_1\cdot\mathrm{S}_3$ by $\mathrm{S}_3\mathrm{S}_1$,
permute $\mathrm{S}_1\mathrm{S}_5$ with neighboring $\mathrm{S}_3$, apply the
cubic relation to the emerging product $\mathrm{S}_1\mathrm{S}_2\mathrm{S}_1$
and obtain the left hand side of the relation
$$
\mathrm{S}_2\mathrm{S}_3\mathrm{S}_4\mathrm{S}_3\cdot
\mathrm{S}_2\mathrm{S}_1\mathrm{S}_2\cdot\mathrm{S}_5
\mathrm{S}_3\mathrm{S}_4\mathrm{S}_3\mathrm{S}_2
=
\mathrm{S}_4\mathrm{S}_2\mathrm{S}_3\mathrm{S}_2\cdot
\mathrm{S}_4\mathrm{S}_1\mathrm{S}_5\mathrm{S}_4\cdot
\mathrm{S}_2\mathrm{S}_3\mathrm{S}_2\mathrm{S}_4.
$$
The right-hand side is obtained after applying the cubic
relation to two products $\mathrm{S}_3\mathrm{S}_4\mathrm{S}_3$
and permuting three operators $\mathrm{S}_2$ with neighboring $\mathrm{S}_4$'s
and one $\mathrm{S}_2$ with neighboring $\mathrm{S}_5\mathrm{S}_4$. Now we permute
neighboring $\mathrm{S}_4$ and $\mathrm{S}_1$ and apply
cubic relations to the products $\mathrm{S}_2\mathrm{S}_3\mathrm{S}_2$ (twice)
and $\mathrm{S}_4\mathrm{S}_5\mathrm{S}_4$. This yields the left-hand
side of the relation
$$
\mathrm{S}_4\mathrm{S}_3\mathrm{S}_2\mathrm{S}_3\cdot
\mathrm{S}_1\mathrm{S}_5\mathrm{S}_4\mathrm{S}_5\cdot
\mathrm{S}_3\mathrm{S}_2\mathrm{S}_3\mathrm{S}_4
=
\mathrm{S}_4\mathrm{S}_3\mathrm{S}_5\mathrm{S}_4\cdot
\mathrm{S}_2\mathrm{S}_1\mathrm{S}_3\mathrm{S}_2\cdot
\mathrm{S}_4\mathrm{S}_3\mathrm{S}_5\mathrm{S}_4.
$$
The right-hand side expression is obtained after pulling
$\mathrm{S}_5$-operators to the left and right from
$\mathrm{S}_4$, permuting neighboring $\mathrm{S}_3$
and $\mathrm{S}_1$, applying the cubic
relation to the emerging product $\mathrm{S}_3\mathrm{S}_4\mathrm{S}_3$,
and, finally, pulling $\mathrm{S}_4$-operators to the left and right
from $\mathrm{S}_3$. And, evidently, it coincides with the right-hand side
expression of equality \eqref{RRR}.

Let us multiply the left-hand side expression in \eqref{RRR} by the operator
$\P_{12}\P_{13}\P_{23}$ and the right-hand side expression by the equal operator
$\P_{23}\P_{13}\P_{12}$.
Pulling permutation operators $\mathbb{P}_{jk}$ to appropriate R-matrices
using relations
$$
\P_{13}\P_{23}\mathrm{R}_{23}(u_1,u_2|v_1,v_2)\P_{23}\P_{13}=
\P_{13}\mathrm{S}_4(u_1-v_2)\,
\mathrm{S}_5(u_1-v_1)\,\mathrm{S}_3(u_2-v_2)\,
\mathrm{S}_4(u_2-v_1)\P_{13}
$$
$$
=\mathrm{S}_2(u_1-v_2)\,
\mathrm{S}_1(u_1-v_1)\,\mathrm{S}_3(u_2-v_2)\,
\mathrm{S}_2(u_2-v_1)=\mathrm{R}_{12}({\bf u})
$$
and
$$
\P_{13}\P_{12}\mathrm{R}_{12}(v_1,v_2|w_1,w_2)\P_{12}\P_{13}
=\P_{13}\mathrm{S}_2(v_1-w_2)\,
\mathrm{S}_3(v_1-w_1)\,\mathrm{S}_1(v_2-w_2)\,
\mathrm{S}_2(v_2-w_1)\,\P_{13}
$$
$$
=\mathrm{S}_4(v_1-w_2)\,
\mathrm{S}_3(v_1-w_1)\,\mathrm{S}_5(v_2-w_2)\,
\mathrm{S}_4(v_2-w_1)
=\mathrm{R}_{23}({\bf u}),
$$
one comes to the desired equation \eqref{YBEfull}.
\end{proof}

From this consideration we conclude that the Yang-Baxter
relation \eqref{RRR} is nothing else than a word identity in
the group algebra of the braid group $\mathfrak{B}_6$.
Equation \eqref{YBEfull} is more complicated since it
involves the external operators $\P_{jk}$. Note that
the described proof does not require the condition $\mathrm{S}_j^2=\II$
reducing $\mathfrak{B}_6$ to the permutation group $\mathfrak{S}_6$.
The operators $\mathrm{S}_j$ which we construct below do
satisfy relations $\mathrm{S}_j^2=\II$ after the analytical continuation
in parameters and, so, they generate the $\mathfrak{S}_6$-group.

\section{Elementary transpositions and intertwining operators}
\label{Sk} \setcounter{equation}{0}

We shall use the factorized form of the L-operator
which allows one to  simplify considerably all calculations
\begin{equation}
\mathrm{L}(u_1,u_2)= \frac{1}{\theta_1(2 z)} \cdot \mathrm{M}(z - u_1 ; z
+u_1)\cdot \left(
\begin{array}{cc}
\mathrm{e}^{\eta\dd} &0\\
0 & \mathrm{e}^{-\eta\dd }
\end{array} \right )\cdot \mathrm{N}(z - u_2 ; z +u_2),
\label{factL}\end{equation}
where
\begin{equation}\label{MN}
\mathrm{M}(a ; b) = \left(
\begin{array}{cc}
\bar{\theta}_3\left(a\right) & -\bar{\theta}_3\left(b\right) \\
-\bar{\theta}_4\left(a\right) & \bar{\theta}_4\left(b\right)
\end{array} \right )\ \ ;\ \
\mathrm{N}(a ; b) = \left(
\begin{array}{cc}
\bar{\theta}_4\left(b\right) & \bar{\theta}_3\left(b\right) \\
\bar{\theta}_4\left(a\right) & \bar{\theta}_3\left(a\right)
\end{array} \right )\,.
\end{equation}
To prove this factorization of the L-operator one has to multiply
explicitly all three matrices involved in it and use the addition formula
$$
\theta_1(x +y) \theta_1(x - y)+ \theta_4(x + y)\theta_4(x - y)
 =\bar\theta_4(x)\bar\theta_3(y)
$$
and its variations which are listed in the Appendix.
The product of matrices $N$ and $M$ has the form
\begin{equation}\label{NM}
\mathrm{N}(a_1 ; b_1)\cdot \mathrm{M}(a_2 ; b_2)= 2\cdot\left(
\begin{array}{cc}
\theta_1\left(b_1-a_2\right)\theta_1\left(b_1+a_2\right) &
-\theta_1\left(b_1-b_2\right)\theta_1\left(b_1+b_2\right) \\
\theta_1\left(a_1-a_2\right)\theta_1\left(a_1+a_2\right) &
-\theta_1\left(a_1-b_2\right)\theta_1\left(a_1+b_2\right)
\end{array} \right ),
\end{equation}
in particular,
\begin{equation}\label{NM1}
\mathrm{N}(a ; b)\cdot \mathrm{M}(a ; b) =
-2\cdot\theta_1\left(a-b\right)\theta_1\left(b+a\right) \left(
\begin{array}{cc}
1 & 0 \\
0 & 1
\end{array} \right )\ .
\end{equation}
To avoid lengthy formulae we use compact notation
$\mathrm{N}(a\mp b)\equiv\mathrm{N}(a-b ; a+b)$,
$\mathrm{M}(a\mp b)\equiv\mathrm{M}(a-b ; a+b)$
and $\theta_j(a,b)=\theta_j(a)\theta_j(b)$,
$\theta_j(a\mp b) = \theta_j\left(a-b\right)\theta_j\left(a+b\right),$
so that
$$
\mathrm{L}(u_1,u_2)= \frac{1}{\theta_1(2 z)} \cdot \mathrm{M}(z \mp u_1)\cdot \left(
\begin{array}{cc}
\mathrm{e}^{\eta\dd} &0\\
0 & \mathrm{e}^{-\eta\dd }
\end{array} \right )\cdot \mathrm{N}(z \mp u_2).
$$

Consider first the defining relation for operator $\mathrm{S}_2$ (\ref{RLL2}),
$$
\mathrm{S}_2\,\underline{\mathrm{M}(z_1\mp u_1)}\, \left(
\begin{array}{cc}
\mathrm{e}^{\eta\dd_1} &0\\
0 & \mathrm{e}^{-\eta\dd_1}
\end{array} \right )\,
\mathrm{N}(z_1\mp u_2)\,\sigma_3\, \mathrm{M}(z_2\mp v_1)\,
\left(
\begin{array}{cc}
\mathrm{e}^{\eta\dd_2} &0\\
0 & \mathrm{e}^{-\eta\dd_2}
\end{array} \right )\,
\underline{\mathrm{N}(z_2\mp v_2)} =
$$
$$ =
\underline{\mathrm{M}(z_1\mp u_1)}\, \left(
\begin{array}{cc}
\mathrm{e}^{\eta\dd_1} &0\\
0 & \mathrm{e}^{-\eta\dd_1}
\end{array} \right )\,
\mathrm{N}(z_1\mp v_1)\,\sigma_3\, \mathrm{M}(z_2\mp u_2)\,
\left(
\begin{array}{cc}
\mathrm{e}^{\eta\dd_2} &0\\
0 & \mathrm{e}^{-\eta\dd_2}
\end{array} \right )\,
\underline{\mathrm{N}(z_2\mp v_2)}\,\mathrm{S}_2.
$$
We underlined the matrices which can be canceled under the commutativity
condition $[\mathrm{S}_2,z_1]=[\mathrm{S}_2,z_2]=0$. This
observation suggests that $\mathrm{S}_2$ is just a multiplication operator:
$$
\left[\mathrm{S}_2\Phi\right](z_1,z_2) =
\mathrm{S}(z_1,z_2)\cdot \Phi(z_1,z_2)\,.
$$
Consequently, the operator $\mathrm{S}_2$ commutes with the
matrices $\mathrm{M}(z_1\mp u_1)$ and
$\mathrm{N}(z_2\mp v_2)$, so that they both cancel
from the equation and we obtain a much simpler defining relation for the
function $\mathrm{S}(z_1,z_2)$:
$$
\mathrm{S}(z_1,z_2)\,\left(
\begin{array}{cc}
\mathrm{e}^{\eta\dd_1} &0\\
0 & \mathrm{e}^{-\eta\dd_1}
\end{array} \right )\cdot
\mathrm{N}(z_1\mp u_2)\,\sigma_3\,\mathrm{M}(z_2\mp v_1)\,
\left(
\begin{array}{cc}
\mathrm{e}^{\eta\dd_2} &0\\
0 & \mathrm{e}^{-\eta\dd_2}
\end{array} \right ) =
$$
$$ =
\left(
\begin{array}{cc}
\mathrm{e}^{\eta\dd_1} &0\\
0 & \mathrm{e}^{-\eta\dd_1}
\end{array} \right )\cdot
\mathrm{N}(z_1\mp v_1)\,\sigma_3\,\mathrm{M}(z_2\mp u_2)\,
\left(
\begin{array}{cc}
\mathrm{e}^{\eta\dd_2} &0\\
0 & \mathrm{e}^{-\eta\dd_2}
\end{array} \right )\,\mathrm{S}(z_1,z_2),
$$
or in the equivalent form
$$
\left(
\begin{array}{cc}
\mathrm{S}(z_1-\eta,z_2) &0\\
0 & \mathrm{S}(z_1+\eta,z_2)
\end{array} \right )\,
\mathrm{N}(z_1\mp u_2)\,\sigma_3\,\mathrm{M}(z_2\mp v_1)=
$$
$$
=
\mathrm{N}(z_1\mp v_1)\,\sigma_3\, \mathrm{M}(z_2\mp u_2)\,
\left(
\begin{array}{cc}
\mathrm{S}(z_1,z_2+\eta) &0\\
0 & \mathrm{S}(z_1,z_2-\eta)
\end{array} \right ).
$$
Using a theta functions identity given in the Appendix one can see that
\begin{equation}\label{NsM}
\mathrm{N}(a_1 ; b_1)\cdot \sigma_3\cdot\mathrm{M}(a_2 ; b_2)= 2\cdot\left(
\begin{array}{cc}
\theta_4\left(b_1\mp a_2\right) &
-\theta_4\left(b_1\mp b_2\right) \\
\theta_4\left(a_1\mp a_2\right) &
-\theta_4\left(a_1\mp b_2\right)
\end{array} \right )\,.
\end{equation}

The derived matrix equation can be simplified further on,
since a number of theta functions depending on the combination of parameters
$u_2 + v_1$ drops out from it. As a result, we come to a system of four linear finite
difference equations of the first order
$$
\theta_4(z_1+z_2 +u_2-v_1)\, \mathrm{S}(z_1-\eta,z_2) =
\theta_4(z_1+z_2 +v_1-u_2)\, \mathrm{S}(z_1,z_2+\eta)\ ,
$$
$$
\theta_4(z_1+z_2 -u_2 +v_1)\, \mathrm{S} (z_1 + \eta ,z_2) =
\theta_4(z_1+z_2 -v_1 +u_2)\, \mathrm{S}(z_1,z_2-\eta)\ ,
$$
$$
\theta_4(z_1-z_2 + u_2 - v_1)\, \mathrm{S} (z_1 - \eta,z_2) =
\theta_4(z_1-z_2 +v_1 -u_2)\, \mathrm{S} (z_1,z_2-\eta)\ ,
$$
$$
\theta_4(z_1-z_2 -u_2 +v_1)\, \mathrm{S} (z_1 + \eta,z_2) =
\theta_4(z_1-z_2 -v_1 +u_2)\, \mathrm{S} (z_1,z_2+\eta)\ .
$$
Their structure suggests  to look for the solution in the factorized form
$$
\mathrm{S}(z_1, z_2) = \Phi_+ (z_1+z_2) \cdot \Phi_- (z_1-z_2).
 $$
Then, in each equality one of the $\Phi_{\pm}$-factors drops out and we obtain
equations
\begin{equation}\label{pm}
\theta_4(z+u_2-v_1)\,\Phi_{\pm} (z-\eta) =
\theta_4(z-u_2+v_1)\,\Phi_{\pm} (z+\eta),
\end{equation}
or, in the equivalent form,
$$
\Phi_{\pm}(z+2\eta) =
\mathrm{e}^{2\pi \textup{i} (u_2-v_1)}\,\frac{\theta_1(z+u_2-v_1+\eta+\textstyle{\frac{\tau}{2}})}
{\theta_1(z-u_2+v_1+\eta+\textstyle{\frac{\tau}{2}})} \cdot
\Phi_{\pm}(z).
$$
In this section we suppose that Im$(\eta)>0$. Then
a particular solution of this equation is described by a ratio
of elliptic gamma functions $\Gamma(z) = \Gamma(z|\tau,2\eta)$
(see formula \eqref{egamma_a} in the Appendix),
\begin{equation}\label{ab}
\Phi(z+2\eta) = \mathrm{e}^{\pi \textup{i} (a-b)}\,
\frac{\theta_1(z+a)}{\theta_1(z+b)} \cdot \Phi(z)\ ;
\qquad \Phi(z) = \frac{\Gamma(z+a|\tau, 2\eta)}{\Gamma(z+b|\tau,2\eta)}, \;
\text{Im}(\eta), \text{Im}(\tau)>0.
\end{equation}
The equation for $\Phi(z)$ does not assume restrictions on $\eta$.
Its solution and corresponding intertwining operators valid for Im$(\eta)<0$
use the function $\Gamma(z|\tau,-2\eta)$ and for Im$(\eta)=0$
one needs the modified elliptic gamma function \cite{S2,spi:essays}. These solutions
will be described in a special section below.

Using \eqref{ab} we can write the general solution $\mathrm{S}(z_1,z_2)$ in
the form
\begin{equation}\label{S2}
\mathrm{S}(z_1,z_2) =
\frac{\Gamma(z_1+z_2 + u_2  - v_1 + \eta+\textstyle{\frac{\tau}{2}})}{\Gamma(z_1+z_2 - u_2+
v_1 + \eta+\textstyle{\frac{\tau}{2}})} \frac{\Gamma(z_1-z_2 + u_2 - v_1
+\eta+\textstyle{\frac{\tau}{2}})}{\Gamma(z_1-z_2 - u_2 + v_1  +
\eta+\textstyle{\frac{\tau}{2}}) }\varphi_2(z_1,z_2),
\end{equation}
where $\varphi_2$ is an arbitrary function satisfying the periodicity
conditions
$$
\varphi_2(z_1+2\eta,z_2)=\varphi_2(z_1,z_2+2\eta)=\varphi_2(z_1+\eta,z_2+\eta)
=\varphi_2(z_1,z_2).
$$
Using the reflection formula for the elliptic gamma function \eqref{refl}
and the notation
$$
\Gamma\left(\pm x \pm y +a\right)\equiv
\Gamma\left(-x - y +a\right)\,\Gamma\left(- x + y +a\right)
\,\Gamma\left(x - y +a\right)\,\Gamma\left(x + y +a\right),
$$
it is possible to rewrite formula \eqref{S2} in a much more compact form
$$
\mathrm{S}(z_1,z_2) =
\Gamma\left(\pm z_1\pm z_2 + u_2  - v_1 + \eta+\textstyle{\frac{\tau}{2}|\tau,2\eta}\right)
\varphi_2(z_1,z_2).
$$

The functional freedom $\varphi_2(z_1,z_2)$ strongly influencing the final results
will be fixed in the section dedicated to the elliptic modular double.
We remark also that our
choice of the L-operator in the form $\sigma_3\mathrm{L}$ is done
for technical reasons in order to have permutational symmetry $z_1\leftrightarrow z_2$
in \eqref{S2}, i.e. $\P_{12}\mathrm{S}_2({\bf u})=\mathrm{S}_2({\bf u})\P_{12}$,
 absent for other choices.

Let us consider now defining equations~(\ref{RLL13}) for
operators $\mathrm{S}_1$ and $\mathrm{S}_3$. Permutation of
the parameters $ u_1 = \frac{u}{2} + \eta\,(\ell_1+\frac{1}{2})$ and $u_2 =
\frac{u}{2} -\eta\,(\ell_1+\frac{1}{2}) $ is equivalent to the change of the
spin $\ell_1 \to -1-\ell_1$ and similarly the permutation of
parameters $ v_1 = \frac{v}{2} + \eta\,(\ell_2+\frac{1}{2})$ and $v_2 =
\frac{u}{2} -\eta\,(\ell_2+\frac{1}{2})$ is equivalent to the change
$\ell_2 \to -1-\ell_2$. In the L-matrix~(\ref{Lax}) only
generators $\mathbf{S}^a$ depend on the spin, therefore defining
equations~(\ref{RLL13}) can be rewritten in terms of the
$\mathbf{S}^a$-generators alone:
\begin{equation}
\mathrm{S}_1\cdot\mathbf{S}^a(\ell_1) =
\mathbf{S}^a(-1-\ell_1)\cdot \mathrm{S}_1\,, \qquad
\mathrm{S}_3\cdot \mathbf{S}^a(\ell_2) =
\mathbf{S}^a(-1-\ell_2)\cdot \mathrm{S}_3\, ,
\label{inter1}\end{equation}
where we explicitly indicate the spin $\ell$ dependence of
$\mathbf{S}^a$-operators. The meaning of these relations is the following:
the operator $\mathrm{S}_1$ intertwines representations with the spins
$\ell_1$ and $-1-\ell_1$ realized in the space of functions of the variable $z_1$
and the operator
$\mathrm{S}_3$ intertwines representations with the spins $\ell_2$
and $-1-\ell_2$ realized in the space of functions of variable $z_2$.
Note that in terms of the variable $g=\eta(2\ell+1)$ this corresponds
to the simple sign change $g\to -g$.

Evidently the operators $\mathrm{S}_1$ and $\mathrm{S}_3$ are equivalent to each other
differing only by the spaces where they are acting.
Let us construct the general intertwining operator $\mathrm{W}$
defined in two equivalent ways: either
as a solution of the matrix equation for the $\mathrm{L}$-operator
$$
\mathrm{W} \cdot \mathrm{L}(u_1,u_2) =
\mathrm{L}(u_2,u_1)\cdot \mathrm{W}\,,
$$
or, alternatively, as a solution of the system of equations involving generators
of the Sklyanin algebra
\begin{equation}\label{inter}
\mathrm{W} \cdot\mathbf{S}^a(\ell) = \mathbf{S}^a(-1-\ell)\cdot \mathrm{W}\,.
\end{equation}
For $2\ell\in\Z_{\geq 0}$ such an intertwining operator $\mathrm{W}$ was
constructed in~\cite{Z1} as a finite sum of the powers of
the finite-difference operator $e^{-\eta \partial_z}$ (see below).
A formal extension of this $\mathrm{W}$ to infinite series and arbitrary
values of $\ell$ proposed in \cite{Z2} does not represent a well defined operator.

In this paper we use a different approach inspired by the elliptic
hypergeometric integrals \cite{spi:umn,S2}.
Namely, we construct the intertwining operator $\mathrm{W}$ using a quite
general ansatz for it as an integral operator
\begin{eqnarray}\label{Wint} &&
\left[\mathrm{W}\,\Phi\right](z) = \int_\alpha^\beta
\, \Delta(z,x)\, \Phi(x) \ dx
\end{eqnarray}
for some integration interval $[\alpha,\beta]\in\R$.

We are going to solve thus the equation
$$
\int_\alpha^\beta\Delta(z,x)[\mathbf{S}^a(\ell)\Phi(x)]dx
=\mathbf{S}^a(-1- \ell)\int_\alpha^\beta\Delta(z,x)\Phi(x)dx,
$$
where $\mathbf{S}^a$-operators on the left-hand side act on the functions
of variable $x$, whereas on the right-hand side --- on the functions of variable $z$.
More explicitly, we have
\begin{eqnarray*} &&
\int_\alpha^\beta\frac{\Delta(z,x)}{\theta_1(2x)}
\Big(\theta_{a}(2x-2\eta\ell)\Phi(x+\eta)
-\theta_{a}(-2x-2\eta\ell)\Phi(x-\eta) \Big)dx
\\ &&
=\int_{\alpha+\eta}^{\beta+\eta}\frac{\theta_{a}(2x-2\eta(\ell+1))}{\theta_1(2x-2\eta)}
\Delta(z,x-\eta)\Phi(x)dx
-\int_{\alpha-\eta}^{\beta-\eta}\frac{\theta_{a}(-2x-2\eta(\ell+1))}{\theta_1(2x+2\eta)}
\Delta(z,x+\eta)\Phi(x)dx
\\ &&
=\int_\alpha^\beta\Big(\frac{\theta_{a}(2z+2\eta(\ell+1))}{\theta_1(2z)}
\Delta(z+\eta,x)-\frac{\theta_{a}(-2z+2\eta(\ell+1))}{\theta_1(2z)}
\Delta(z-\eta,x)  \Big)\Phi(x)dx,
\end{eqnarray*}
where $a=1,2,3,4$.
To get consistent equations for the kernel $\Delta(z,x)$ we impose the
constraint that the integrations $\int_{\alpha\pm\eta}^{\beta\pm\eta}dx$
give the same result as the integration $\int_\alpha^\beta dx$. For Im$(\eta)>0$
or Im$(\eta)<0$ this is so if the  integrals
over intervals $[\alpha,\alpha+\eta]$ and $[\beta,\beta+\eta]$
as well as over $[\alpha,\alpha-\eta]$ and $[\beta,\beta-\eta]$ coincide
and if the contour integrals over the parallelograms
$[\alpha,\beta,\beta+\eta,\alpha+\eta]$ and
$[\alpha,\beta,\beta-\eta,\alpha-\eta]$ are equal to zero.
The former constraint is satisfied if the integrands are periodic
with the period $\beta-\alpha$ and the latter condition is fulfilled
if the integrands are analytical in the respective parallelograms and
have no simple poles there (or the sum of their residues is equal to zero).
The case Im$(\eta)=0$ will be considered separately.

Supposing that these demands are satisfied, which will be analyzed a posteriori,
we obtain the equation
\begin{eqnarray}\nonumber &&
\int_{\alpha}^{\beta}\Big(\frac{\theta_{a}(2x-s)}{\theta_1(2x-2\eta)}
\Delta(z,x-\eta)-\frac{\theta_{a}(-2x-s)}{\theta_1(2x+2\eta)}
\Delta(z,x+\eta)
\\ &&\makebox[2em]{}
-\frac{\theta_{a}(2z+s)}{\theta_1(2z)}
\Delta(z+\eta,x)+\frac{\theta_{a}(-2z+s)}{\theta_1(2z)}
\Delta(z-\eta,x)  \Big)\Phi(x)dx=0,
\label{auxil}\end{eqnarray}
where $s=2\eta(\ell+1)$.
Since this integral should vanish for arbitrary admissible
function $\Phi(x)$, its integrand should vanish on its own. Therefore
the following system of four finite-difference equations should be true
\begin{eqnarray}\nonumber &&
\frac{\theta_{a}(2x-s)}{\theta_1(2x-2\eta)}
\Delta(z,x-\eta)-\frac{\theta_{a}(-2x-s)}{\theta_1(2x+2\eta)}
\Delta(z,x+\eta)
\\ && \makebox[2em]{}
=\frac{\theta_{a}(2z+s)}{\theta_1(2z)}
\Delta(z+\eta,x)-\frac{\theta_{a}(-2z+s)}{\theta_1(2z)}
\Delta(z-\eta,x).
\label{W-sys}\end{eqnarray}

Let us multiply the equation with $a=3$ by $\theta_4(2z+s)$
and the equation with $a=4$ by $\theta_3(2z+s)$ and subtract them from each other.
Using theta-functions identity \eqref{11} from the Appendix we obtain
the equality
\begin{eqnarray}\nonumber &&
\frac{\theta_1(2z+2x,2z-2x+2s|2\tau)}{\theta_1(2x-2\eta|\tau)}
\Delta(z,x-\eta)-\frac{\theta_1(2z-2x,2z+2x+2s|2\tau)}{\theta_1(2x+2\eta|\tau)}
\Delta(z,x+\eta)
\\&& \makebox[6em]{}
=-\frac{\theta_1(2s,4z|2\tau)}{\theta_1(2z|\tau)}\Delta(z-\eta,x).
\label{eqx1}\end{eqnarray}

Similarly, we multiply now the equation with $a=1$ by $\theta_2(2z+s)$
and the equation with $a=2$ by $\theta_1(2z+s)$ and subtract them from each other.
Applying theta-functions identity \eqref{14} from the Appendix we obtain
\begin{eqnarray}\nonumber &&
\frac{\theta_4(2z+2x|2\tau)\theta_1(2z-2x+2s|2\tau)}{\theta_1(2x-2\eta|\tau)}
\Delta(z,x-\eta)-\frac{\theta_4(2z-2x|2\tau)\theta_1(2z+2x+2s|2\tau)}
{\theta_1(2x+2\eta|\tau)}\Delta(z,x+\eta)
\\&& \makebox[6em]{}
=-\frac{\theta_4(2s|2\tau)\theta_1(4z|2\tau)}{\theta_1(2z|\tau)}\Delta(z-\eta,x).
\label{eqx2}\end{eqnarray}
Exclude the term $\Delta(z-\eta,x)$ from the obtained equations. Namely, divide
equation \eqref{eqx1} by $\theta_1(2s|2\tau)$, equation \eqref{eqx2}
by $\theta_4(2s|2\tau)$ and subtract them. This yields the equation
$$
\Big(\theta_1(2z+2x|2\tau)\theta_4(2s|2\tau)-\theta_4(2z+2x|2\tau)\theta_1(2s|2\tau)\Big)
\frac{\theta_1(2z-2x+2s|2\tau)}{\theta_1(2x-2\eta|\tau)}\Delta(z,x-\eta)
$$
$$
=\Big(\theta_1(2z-2x|2\tau)\theta_4(2s|2\tau)-\theta_4(2z-2x|2\tau)\theta_1(2s|2\tau)\Big)
\frac{\theta_1(2z+2x+2s|2\tau)}{\theta_1(2x+2\eta|\tau)}\Delta(z,x+\eta).
$$
Applying theta-function identities \eqref{12} and \eqref{thetadup}
from the Appendix, we come to the following compact linear first order
finite difference equation
\begin{equation}
\frac{\Delta(z,x+\eta)}{\Delta(z,x-\eta)}=\frac{\theta_1(2x+2\eta,z+x-s,z-x+s)}
{\theta_1(2x-2\eta,z-x-s,z+x+s)}.
\label{eqxfin}\end{equation}

Exclude now from system \eqref{W-sys} the terms $\Delta(z,x\pm\eta)$.
Repeating similar steps as before, we exclude first $\Delta(z,x-\eta)$,
then $\Delta(z,x+\eta)$, and come to analogous equations with $\theta_1(2x\pm2\eta)$
replaced by $\theta_1(2z)$ and slightly changed arguments in other theta
functions. The final result is
\begin{equation}
\frac{\Delta(z+\eta,x)}{\Delta(z-\eta,x)}=\frac{\theta_1(x-z+s,x+z-s)}
{\theta_1(x+z+s,x-z-s)}.
\label{eqzfin}\end{equation}

Excluding the $\Delta(z,x-\eta)$-term from equations \eqref{eqx1} and \eqref{eqx2}
we obtain the third needed equation
\begin{equation}
\frac{\Delta(z,x+\eta)}{\Delta(z-\eta,x)}=\frac{\theta_1(2x+2\eta,z+x-s)}
{\theta_1(2x,z+x+s)}.
\label{eqxz}\end{equation}
The same equation emerges (in a different way) if one excludes $\Delta(z+\eta,x)$-term
from the pair of equations leading to relation \eqref{eqzfin}.

After these considerations it is not difficult to find the general solution
of equations \eqref{eqxfin}, \eqref{eqzfin}, and \eqref{eqxz} valid for Im$(\eta)>0$:
\begin{equation}\label{Delta}
\Delta(z,x) = \mathrm{e}^{\frac{\pi \textup{i}}{\eta} (x^2-z^2)}
\,\frac{\Gamma(\pm z \pm x +\eta-s|\tau,2\eta)}
{\Gamma(\pm 2x|\tau,2\eta)}\varphi(z,x),
\end{equation}
where
$\varphi(z+2\eta,x)=\varphi(z,x+2\eta)=\varphi(z+\eta,x+\eta)=\varphi(z,x)$
is an arbitrary periodic function. A way to fix this functional freedom
by imposing an additional symmetry will be considered in the section
on elliptic modular double.

The key ingredient in expression \eqref{Delta} described by the
ratio of elliptic gamma functions is periodic in $x$ with the
period 1. Similarly, $\theta_a(2x)$-functions entering
the intertwining relation \eqref{auxil} have this period.
Therefore we have to demand that the rest of the
integrands in \eqref{auxil} be periodic with the same period 1.
This condition forces the length of the integration interval to be equal
to 1, $\beta-\alpha=1$. This periodicity allows us to fix the point $\alpha$
arbitrarily, and we take $\alpha=0$ and $\beta=1$.
After the shifts $x\to x\pm\eta$ the factor $e^{\pi\textup{i}x^2/\eta}$
gets multiplied by the function $e^{\pm2\pi\textup{i}x+\pi\textup{i}\eta}$
which has period 1. Therefore we have to demand that the products
$e^{\pi\textup{i}x^2/\eta}\varphi(z,x\pm\eta)\Phi(x)$ are periodic
functions of $x$ with period 1.

Now we note that the ratio of elliptic gamma functions in $\Delta(z,x)$
is invariant under the transformations $z\to z+1$ and $z\to -z$. If we require that
similar properties are obeyed by the function $\varphi(z,x)$, then
the functions $\Psi(z)=e^{\pi\textup{i}z^2/\eta}[W\Phi](z)$
become periodic $\Psi(z+1)=\Psi(z)$ and even $\Psi(-z)=\Psi(z)$.
Therefore it is natural to demand that the original functions
$\Phi(x)$ belong to the same class of functions, i.e. that
$e^{\pi\textup{i}x^2/\eta}\Phi(x)$ are invariant under the
transformations $x\to x+1$ and $x\to -x$.
This assumes that $\varphi(z,x+1)=\varphi(z,-x)=\varphi(z,x)$,
which resolves all the periodicity restrictions and forces
$\varphi(z,x)$ to be an even elliptic function of $z$ and $x$
with periods $2\eta$ and $1$.
This fixes the space of functions $\Phi(x)$ where our operator $W$
can work as an intertwining operator of the Sklyanin algebra generators.
Note that its structure does not contradict with the property that for $2\ell\in \Z_{\geq 0}$
the Sklyanin algebra has finite-dimensional representations in the space of
even theta functions of order $4\ell$, since $e^{\pi\textup{i}x^2/\eta}$
is an even theta function of order zero.

It remains to consider singularities of the integrand in \eqref{auxil}.
The reflection equation for elliptic gamma function \eqref{refl} shows
that the product $\theta_1(2x|\tau)\Gamma(2x|\tau,2\eta)\Gamma(-2x|\tau,2\eta)$
has not zeros, i.e. the poles at $x=\eta$ or $x=-\eta$ in the integrands
of relation \eqref{auxil} are spurious.
The divisor structure of the elliptic gamma function shows that if
$$
e^{2\pi\textup{i}(\pm x\pm z+\eta-s+\tau j+ 2\eta k)}\neq 1,
\quad j,k=0,1,2,\ldots,
$$
for any choice of signs when $x$ varies in the rectangle
$[-\text{Im}(\eta),1-\text{Im}(\eta),1+\text{Im}(\eta),\text{Im}(\eta)]$,
then no poles enter the needed domain. In the multiplicative notation
$X=e^{2\pi\textup{i}x}, $ $Z=e^{2\pi\textup{i}z},$ $p=e^{2\pi\textup{i}\tau}$,
$q=e^{4\pi\textup{i}\eta}$, one has the constraint $|XZ^{\pm1}tp^jq^k|\neq 1$,
$t=e^{2\pi\textup{i}(\eta-s)}=e^{-2\pi\textup{i}\eta(2\ell+1)}
=e^{2\pi\textup{i}(u_2-u_1)}$, when $|q|^{1/2}\leq |X|\leq |q|^{-1/2}$.
Evidently, for such values of $X$ the annuli
$|Xq^j|, j=0,1,\ldots,$ cover the whole disk of radius $|q|^{-1/2}$.
Therefore we escape poles, if
\begin{equation}
|Z^{\pm1}t|<|q|^{1/2} \quad \text{or}
\quad \text{Im}(u_1-u_2\pm z)<\text{Im}(\eta).
\label{Zconstr}\end{equation}
If $z$ is a real number, i.e. $|Z|=1$, then we come to the constraint
$|t|<|q|^{1/2}$ or $\text{Im}(\eta(\ell+1))<0$. For real $\ell$ this
means that $\ell<-1$, since Im$(\eta)>0$. The finite-dimensional realizations
of the $\mathrm{R}$-matrix emerging for half-integer spins,
$2\ell\in\Z_{\geq0}$, require thus a special treatment.

Demanding that our basic space functions $\Phi(x)$ and the periodic factors
$\varphi(z,x)$ have no simple poles in the domain of interest, we satisfy
completely the conditions guaranteeing validity of equations \eqref{W-sys}.
We stress that the function obtained after action of our operator
$\int_0^1\Delta(z,x)\Phi(x)dx$ satisfies the demands imposed on $\Phi(x)$.
Indeed, since Im$(x)=0$
the function of interest is analytical when $z$ varies in the rectangle
$[-\text{Im}(\eta),1-\text{Im}(\eta),1+\text{Im}(\eta),\text{Im}(\eta)]$.

Presence of the exponential $e^{\pi\textup{i}x^2/\eta}$ in the
definition of base space functions is annoying and we wish to pass to a more
natural setting. In order to do that we conjugate all our operators by this
exponential and define new realization of the Sklyanin algebra generators
\begin{eqnarray}\nonumber && \makebox[-2em]{}
\mathbf{S}^a_{mod}= e^{\pi\textup{i}z^2/\eta}\mathbf{S}^a
e^{-\pi\textup{i}z^2/\eta} = e^{-\pi\textup{i} \eta}
\frac{(\textup{i})^{\delta_{a,2}}
\theta_{a+1}(\eta)}{\theta_1(2 z) } \Bigl[\,\theta_{a+1} \left(2
z-g+\eta\right)\cdot e^{-2\pi\textup{i}z}\cdot \mathrm{e}^{\eta \partial_z}
\\ && \makebox[10em]{}
- \theta_{a+1}
\left(-2z-g+\eta\right)\cdot e^{2\pi\textup{i}z}\cdot  \mathrm{e}^{-\eta \partial_z}\, \Bigl],
\label{modSklgen}\end{eqnarray}
where we denoted $g=\eta (2\ell+1)$.
Analogously, we define
$$
\mathrm{W}_{mod}=e^{\pi\textup{i}z^2/\eta}\mathrm{W}e^{-\pi\textup{i}x^2/\eta},
$$
acting as
$$
[\mathrm{W}_{mod}\Psi](z)=\int_0^1\frac{\Gamma(\pm z \pm x +\eta-s|\tau,2\eta)}
{\Gamma(\pm 2x|\tau,2\eta)}\varphi(z,x)\cdot \Psi(x)dx.
$$
Then the intertwining relation
$\mathrm{W}_{mod} \cdot\mathbf{S}^a_{mod}(\ell)
= \mathbf{S}^a_{mod}(-1-\ell)\cdot \mathrm{W}_{mod}$
is true provided our operators act in the space of even periodic functions
$\Psi(x)=\Psi(-x)=\Psi(x+1)$ which do not have simple poles in the domain
$-\text{Im}(\eta)\leq \text{Im}(x)\leq\text{Im}(\eta)$. Let us summarize obtained results.

\begin{theorem}
Let $\mathrm{Im}(\eta)>0$ (or $|q|<1$). Denote as $V$ the space of functions
of two complex variables $\Psi(z_1,z_2)$ which are even and periodic in each variable
with the period 1 and which do not have simple poles in the domains
$-\text{Im}(\eta)\leq \text{Im}(z_1), \text{Im}(z_2)\leq\text{Im}(\eta)$. Define three operators
\begin{equation}
[\mathrm{S}_2(u_2-v_1)\Psi](z_1,z_2)=
\Gamma\left(\pm z_1\pm z_2 + u_2  - v_1 + \eta+\textstyle{\frac{\tau}{2}}\right)
\varphi_2(z_1,z_2)\cdot\Psi(z_1,z_2),
\label{S2fin}\end{equation}
where Im$(u_2-v_1+\tau/2-\eta)>0$ (or $|\sqrt{pq}e^{2\pi\textup{i}(u_2  - v_1)}|<|q|$),
\begin{equation}
[\mathrm{S}_1(u_1-u_2)\Psi](z_1,z_2)=\frac{\kappa }{\Gamma(2u_2-2u_1)}
\int_0^1\frac{\Gamma(\pm z_1 \pm x +u_2-u_1)}
{\Gamma(\pm 2x)}\varphi_1(z_1,x)\cdot \Psi(x,z_2)dx,
\label{S1fin}\end{equation}
where Im$(u_2-u_1\pm z_1-\eta)>0$
(or $|e^{2\pi\textup{i}(u_2  - u_1\pm z_1)}|<|q|^{1/2}$),
$$
\kappa = \frac{(q;q)_\infty\,(p;p)_\infty}{2}, \qquad
(t;q)_\infty:=\prod_{j=0}^\infty(1-tq^j),
$$
and
\begin{equation}
[ \mathrm{S}_3(v_1-v_2)\Psi](z_1,z_2)=\frac{\kappa }{\Gamma(2v_2-2v_1)}
\int_0^1\frac{\Gamma(\pm z_2 \pm x +v_2-v_1)}
{\Gamma(\pm 2x)}\varphi_3(z_2,x)\cdot \Psi(z_1,x)dx,
\label{S3fin}\end{equation}
where Im$(v_2-v_1\pm z_2-\eta)>0$
(or $|e^{2\pi\textup{i}(v_2  - v_1\pm z_2)}|<|q|^{1/2}$).
Here $\varphi_k(z,x)$ are arbitrary even elliptic functions of $z$ and $x$
with periods $1$ and $2\eta$ satisfying additional constraints
$$
\varphi_k(z+\eta,x+\eta)=\varphi_k(z,x), \quad k=1,2,3,
$$
and not having simple poles in the domains
$-\text{Im}(\eta)\leq \text{Im}(z), \text{Im}(x)\leq\text{Im}(\eta)$.

Then the operators $\mathrm{S}_k(\bf{u})$ map the space $V$ onto itself and satisfy
the defining intertwining relations \eqref{RLL13}, \eqref{RLL2}, and \eqref{RLL4},
provided in the corresponding $\mathrm{L}$-operator \eqref{L_op} one uses the Sklyanin
algebra generators in the form \eqref{modSklgen}.
\end{theorem}

Here we do not use the notation $\mathrm{S}_{k, mod}$ for brevity
assuming that there will be no confusion in the following
which particular form of the intertwining operators is used.
The functions $\varphi_k$ may depend on the parameters $u_1,u_2,v_1,v_2$
in arbitrary way. If they would not depend on the coordinates $z$ and $x$,
then we can drop $\varphi_k$ completely, since they become irrelevant for
solutions of YBE (without consideration of the unitarity condition).
Then the operators $\mathrm{S}_{1,3}$ do not depend on the spectral parameter
since they involve only the differences $u_1-u_2$ or $v_1-v_2$.
Similarly, $\mathrm{S}_2$ will depend only on the difference $u_2-v_1$.
The above choice of the normalization constant $\kappa $, as well as of the elliptic gamma
function prefactors  in $\mathrm{S}_{1,3}$ are dictated by the elliptic beta integral
\cite{spi:umn}, as described explicitly in the next section for a
special choice $\varphi_k(z,x)=1$.

Denote $y_{1,2}=e^{2\pi\textup{i} z_{1,2}}$. Fourier series expansions
for our basic functions $\Psi(z_1,z_2)$ show that our space is equivalent
to the space of meromorphic functions of $y_1, y_2\in\C^*$ satisfying
the constraints $f(y_1^{-1},y_2)=f(y_1,y_2^{-1})=f(y_1,y_2)$.
Therefore we can pass in the definition of $\mathrm{S}_{1,3}$-operators
from real integrals over $[0,1]$ to contour integrals
over the unit circle $\mathbb{T}$ of positive orientation
\begin{equation}
[\mathrm{S}_1(u_1-u_2)f](y_1,y_2)=\kappa
\int_{\mathbb{T}}\frac{\Gamma(ty_1^{\pm1}y^{\pm1};p,q)}
{\Gamma(t^2,y^{\pm2};p,q)}\varphi_1(y_1,y)\cdot f(y,y_2)\frac{dy}{2\pi\textup{i}y},
\label{S1_mult}\end{equation}
where $t=e^{2\pi\textup{i}(u_2-u_1)}$, $|ty_1^{\pm1}|<|q|$, $\varphi_1(qy_1,y)
=\varphi_1(y_1,qy)=\varphi_1(q^{1/2}y_1,q^{1/2}y)=\varphi_1(y_1,y)$,
and $\Gamma(t;p,q)$
is the elliptic gamma function in multiplicative notation \eqref{egamma_m}.
Evidently,  sequential actions of the $\mathrm{S}_k$-operators create
multiple contour integral operators. Deforming the integration
contours one can relax the constraints on parameters
and define resulting functions
by analytical continuation in parameters.

\section{Coxeter relations and the elliptic beta integral}
\label{Coxeter} \setcounter{equation}{0}

In this section we prove that the derived operators
$\mathrm{S}_i$ obey  relations~(\ref{def1}),~(\ref{def2}), and~(\ref{def3})
generating the permutation group $\mathfrak{S}_4$ for a special choice of
the periodic factors
$$
\varphi_k(z,x)=1, \qquad k=1,2,3.
$$
Under these conditions operators \eqref{S2fin}-\eqref{S3fin} depend on the differences of parameters.
Therefore Coxeter relations for them can be represented in a simpler form:
 \be \label{121}\mathrm{S}_k(a)\,\mathrm{S}_k(-a)= \II\,
;\qquad \mathrm{S}_1(a)\,\mathrm{S}_2(a+b)\,\mathrm{S}_1(b)=
\mathrm{S}_2(b)\,\mathrm{S}_1(a+b)\,\mathrm{S}_2(a)\,, \ee
\be\label{232} \mathrm{S}_1(a)\,\mathrm{S}_3(b)=
\mathrm{S}_3(b)\,\mathrm{S}_1(a)\, ;\qquad
\mathrm{S}_2(a)\,\mathrm{S}_3(a+b)\,\mathrm{S}_2(b)=
\mathrm{S}_3(b)\,\mathrm{S}_2(a+b)\,\mathrm{S}_3(a).
\ee
In the theory of quantum integrable systems the cubic relations
are known as the star-triangle relations \cite{DM}. There are
two evident equalities. The operator $\mathrm{S}_3(a)$ differs
from the operator $\mathrm{S}_1(a)$ only by the change of variable
$z_1\to z_2$ and therefore these operators commute,
$\mathrm{S}_1(a)\,\mathrm{S}_3(b) = \mathrm{S}_3(b)\,\mathrm{S}_1(a).$
If one would stick to the generating relations in the form $(s_is_{i\pm 1})^3=1$,
then the cubic relation in \eqref{121} is replaced by
$$
\mathrm{S}_1(-b)\,\mathrm{S}_2(-a-b)\,\mathrm{S}_1(-a)\,
\mathrm{S}_2(b)\,\mathrm{S}_1(a+b)\,\mathrm{S}_2(a)\,=\II
$$
and a similar replacement holds for \eqref{232}.

The equality $\mathrm{S}_2(a)\,\mathrm{S}_2(-a) = \II$ is easily verified
with the help of reflection formula for elliptic gamma function
(we set $\varphi_2=1$), since the operator $\mathrm{S}_2$ reduces to
the multiplication by a given function. If we would
have an arbitrary periodic function $\varphi_2(z_1,z_2;a)$
in the definition of $\mathrm{S}_2(a)$ (in fact, $\varphi_2$ may depend
on $u_2$ and $v_1$ separately, not only on their difference),
then this condition does not fix $\varphi_2$, there are many nontrivial
elliptic functions of $z_1, z_2$ obeying the constraint
$\varphi_2(z_1,z_2;a)\varphi_2(z_1,z_2;-a)=1$. Only if $\varphi_2(z_1,z_2;a)$
does not depend on $z_1,z_2$, this condition fixes $S_2$
up to a constant multiplier $\varphi(a)$ satisfying the
constraint $\varphi_2(a)\varphi_2(-a)=1$.

Let us  show that the remaining nontrivial
Coxeter relations follow from the elliptic beta integral
evaluation formula \cite{spi:umn}. Denote
$$
\left[\mathrm{S}_2(a)\,\Psi\right](z_1,z_2) =
\mathrm{D}_a(z_1,z_2)\cdot \Psi(z_1,z_2) \ \ ;\ \
\mathrm{D}_a(z_1,z_2) =
\Gamma(\pm z_1\pm z_2 + a + \eta+\textstyle{\frac{\tau}{2}})\,,
$$
and
\begin{eqnarray*} &&
\left[\mathrm{S}_1(b)\,\Psi\right](z_1,z_2) = \int_0^1 \mathrm{d} z\,
\mathrm{W}_b(z_1,z)\, \Psi(z,z_2) \ \ ;\ \
\mathrm{W}_b(z_1, z) = \kappa \frac{\Gamma(\pm z\pm z_1 -b)}{\Gamma(-2b,\pm 2z)}\,,
\end{eqnarray*}
where $\Gamma(a,b)=\Gamma(a|2\eta,\tau)\Gamma(b|2\eta,\tau)$.
Similarly, $\mathrm{S}_3(b)$ has the same form as $\mathrm{S}_1(b)$, but it
acts in the space of functions of $z_2$.

The first Coxeter relation
\begin{equation}\label{key2}
\mathrm{S}_1(a)\,\mathrm{S}_2(a+b)\,\mathrm{S}_1(b)=
\mathrm{S}_2(b)\,\mathrm{S}_1(a+b)\,\mathrm{S}_2(a)
\end{equation}
is equivalent to the following equation for the kernels
\begin{equation}\label{int1}
\int_0^1 \mathrm{d} z\,
\mathrm{W}_a (z_1,z)\, \mathrm{D}_{a+b}(z,z_2)\,\mathrm{W}_b (z,x)
=\mathrm{D}_b(z_1,z_2)\,\mathrm{W}_{a+b}(z_1,x)\, \mathrm{D}_a(x,z_2)\,.
\end{equation}
The second Coxeter relation
\begin{equation}\label{key22}
\mathrm{S}_3(a)\,\mathrm{S}_2(a+b)\,\mathrm{S}_3(b)=
\mathrm{S}_2(b)\,\mathrm{S}_3(a+b)\,\mathrm{S}_2(a)
\end{equation}
is equivalent to a similar equation for the kernels
\begin{equation}\label{int2}
\int_0^1 \mathrm{d} z\,
\mathrm{W}_a (z_2,z)\, \mathrm{D}_{a+b}(z_1,z)\,\mathrm{W}_b (z,x)
=\mathrm{D}_b(z_1,z_2)\,\mathrm{W}_{a+b}(z_2,x)\, \mathrm{D}_a(z_1,x)\,.
\end{equation}
This equation can be obtained from the first one after permutation $z_1\leftrightarrow z_2$
provided the function $\mathrm{D}_a(z_1,z_2)$ is symmetric,
$\mathrm{D}_a(z_1,z_2) = \mathrm{D}_a(z_2,z_1)$, which is evident in our case.
 Therefore we have to prove the first Coxeter relation only.

Let us compare the elliptic beta integral evaluation formula~\cite{spi:umn}
\begin{equation}
\kappa  \int_0^1 \frac{ \prod_{k=1}^6 \Gamma(g_k\pm z)}{\Gamma(\pm 2z)} \,\mathrm{d} z
 = \prod_{1\leq j<k\leq6} \Gamma(g_j+g_k)\,,
\label{ell_beta}\end{equation}
where
\begin{equation}
g_1+\cdots+g_6 = 2\eta+\tau \, (\text{mod}\  \Z), \qquad \mathrm{Im}(g_k),
\mathrm{Im}(\eta),\mathrm{Im}(\tau)>0,
\label{balance}\end{equation}
with relation \eqref{int2} which has the following explicit form
$$
\kappa \int_0^1 \mathrm{d} z\, \frac{\Gamma(\pm z\pm z_1-a)}{\Gamma(-2a,\pm 2z)}
\cdot
\Gamma(\pm z\pm z_2+a+b+\eta+{\textstyle\frac{\tau}{2}})
\cdot
\frac{\Gamma(\pm x\pm z-b)}{\Gamma(-2b,\pm 2x)}\, =
$$
$$
= \Gamma(\pm z_1\pm z_2+b+\eta+{\textstyle\frac{\tau}{2}})\cdot
\frac{\Gamma(\pm x\pm z_1-a-b)}{\Gamma(-2a-2b,\pm 2x)}\,
\cdot\Gamma(\pm x\pm z_2+a+\eta+{\textstyle\frac{\tau}{2}})\,.
$$
After evident simplifications we arrive at the elliptic beta integral \eqref{ell_beta}
for the following choice of parameters
\begin{eqnarray*} &&
g_1= z_1-a\ , \ g_2 = -z_1-a \ ,\ g_3 = z_2+a+b+\eta+{\textstyle\frac{\tau}{2}}\ ,
\\  &&
g_4 = -z_2+a+b+\eta+{\textstyle\frac{\tau}{2}}\ , \ g_5= x-b\ , \ g_6 = -x-b\,.
\end{eqnarray*}
The constraints on values of $z_{1}, z_{2}, x, a$, and $b$ used in the construction
of intertwining operators guarantee that $|e^{2\pi\textup{i}g_{1,2,5,6}}|< |q|^{1/2}$
and $|e^{2\pi\textup{i}g_{3,4}}|< |q|$.
Since $\prod_{k=1}^6 e^{2\pi\textup{i}g_k}=pq$, we conclude that we should
have $|p|<|q|^3$.
However, by analytical continuation one can see that identity \eqref{int1} is
valid for a wider region of parameters  $\mathrm{Im}(g_k)>0$ which is symmetric in
$2\eta$ and $\tau$ (or $p$ and $q$). Thus, the cubic Coxeter relations hold true
as a consequence of the exact integration formula \eqref{ell_beta}.

The same result is obtained for the original Sklyanin
algebra generators realization \eqref{Sklyan} leading to the
exponential factors $e^{\pi\textup{i}(z^2-x^2)/\eta}$ in the
definition of $\mathrm{S}_{1,3}$-operators. Namely, all such exponentials
cancel from the integral identity.

It is not clear whether our normalization $\varphi_k=1$
(or, more generally, the demand that $\varphi_k$ do not depend on $z_1, z_2$)
is crucial for the obtained result. Cubic Coxeter relation is valid
for all $\varphi_k$-functions satisfying the constraint
$$
\varphi_1(z_1,z;a)\varphi_2(z,z_2;a+b)\varphi_1(z,x;b)
=\varphi_2(z_1,z_2;b)\varphi_1(z_1,x;a+b)\varphi_2(x,z_2;a).
$$
It is necessary to understand how strongly this relation restricts
functions $\varphi_k(z,x)$. Perhaps there are nontrivial solutions,
so that it is desirable to fix the functional freedom
in somewhat different way.

Let us take the limit $b\to -a$ in relation \eqref{key2}. It is evident that
$\mathrm{S}_2(0)=\II$, i.e. the left-hand side expression reduces to
the product $\mathrm{S}_1(a)\,\mathrm{S}_1(-a)$ and the right-hand side --
to $\mathrm{S}_2(-a)\mathrm{S}_1(0)\mathrm{S}_2(a)$. As it will be shown below
in the section on finite-dimensional reductions, one has $\mathrm{S}_1(0)=\II$
(this follows after careful resolution of the ambiguity arising from
vanishing multiplier in front of the integral and diverging
value of the integral itself).  Since $\mathrm{S}_2(-a)\mathrm{S}_2(a)=\II$,
we formally come to the remaining quadratic Coxeter relations
$$
\mathrm{S}_1(a)\,\mathrm{S}_1(-a) = \II \ ;\quad
\mathrm{S}_3(a)\,\mathrm{S}_3(-a) = \II\,.
$$
The problem is that for the moment we have defined operators $\mathrm{S}_{1,3}(a)$
only in the restricted domain of values of $a$ and $z_1$, namely,
$\text{Im}(a\pm z_1)<-\text{Im}(\eta)$.
Therefore we should give proper definition to $\mathrm{S}_{1,3}(-a)$-operators
for the inversion relations to be true.

Suppose that the integral operator $\mathrm{S}_{1}(a)$ in the form \eqref{S1_mult}
acts on a holomorphic function of $y\in\C^*$ so that all the singularities
of the resulting function are determined by the kernel of $\mathrm{S}_{1}(a)$.
From the divisor structure of elliptic gamma functions one can see that
the latter kernel has poles at
\begin{equation}
in: \; y=ty_1^{\pm1}p^jq^k, \qquad
out: \; y=t^{-1}y_1^{\pm1}p^{-j}q^{-k},\;
\label{poles}\end{equation}
where $j, k\in \Z_{\geq 0}$, with the first sequence converging to
zero $y=0$ and the second
one going to infinity. For $|tx^{\pm1}|<|q|^{1/2}$ the contour $\mathbb T$
separates these two sets of poles. Let us replace now the integration
contour $\mathbb T$ in \eqref{S1_mult} by an arbitrary contour $C$
separating poles $in$ from $out$. By Cauchy theorem no singularities
emerge after changing values of variables $t$ and $x$ as soon as the
kernel poles do not cross the integration contour. Evidently, this procedure
extends the definition of the operator $\mathrm{S}_{1}$ to arbitrary
values of $t$ and $x$ guaranteeing existence of the contour $C$.
The latter condition is violated if the poles from $in$ and $out$ sets
pinch the integration contour, which can happen if
\begin{equation}
ty_1^{\pm1}p^{j_1}q^{k_1}= (ty_1^{\pm1})^{-1}p^{-j_2}q^{-k_2}
\quad \text{or} \quad
ty_1^{\pm1}p^{j_1}q^{k_1}= (ty_1^{\mp 1})^{-1}p^{-j_2}q^{-k_2}.
\label{t1}\end{equation}
The $\mathrm{S}_{1}$-operator kernel contains also the multiplier
$1/\Gamma(t^2;p,q)$ vanishing for $t^2\to p^{-j}q^{-k}$ and diverging
for $t^2\to p^{j+1}q^{k+1}$ with $j, k\in \Z_{\geq 0}$.
As a result, we come to the constraints
\begin{equation}
t^2\neq p^{-j}q^{-k}, \; p^{j+1}q^{k+1}, \quad j, k\in \Z_{\geq 0},
\label{t2}\end{equation}
and $y_1^2\neq t^2p^jq^k,\; t^{-2}p^{-j}q^{-k}$ for arbitrary fixed $t$
and $j, k\in \Z_{\geq 0}$. Thus we have defined the action of operator
$[\mathrm{S}_{1}(a)\Phi](z)$ for arbitrary generic values of $z$
and  $a\neq \eta j+ \tau k/2$ or $a\neq 1/2 +\eta j+ \tau k/2$,
$j,k\in\Z,\, (j,k)\neq (1,0), (0,1)$.

For $t^2=q^{-n}p^{-m}$ with integer $n,m\geq 0$ and generic values of $q$
and $p$, which we always assume, the second set of
equalities in \eqref{t1} is satisfied if $j_1+j_2=m$ and $k_1+k_2=n$, i.e. there
are  $(n+1)\times (m+1)$ pairs of poles pinching the contour $C$.
As will be shown in Sect. 8 for these exceptional values of $t$, with the
exclusion of the points $t=\pm 1$ (i.e., $a=0,1/2$),
the operator $\mathrm{S}_{1}(a)$ has nontrivial zero modes. Analogously,
the operator $\mathrm{S}_{1}(-a)$ has zero modes for $t^2=q^{n}p^{m}$,
with integer $n,m\geq0$, $(n,m)\neq (0,0)$. As a result, the relation
$\mathrm{S}_{1}(a)\mathrm{S}_{1}(-a)=\II$ cannot be true for
$t^2=q^{n}p^{m}$, $n,m\in\Z$, $(n,m)\neq 0$.

The following rigorous inversion statement was established in
\cite{spi-war:inversions}.
\begin{theorem}\label{inversionthm}
Let $p,q,t\in\C$ such that $\max\{|p|,|q|\}<|t|^2<1$.
For fixed $w\in\mathbb{T}$ let $C_w$ denote a contour
inside the annulus $\mathbb{A}=\{z\in\C;~
|t|-\epsilon<|z|<|t|^{-1}+\epsilon\}$
for infinitesimally small but positive $\epsilon$, such that
$C_w$ has the points $wt^{-1},(wt)^{-1}$ in its interior and excludes their
reciprocals. Let $f(z)=f(z^{-1})$ be a function holomorphic on $\mathbb{A}$.
Then for $|t|<|x|<|t|^{-1}$ there holds
\begin{equation}\label{inversion}
\frac{(p;p)_\infty^2(q;q)_\infty^2}{(4\pi\textup{i})^2}
\int_{\mathbb{T}}\biggl(\int_{C_w}\frac{\Gamma(tw^{\pm1}x^{\pm1},
t^{-1}w^{\pm1}z^{\pm1};p,q)} {\Gamma(t^{\pm 2},z^{\pm 2},w^{\pm 2};p,q)}
f(z)\frac{dz}{z}\biggr) \frac{dw}{w}=f(x).
\end{equation}
\end{theorem}

Rewriting relation \eqref{inversion} as
$$
\kappa \int_{\mathbb{T}}\frac{\Gamma(tw^{\pm1}x^{\pm1};p,q)}
{\Gamma(t^{2},w^{\pm 2};p,q)} \frac{dw}{2\pi\textup{i}w}
\biggl(\kappa \int_{C_w}\frac{\Gamma(t^{-1}w^{\pm1}z^{\pm1};p,q)}
{\Gamma(t^{-2},z^{\pm 2};p,q)} f(z)\frac{dz}{2\pi\textup{i}z}\biggr)=f(x),
$$
one evidently comes to the equality $[\mathrm{S}_{1}(a)\mathrm{S}_{1}(-a)f](x)=f(x)$
for analytically continued $\mathrm{S}_{1}(a)$-operators. Here
$\mathrm{S}_{1}(a)$ is extended to the domain of parameters $|tx^{\pm1}|<1$
and $\mathrm{S}_{1}(-a)$ is a continuation of $\mathrm{S}_{1}(a)$-operator
to the domain $|tw^{\pm1}|<1/\max\{|p|^{1/2},|q|^{1/2}\}+\epsilon$.

In \eqref{inversion} it is assumed, of course, that  the integrand poles at
$z=t^{-1}w^{\pm1}p^{j}q^{k}$ for $j,k\in \Z_{\geq0}$
sit inside $C_w$ and their reciprocals -- outside of $C_w$.
The lower bound $\max\{|p|,|q|\}<|t^2|$ was imposed in order to guarantee that
only poles $t^{-1}w^{\pm1}$ and $tw^{\mp1}$ are crossed over when
one deforms the contour $C_w$ to $\mathbb T$. Under the weaker
condition $|q|,|p|<|t|$ a number of poles may enter the annulus $\mathbb{A}$,
from both sides, but they still do not cross the unit circle.
Therefore the same arguments as in \cite{spi-war:inversions} apply
and we obtain the needed inversion relation $\mathrm{S}_1(a)\mathrm{S}_1(-a)=\II$
in a wider region of parameter $a$. In particular, for $|x|=1$
and $\max\{|p|, |q|\}<|t|<|q|^{1/2}$ we can satisfy even the original
restrictions for $\mathrm{S}_1(a)$-operator parameters, $|tx^{\pm1}|<|q|^{1/2}$.
Thus, we have proved that the analytically continued operators $\mathrm{S}_{1,2,3}$
do satisfy Coxeter relations of the permutation group $\mathfrak{S}_4$.

If one substitutes into relation $\mathrm{S}_1(a)\mathrm{S}_1(-a)=\II$
the expression \eqref{S1fin} without taking care of existing restrictions
and assumes that $z_1$ and $z_2$ are real, then the straightforward
consideration yields
\begin{equation}
\int_0^1
\mathrm{W}_a (z_1,z)\, \mathrm{W}_{-a}(z,z_2)\, \mathrm{d} z
=\frac{1}{2}\,\left[\,\delta(z_1-z_2)+\delta(z_1+z_2)\,\right]\,,
\end{equation}
where $\delta(z)$ is the Dirac delta-function.
Using explicit expressions for $W_a$-functions,
one obtains the following formal integral identity
\begin{equation}\label{forminv}
\kappa \int_0^1
\frac{\Gamma(\pm z_1\pm z - a,\pm z_2\pm z +a)}{\Gamma(\pm 2z)}\mathrm{d} z
=\frac{\Gamma(\pm 2a,\pm 2z_2)}
{2\kappa }\,\left[\,\delta(z_1-z_2)+\delta(z_1+z_2)\,\right]\,,
\end{equation}
which was partially considered in~\cite{AA2008,BS,spi:conm}.

On the right-hand side of equality \eqref{forminv} one has the product
$$
\Gamma(\pm 2a)=\Gamma(t^{\pm2};p,q)=\frac{1}{(t^2;p)_\infty(pt^{-2};p)_\infty
(t^{-2};q)_\infty(qt^{2};q)_\infty},
$$
which diverges for $t^2=p^{j}$ or $t^2=q^{j}$, $j\in\Z$. Therefore, for generic
values of $z_1, z_2$ one would expect problems with the inversion
relation at least for these values of the parameter $t$. However,
as mentioned above, the inversion relation remains true for $t=\pm1$
and it breaks down for  $t^2=p^jq^k,\; j,k\in\Z,\; (j,k)\neq (0,0)$.
Therefore relation \eqref{forminv} cannot be true for these values of
the $t$-parameter. It is necessary to perform a careful investigation
of the inversion relation for wider regions of parameters than it was
discussed above and understand properly when formula \eqref{forminv}
is valid.

A more detailed discussion of the intertwining operator
properties is given in the following sections.

\section{Connection with the elliptic Fourier transform}
\setcounter{equation}{0}

The Bailey chains technique is a well known tool for generating
infinite sequences of identities for plain and $q$-hypergeometric series
\cite{aar}.
It is very useful for proving the Rogers-Ramanujan identities
needed for solving $2d$ statistical mechanics models \cite{Baxter}.
An integral analogue of the Bailey chains was discovered in \cite{spi:bailey}
directly at the elliptic level
using a new universal integral transform for functions depending on
one parameter $t$. As shown in \cite{spi-war:inversions} the inverse of this
transform is substantially equivalent to the reflection of the parameter
$t\to t^{-1}$ which resembles the Fourier transform. Let us describe the key
ingredients of this elliptic Fourier transformation technique.

For a given function $\alpha(z,t)$ analytic near the unit circle $z\in\mathbb{T}$
define the integral transformation
\begin{equation}
\beta(w,t)=M(t)_{wz}\alpha(z,t):=\frac{(p;p)_\infty(q;q)_\infty}{4\pi\textup{i}}\int_\mathbb{T}
\frac{\Gamma(tw^{\pm1}z^{\pm1};p,q)}
{\Gamma(t^2,z^{\pm2};p,q)}\alpha(z,t)\frac{dz}{z},
\label{EFT}\end{equation}
where $|tw|, |t/w|<1$.
In \cite{spi:bailey} the functions related in this way were said to form
an integral elliptic Bailey pair with respect to the parameter $t$
(this particular renormalized form of the definition was presented
in \cite{spi-war:inversions}).

An integral analogue of the Bailey lemma provides a method to
generate an infinite sequence of Bailey pairs from a given germ pair.
Namely, suppose that $\alpha(z,t)$ and $\beta(w,t)$ form an integral elliptic Bailey pair
with respect to the parameter $t$. Then for $|s|,|t|<1, |\sqrt{pq}y^{\pm1}|<|st|$
the new functions
\begin{eqnarray} &&
\alpha'(w,st)=D(s;y,w)\alpha(w,t),\quad D(s;y,w)
=\Gamma(\sqrt{pq}s^{-1}y^{\pm1}w^{\pm1};p,q),
\label{D-function}  \\  &&
\beta'(w,st)=D(t^{-1};y,w) M(s)_{wx}D(st;y,x)\beta(x,t),
\label{BL}\end{eqnarray}
where $w\in \T$, form an integral elliptic Bailey pair with
respect to the parameter $st$.

The proof of this statement is easy. Indeed, the demand
$$
\beta'(w,st)=M(st)_{wz}\alpha'(z,st)
$$
is equivalent to the equality
$$
D(t^{-1};y,w) M(s)_{wx}D(st;y,x)M(t)_{xz}\alpha(z,t)=M(st)_{wz}D(s;y,z)\alpha(z,t).
$$
Since $D(t^{-1};y,w)D(t;y,w)=1$ due to the reflection equation
for the elliptic gamma function, we can rewrite it as an operator identity
\begin{equation}
M(s)_{wx}D(st;y,x)M(t)_{xz}=D(t;y,w)M(st)_{wz}D(s;y,z).
\label{STR}\end{equation}
Substitute explicit expressions for $M$ and $D$-operators
and change the order of integrations in the left-hand side expression. Then
one can check that relation \eqref{STR} is equivalent to the elliptic beta integral
evaluation formula with correct restrictions on the parameters,
which proves the statement.

Relation \eqref{STR} is nothing else that the operator
form of the star-triangle relation \cite{DM}.
Equality \eqref{STR} was explicitly presented in \cite{spi:essays} in the matrix
form as relation (6.5) (which is connected with some finite dimensional
reduction of the intertwining operators).

Comparing the operator $M(t)$ with the intertwining operator of the previous
section in multiplicative form \eqref{S1_mult}, we see that
\begin{equation}
M(t)_{wy}f(y)=[\mathrm{S}_1(u_1-u_2)f](w)
\label{inter'}\end{equation}
for $t=e^{2\pi\textup{i}(u_2-u_1)}$ and $\varphi_1=1$.
Also evidently
$$
\mathrm{S}_2(a)= D(e^{-2\pi\textup{i} a};e^{2\pi\textup{i}z_1},e^{2\pi\textup{i}z_2}),
$$
provided in \eqref{D-function} one has $+\sqrt{pq}$
in the elliptic gamma function arguments, since $e^{2\pi\textup{i}(\eta+\tau/2)}=+\sqrt{pq}$.
Another square root sign choice in \eqref{D-function} yields
the $\mathrm{S}_2$-operator corresponding to the Sklyanin algebra
generators differing from \eqref{Sklyan} by the addition of $1/2$
to the arguments of theta functions depending on the spin $\ell$
(see below).

The operators entering the integral Bailey
lemma $D$ and $M$ define thus elementary transposition operators
satisfying inversion relations $\mathrm{S}_2(a)\mathrm{S}_2(-a)=\II$,
$\mathrm{S}_1(a)\mathrm{S}_1(-a)=\II$
and the basic relation \eqref{STR} is equivalent to the Coxeter relation
\begin{equation}
\mathrm{S}_1(a)\mathrm{S}_2(a+b)\mathrm{S}_1(b)
=\mathrm{S}_2(b)\mathrm{S}_1(a+b)\mathrm{S}_2(a).
\label{Cox=STR}\end{equation}
Take this operator identity (or \eqref{STR}),
use the additive notation for $\mathrm{S}_k$-operators
\eqref{S2fin}--\eqref{S3fin}  (with $\varphi_k=1$) and act by it
onto the Dirac delta-function $(\delta(x-z)+\delta(x+z))/2$. Then one obtains
the equality \eqref{int1} which can be written as
\begin{eqnarray}\nonumber
&& \int_{0}^{1} \rho(u)\mathrm{D}_{\xi-a}(x,u)\mathrm{D}_{a+b}(y,u)
\mathrm{D}_{\xi-b}(w,u)du
\\ && \makebox[2em]{}
=\chi(a, b)\mathrm{D}_b(x,y)\mathrm{D}_{\xi-a-b}(x,w)\mathrm{D}_{a}(y,w),
\label{astr}\end{eqnarray}
where
\begin{equation}
\mathrm{D}_{a}(x,u)=\mathrm{S}_2(a)= D(e^{-2\pi \textup{i}a}; e^{ 2\pi\textup{i}  x},
e^{2\pi \textup{i}  u})=\Gamma(e^{ 2\pi\textup{i}(a-\xi\pm x \pm u)};p,q)
\label{weight}\end{equation}
and
\begin{eqnarray*}
&& e^{-4\pi\textup{i}\xi}:=pq,\quad
\rho(u)=\frac{ (p;p)_\infty(q;q)_\infty}{2\Gamma(e^{\pm 4\pi\textup{i}  u};p,q)},
\quad \chi(a,b)= \Gamma(e^{-4\pi\textup{i}a},e^{-4\pi\textup{i}b},
e^{4\pi\textup{i}(a+b-\xi)};p,q).
 \end{eqnarray*}
Equality \eqref{astr} represents a functional form of the star-triangle relation
considered in \cite{BS}, with
$\mathrm{D}_a(x,u)$ being the Boltzmann weight for edges connecting spins $x$ and $u$
sitting in the neighboring vertices of a lattice, $\rho(u)$ is related
to the self-energy for spins, $\xi$ is the crossing parameter. In this
picture the integration means a computation of the partition function
for an elementary star-shaped cell with contributions coming from all possible
values of the continuous ``spin" sitting in the central vertex.
So, we may conclude that the Bailey lemma
established in \cite{spi:bailey} is equivalent to the star-triangle relation
for particular elliptic hypergeometric Boltzmann weights.

Defining $\mathrm{S}_3(a)$ as an elliptic Fourier transformation operator for
the variable $z_2$ we obtain another generator of the group $\mathfrak{S}_4$.
In this way, the algebraic relations for the Bailey lemma ingredients appear to be
equivalent to the Coxeter relations for permutation group generators.
This fact was not understood in the original paper \cite{spi:bailey},
however the connection of the integral Bailey transformation to the
star-triangle relation was briefly remarked in \cite{spi:conm}.

\section{Uniqueness of solutions and the elliptic modular double}
\setcounter{equation}{0}

The concept of elliptic modular double introduced in \cite{AA2008}
allows us to fix the functional freedom in the definition of the
$\mathfrak{S}_4$-permutation group generators.

The general infinite-dimensional space solution of the Yang-Baxter e
quation we obtain
$$
\mathbb{R}_{12}(\mathbf{u}) =\mathbb{P}_{12}
\mathrm{S}_2(u_1-v_2)\,
\mathrm{S}_1(u_1-v_1)\,\mathrm{S}_3(u_2-v_2)\,
\mathrm{S}_2(u_2-v_1)
$$
is symmetric in the parameters $p$ and $q$, since we assume that
$u_{1,2}$ and $v_{1,2}$ are independent variables (or, equivalently,
$u$ and $g=\eta(2\ell+1)$ should be considered as independent of $\eta$ and $\tau$).
Using the multiplicative notation and the modified forms
of the Sklyanin algebra generators and intertwining operators,
this R-operator has the following explicit action on the functions
of two variables
\begin{eqnarray}\label{Rexplicit} &&
[\mathbb{R}_{12}(\mathbf{u})f](z_1,z_2)
=\frac{(p;p)_\infty^2(q;q)_\infty^2}{(4\pi \textup{i})^2}
\Gamma(\sqrt{pq}z_1^{\pm1}z_2^{\pm 1}e^{2\pi\textup{i}(u_1-v_2)};p,q)
\\ && \makebox[2em]{} \times
\int_{\mathbb{T}^2}
\frac{
\Gamma(e^{2\pi\textup{i}(v_1-u_1)}z_2^{\pm1}x^{\pm1},
e^{2\pi\textup{i}(v_2-u_2)}z_1^{\pm1}y^{\pm1},
\sqrt{pq}e^{2\pi\textup{i}(u_2-v_1)}x^{\pm1}y^{\pm1};p,q)}
{\Gamma(e^{4\pi\textup{i}(v_1-u_1)},e^{4\pi\textup{i}(v_2-u_2)},
x^{\pm2},y^{\pm2};p,q)}
f(x,y)\frac{dx}{x}\frac{dy}{y}.
\nonumber\end{eqnarray}

Because of the symmetry in $p$ and $q$, we have not one but two
$\mathrm{RLL}$-intertwining relations. The second one is obtained from \eqref{RLL}
simply by permuting $2\eta$ and $\tau$ (or $p$ and $q$):
\begin{equation}\label{RLL_doub}
\mathbb{R}_{12}(u-v)\,\mathrm{L}_1^{doub}(u)
\,\sigma_3\,\mathrm{L}_2^{doub}(v)=
\mathrm{L}_2^{doub}(v)\,\sigma_3
\,\mathrm{L}_1^{doub}(u)\,\mathbb{R}_{12}(u-v)\,.
\end{equation}
Remind that after our similarity transformation the initial
$\mathrm{L}$-operator is given by the expression \eqref{L_op}
where $w_a(u)=\theta_{a+1}(u+\eta|\tau)/\theta_{a+1}(\eta|\tau)$ and Sklyanin generators
$\mathbf{S}^a$ have the form \eqref{modSklgen}. The operator $\mathrm{L}^{doub}(u)$ is
obtained from $\mathrm{L}(u)$ simply by permuting $2\eta$ and $\tau$.
This means that $\mathrm{L}^{doub}(u)$ also has the form \eqref{L_op},
but now $w_a(u)=\theta_{a+1}(u+\tau/2|2\eta)/\theta_{a+1}(\tau/2|2\eta)$
and the operators $\mathbf{S}^a_{mod}$ \eqref{modSklgen}  should be replaced by
\begin{eqnarray}\nonumber &&
\mathbf{\tilde S}^a_{mod} = e^{-\pi\textup{i}\frac{\tau}{2}}
\frac{(\textup{i})^{\delta_{a,2}}\theta_{a+1}(\tau/2|2\eta)}{\theta_1(2 z|2\eta) }
 \Bigl[\,\theta_{a+1} \left(2z-g+\frac{\tau}{2}\big|2\eta\right)\cdot e^{-2\pi\textup{i}z}\cdot
 \mathrm{e}^{\frac{1}{2}\tau \partial_z}
\\ && \makebox[10em]{}
 - \theta_{a+1}\left(-2z-g+\frac{\tau}{2}\big|2\eta\right)\cdot e^{2\pi\textup{i}z}\cdot
\mathrm{e}^{-\frac{1}{2}\tau \partial_z}\, \Bigl]\,,
\label{mod_doub2}\end{eqnarray}
where $g$-parameter is the same arbitrary parameter as in \eqref{modSklgen}.

Such a direct product of two Sklyanin algebras was introduced
in \cite{AA2008} under the name {\em an elliptic modular double}.
It represents an elliptic analogue of Faddeev's
modular double \cite{fad:mod} described as a direct product of two
$q$-analogues of $sl(2)$-algebra, $U_q(sl(2))\otimes U_{\tilde q}(sl(2))$,
with $q=e^{4\pi \textup{i}\eta}$ and $\tilde q=e^{-\pi \textup{i}/\eta}$.

The modified operators $\mathrm{S}_{1,3}(a)$ are are invariant under
permutation of $p$ and $q$ and, therefore, they satisfy in addition
to \eqref{inter1} the $\mathbf{\tilde S}^a$-operator
intertwining relations as well:
\begin{equation}
\mathrm{S}_1\cdot\mathbf{\tilde S}^a(\ell_1) =
\mathbf{\tilde S}^a(-1-\ell_1)\cdot \mathrm{S}_1\,, \qquad
\mathrm{S}_3\cdot \mathbf{\tilde S}^a(\ell_2) =
\mathbf{\tilde S}^a(-1-\ell_2)\cdot \mathrm{S}_3\, ,
\label{inter2}\end{equation}
This bonus symmetry in $p$ and $q$ in $\mathrm{S}_k(a)$
and R-operator originates from a particular choice of the arbitrary elliptic
functions $\varphi_k(z_1,z_2)$ emerging in solutions of corresponding finite-difference
equations. We can invert the logic and demand from the very beginning
existence of the elliptic modular double \eqref{RLL_doub}.
Then we can repeat the same considerations as before
and get the same solutions of the finite-difference equations,
but now the phase factors are restricted by additional periodicity requirements
\begin{equation}
\varphi_k(z_1+\tau,z_2)=\varphi_k(z_1,z_2+\tau)=\varphi_k(z_1+\tau/2,z_2+\tau/2)=\varphi(z_1,z_2).
\label{period3}\end{equation}
Then we use the well-known Jacobi theorem stating that any function with three
incommensurate periods
$$
\varphi(z+\omega_1)=\varphi(z+\omega_2) =\varphi(z+\omega_3)=\varphi(z),
\quad \sum_{k=1}^3n_k\omega_k\neq 0,\quad n_k\in\Z,
$$
must be a constant, $\varphi=const$.
Since our $\varphi_k$-functions were fixed already
to be elliptic functions with periods $1$ and $2\eta$ we conclude
that the constraints \eqref{period3} with generic values of $\tau$ and $\eta$
enforce $\varphi_k(z_1,z_2)=const$.
Thus $\varphi_k$ may depend only on the parameters $u_{1,2}, v_{1,2}$, and we have
chosen the normalization $\varphi_k=1$ consistent with the unitarity constraint.

Thus the elliptic modular double (i.e., two $\mathrm{RLL}$-relations
for a given R-operator) fixes the intertwining operators $\mathrm{S}_k$
uniquely (up to the multiplication by a constant). Note that this double algebra
generators $\mathbf{S}^a$ and $\tilde{\mathbf{S}}^a$
do not commute with each other, however they satisfy some simple
anticommutation relations \cite{AA2008}.
As described also in \cite{AA2008}, there exists a second modular double
obtained by employing a modular transformation of theta functions
and the modified elliptic gamma function \cite{S2}, which will be considered below.

\section{Reduction to a finite-dimensional case}
\setcounter{equation}{0}

An intertwining operator for the Sklyanin algebra generators with positive
integer values of $2\ell+1$ was constructed in \cite{Z1} as a finite-difference
operator of a finite order. Its formal extension to infinite order
was used in \cite{DKK} for building solutions of YBE along the
same lines as described above using the realization of
$\mathrm{S}_{1,3}(a)$. In particular, the operator $\mathrm{S}_1(a)$
had the following form (see the beginning of Sect. 5.1 in \cite{DKK}):
\begin{eqnarray} &&
\mathrm{S}_1(a) = e^{2\pi\textup{i}a\eta x+\pi\textup{i}\eta a^2}
\frac{\Gamma(2\eta x)}{\Gamma(2\eta(x+a))}
\sum_{k=0}^\infty\frac{[2k-x-a]}{[-x-a]}\prod_{j=0}^{k-1}
\frac{[j-x-a][j-a]}{[j+1-x][j+1]} e^{(a-2k)\partial_x},
\label{S1_discr}\end{eqnarray}
where $[x]=\theta_1(2\eta x|\tau)$. However, this operator is not well defined
unless the infinite series terminates. In particular, the formal action of
this operator on a meromorphic function of $x$ yields in general a diverging series.
Therefore we assume, as in \cite{Z1}, that $a=2\ell+1$ is a positive integer.

Let us denote
$$
q=e^{4\pi \textup{i}\eta},\quad p=e^{2\pi\textup{i} \tau},
\quad w=e^{-2\pi \textup{i} \eta x},
\quad t=e^{-2\pi \textup{i} \eta a},
$$
and assume that $|q|<1$.
Then, using the theta function $\theta(x;p)=(x;p)_\infty(px^{-1};p)_\infty$,
we can write
\begin{eqnarray}  \nonumber &&
\mathrm{S}_1(a)f(w)=q^{\frac{ax}{2}+\frac{a^2}{4}}
\frac{\Gamma(w^{-2};p,q)}{\Gamma(t^{-2}w^{-2};p,q)}
\\ && \makebox[4em]{}\times
\sum_{k=0}^\infty t^{-2k}\frac{\theta(t^2w^2q^{2k};p)}
{\theta(t^2w^2;p)}\prod_{j=0}^{k-1}
\frac{\theta(t^2w^2q^j,t^2q^j;p)}{\theta(w^2q^{j+1},q^{j+1};p)}
f(tq^kw),
\label{S1_short}\end{eqnarray}
where $\theta(t_1,\ldots,t_k;p)=\theta(t_1;p)\ldots\theta(t_k;p)$ and
$\Gamma(x;p,q)$ is the elliptic gamma function \eqref{egamma_m}.
Because of the series termination, operator \eqref{S1_short}
is well defined for $|q|<1$ and $|q|>1$, and even $|q|=1$, provided $q$ is
not a root of unity.

Introduce now an integral transformation operator acting in the space
of $z\to z^{-1}$ invariant functions, $f(z^{-1})=f(z)$,
$$
[B(t)f](w)=\frac{\kappa}{\Gamma(t^2;p,q)}
\int_{C}\frac{dz}{2\pi\textup{i}z} g(z,w,t)\frac{\Gamma(tw^{\pm1}z^{\pm1};p,q)}
{\Gamma(z^{\pm2};p,q)}f(z),
$$
where $g(z,w,t)$ is some fixed function to be determined from the comparison
with operator $\mathrm{S}_1(a)$ \eqref{S1_discr} and
$\kappa$ was fixed earlier.
The contour of integration $C$ is chosen in such a way that it separates
geometric progressions of the integrand poles converging to zero $z=0$ from
their $z\to 1/z$ reciprocals. In particular, if $|tw|,|t/w|<1$ and the product $g(z,w,t)f(z)$
is an analytical function near the unit circle with positive orientation $\mathbb{T}$,
then one can choose $C=\mathbb{T}$.

Now we take $|twq^k|>1$, $k=0,\ldots,N$, $|twq^{N+1}|<1,$ $|twp|<1,$ $|t/w|<1$, and
the contour of integration $C$ as a deformation of $\mathbb{T}$ containing
in its interior the poles at $z=twq^k,\ k=0,\ldots,N$, lying outside $\mathbb{T}$
and excluding the poles at $z=t^{-1}w^{-1}q^{-k}$, $k=0,\ldots,N$, which
enter $\mathbb{T}$. Now we pull $C$ to $\mathbb{T}$ and pick up the residues
at $z^{\pm1}=twq^k$, $k=0,\ldots,N$. This yields the formula
\begin{eqnarray*} &&
[B(t)f](w)=\frac{\kappa}{\Gamma(t^2;p,q)}\sum_{k=0}^N
\stackreb{\lim}{z\to twq^k}(z-twq^k)\frac{\Gamma(tw^{\pm1}z^{\pm1};p,q)}
{z\Gamma(z^{\pm2};p,q)}
\\ && \makebox[4em]{}  \times
\left(g(twq^k,w,t)+g(t^{-1}w^{-1}q^{-k},w,t)\right)f(twq^k)
\\ &&
+\frac{\kappa}{\Gamma(t^2;p,q)}
\int_{\mathbb{T}}\frac{dz}{2\pi\textup{i}z} g(z,w,t)\frac{\Gamma(tw^{\pm1}z^{\pm1};p,q)}
{\Gamma(z^{\pm2};p,q)}f(z).
\end{eqnarray*}

Using the relation
$$
\stackreb{\lim}{z\to 1}(1-z)\Gamma(z;p,q)=\frac{1}{(p;p)_\infty(q;q)_\infty},
$$
we find
$$
\stackreb{\lim}{z\to twq^k}\left(1-\frac{twq^k}{z}\right)\Gamma(twz^{-1};p,q)
=\frac{1}{(p;p)_\infty(q;q)_\infty}\frac{1}{\theta(q^{-k},\ldots,q^{-1};p)}.
$$
As a result, we obtain
\begin{eqnarray*} &&
[B(t)f](w)=\frac{\Gamma(w^{-2};p,q)}{\Gamma(t^{-2}w^{-2};p,q)}\sum_{k=0}^N
\frac{g(twq^k,w,t)+g(t^{-1}w^{-1}q^{-k},w,t)}{2}
\\ && \makebox[4em]{}  \times
t^{-4k}w^{-2k}q^{-k^2}\frac{\theta(t^2w^2q^{2k};p)}{\theta(t^2w^2;p)}
\prod_{j=0}^{k-1}\frac{\theta(t^2q^j,t^2w^2q^j;p)}
{\theta(q^{j+1},w^2q^{j+1};p)}f(twq^k)
\\ &&
+\frac{\kappa}{\Gamma(t^2;p,q)}
\int_{\mathbb{T}}\frac{dz}{2\pi\textup{i}z} g(z,w,t)\frac{\Gamma(tw^{\pm1}z^{\pm1};p,q)}
{\Gamma(z^{\pm2};p,q)}f(z).
\end{eqnarray*}

Now we take the limit $t^2\to q^{-2\ell-1}$, which means $a\to 2\ell+1\in\Z_{>0}$.
Since $1/\Gamma(t^2;p,q)=1/\Gamma(q^{-2\ell-1};p,q)=0$, and the
integral $\int_{\mathbb{T}}$ is finite, we see that the second term
disappears with the final result
\begin{eqnarray} \nonumber &&
[B(t)f](w)\Big|_{t^2=q^{-2\ell-1}}=\frac{\Gamma(w^{-2};p,q)}{\Gamma(q^{2\ell+1}w^{-2};p,q)}
\sum_{k=0}^{2\ell+1}
\frac{ g(wq^{k-\frac{2\ell+1}{2} },w,q^{-\frac{2\ell+1}{2} })+
g(w^{-1}q^{-k+\frac{2\ell+1}{2} },w,q^{-\frac{2\ell+1}{2} }) }{2}
\\ && \makebox[2em]{}  \times
q^{2k(2\ell+1)}w^{-2k}q^{-k^2}\frac{\theta(w^2q^{2k-2\ell-1};p)}
{\theta(q^{-2\ell-1}w^2;p)}
\prod_{j=0}^{k-1}\frac{\theta(q^{j-2\ell-1},w^2q^{j-2\ell-1};p)}
{\theta(q^{j+1},w^2q^{j+1};p)}f(wq^{k-\frac{2\ell+1}{2} }).
\label{Bdiscr}\end{eqnarray}
This series terminates automatically at $k=2\ell+1=N$.

We have derived this result under the following constraints for $w$-variable
$|twq^{N+1}|<1$ and $|t/w|<1$, or $|q^{-\ell-1/2}|< |w|<|q^{\ell- N-1/2}|$.
However, the derived result is independent on these restrictions since in the limit
$t^2\to q^{-2\ell-1}$ there are $2\ell+2$ pairs of poles which pinch the integration
contour and one inevitably has to pass to the residues sum of the above form
independently on $w$-values.

Equating this expression with the terminating series operator $\mathrm{S}_1(a)$,
we come to the relation
$$
\frac{1}{2}\left(g(z,w,t)+g(z^{-1},w,t)\right)t^{-4k}w^{-2k}q^{-k^2}
=q^{\frac{ax}{2}+\frac{a^2}{2}}t^{-2k},\quad z=twq^k.
$$
It has unique solution invariant under the transformation $z\to 1/z$
$$
g(z,w,t)=q^{\frac{ax}{2}+\frac{a^2}{2}}\exp \left[\frac{(\log z)^2-(\log tw)^2}{\log q}\right],
$$
as wanted. However, it is not an analytical function of $z$. But for validity of
our consideration we needed only that the product $g(z,w,t)f(z)$ be analytical.
Let us denote $z=e^{2\pi\textup{i}\eta u}$. Then we can write
$$
g(z=e^{2\pi\textup{i}\eta u},w=e^{-2\pi\textup{i}\eta x},
t=e^{-2\pi\textup{i}\eta a})=e^{\pi \textup{i}\eta(u^2-x^2)}.
$$
Thus we have to demand that $\phi(z):=e^{\pi \textup{i}\eta u^2}f(z)$
is a meromorphic function of $z$. Let us demand additionally that
$\phi(z^{-1})=\phi(z)$. Then the operator $B_{mod}(t)$ defined after
the similarity transformation
$$
B_{mod}(t)=e^{\pi \textup{i}\eta x^2}B(t)e^{-\pi \textup{i}\eta u^2}
$$
maps the space of meromorphic $A_1$-invariant functions $\phi(z)$ onto itself.
Let us replace in $B_{mod}(t)$ the variables $u\to u/\eta, x\to x/\eta$,
i.e. we pass to the parameterization $z=e^{2\pi\textup{i} u},w=e^{-2\pi\textup{i} x}.$
Then explicitly we have
$$
[B_{mod}(t)\psi](w)=\frac{(p;p)_\infty(q;q)_\infty}{2\Gamma(t^2;p,q)}
\int_{0}^1 \frac{\Gamma(te^{2\pi\textup{i} (\pm x\pm u)};p,q)}
{\Gamma(e^{\pm 4\pi\textup{i} u};p,q)}\psi(e^{2\pi\textup{i} u}) du.
$$
Evidently, for $t=e^{-2\pi\textup{i}(u_1-u_2)}$ this operator coincides with
the intertwining operator $\mathrm{S}_1(u_1-u_2)$ \eqref{S1fin} (for $\varphi_1=1$)
in such a way that
$$
[\mathrm{S}_1(b)\Psi](x)=[B_{mod}(e^{-2\pi\textup{i} b})\psi](e^{2\pi\textup{i}x}),
$$
where $\Psi(x)=\psi(e^{2\pi\textup{i}x})$ (i.e., the difference is only in the additive
or multiplicative notation). We conclude that the discrete intertwining operator
of \cite{Z1} is a special limiting case of our intertwining operator $\mathrm{S}_1(b)$.

Consider the identity
$$
\mathrm{S}_{3}(u_1-u_2) \mathrm{R}_{12}(u_1,u_2| v_1,v_2) =
\mathrm{R}_{12}(u_2,u_1| v_1,v_2)\mathrm{S}_{1}(u_1-u_2),
$$
which is easy to check using the R-matrix factorization.
Multiplying it by the permutation operator $\P_{12}$ and using
relation $\P_{12}\mathrm{S}_{3}=\mathrm{S}_{1} \P_{12}$,  we obtain the identity
$$
\mathrm{S}_{1}(u_1-u_2) \R_{12}  =\R_{12}' \mathrm{S}_{1}(u_1-u_2),
$$
where $\R_{12}=\P_{12}\mathrm{R}_{12}(u_1,u_2| v_1,v_2)$
and $\R_{12}'=\P_{12}\mathrm{R}_{12}(u_2,u_1| v_1,v_2)$.
It shows that the kernel space of $\mathrm{S}_1$-operator
is mapped onto itself by our R-matrix $\R_{12}$, i.e. zero modes of  $\mathrm{S}_1$
form an invariant space for the action of operator  \eqref{Rexplicit}.

In the same way, intertwining relations \eqref{inter1} and \eqref{inter2}
show that the kernel space of $\mathrm{S}_1$-operator forms an
invariant space for the elliptic modular double, i.e. it is mapped onto
itself by the Sklyanin algebra generators ${\mathbf S}^a_{mod}$ \eqref{modSklgen}
and  $\mathbf{\tilde S}^a_{mod}$ \eqref{mod_doub2}. The standard Sklyanin algebra
admits finite-dimensional representations in the space of theta functions of
modulus $\tau $ and, naturally, its modular partner has similar representations
in the space of theta functions of modulus $2\eta$. Therefore, there
should exist finite-dimensional representations of the Sklyanin algebra
in the space of products of theta functions of moduli $\tau$ and $2\eta$.
As shown in \cite{AA2008}, the elliptic modular double has a non-trivial
automorphism permuting these theta-function submodules.

In general, the latter factorization of Sklyanin-algebra modules is related
to the fact that sums of residues of poles of elliptic hypergeometric integrals
factorizes to the product of two elliptic hypergeometric series with permuted
modular parameters $p$ and $q$, which leads to the concept of two-index
biorthogonality relation \cite{S2}.
Let us show that the corresponding residue calculus
demonstrates existence of nontrivial finite-dimensional
kernel space for $\mathrm{S}_1(a)$-operator
for $a=\eta(2\ell_q+1)+\tau(\ell_p+1/2)$ and
$a=1/2+\eta(2\ell_q+1)+\tau(\ell_p+1/2)$, where $2\ell_q, 2\ell_p
\in \Z_{\geq 0}$ (i.e., the reduction considered above is only a special
case of the much more general finite-dimensional reduction).

For this, let us repeat our consideration with a different set of
poles taken into account. Namely, let us take the limit $t^2\to q^{-N}
p^{-M}$ with $N,M\in\Z_{\geq0}$. Now a number of poles
leave the unit circle and a number of them enter it. As indicated
during the discussion of analytical continuation of the $\mathrm{S}_1$-operator,
there will be precisely $(N+1)(M+1)$ pairs of poles pinching the
integration contour $C$ in $B_{mod}$ (we choose $g(z,w,t)=1$).
Let us pull the integration contour through one half of the poles
approaching it, say, through $z=tp^jq^k,\, j=0,\ldots, M,\, k=0,\ldots,N$,
and sum the corresponding residues. Denote also $N=2\ell_q+1$
and $M=2\ell_p+1$ with half integer $\ell_p,\ell_q\geq -1/2$ in order
to match notation with the spins of Sklyanin algebras entering the
elliptic modular double. Now, using relations
\begin{equation}
\stackreb{\lim}{z\to twq^kp^j}\left(1-\frac{twq^kp^j}{z}\right)\Gamma\Big(
\frac{tw}{z};p,q\Big)
=\frac{(-1)^{jk+j+k}q^{(j+1)k(k+1)/2}
p^{(k+1)j(j+1)/2}}{(p;p)_\infty(q;q)_\infty\theta(q,\ldots,q^k;p)\theta(p,\ldots,p^j;q)}
\label{residues}\end{equation}
and
$$
\theta(p^kz;p)=(-z)^{-k}p^{-\frac{k(k-1)}{2}}\theta(z;p),\quad k\in\Z,
$$
we find
\begin{eqnarray}\nonumber &&
[B_{mod}(t)f](w)
=\frac{\Gamma(w^{-2};p,q)}{\Gamma(t^{-2}w^{-2};p,q)}
\sum_{k=0}^{2\ell_q+1}\frac{\theta((tw)^2q^{2k};p)}
{\theta((tw)^2;p)}
\prod_{b=0}^{k-1}\frac{\theta(t^2q^b,(tw)^2q^b;p)}
{\theta(q^{b+1},w^2q^{b+1};p)}
\\ && \makebox[2em]{}  \times
\sum_{j=0}^{2\ell_p+1}
\frac{\theta((tw)^2p^{2j};q)}
{\theta((tw)^2;q)}
\prod_{a=0}^{j-1}\frac{\theta(t^2p^a,(tw)^2p^a;q)}
{\theta(p^{a+1},w^2p^{a+1};q)}
\frac{f(tq^kp^jw)}{t^{4(jk+j+k)}w^{2(j+k)}p^{2jk+j^2}q^{2jk+k^2}},
\label{Bact1}\end{eqnarray}
where we should substitute the actual value of
$t=\pm q^{-\ell_q-1/2}p^{-\ell_p-1/2}$.

Note that for $\ell_p=\ell_q=-1/2$, when $t=\pm1$ (or $a=0,1/2$),
this series contains only one term. Thus, for $t=1$ (or $a=0$)
the intertwining operator becomes the unity operator,
$B_{mod}(1)=\II$ (or $\mathrm{S}_1(0)=\II$). For $t=-1$ (or $a=1/2$)
the intertwining operator becomes the parity operator,
$$
B_{mod}(-1)=P, \qquad Pf(w)=f(-w).
$$
In the additive notation, we can write $\mathrm{S}_1(1/2)=e^{\frac{1}{2}\partial_x}$,
which is the half-period shift for the variable $w=e^{-2\pi\textup{i} x}$.
All our functions are analytical in $w$ (after removal of the exponential
factors from $\mathrm{S}_{1,3}$-operators) and $t=e^{-2\pi\textup{i} a}$.
Therefore $\mathrm{S}_1(a+1)=\mathrm{S}_1(a)$ and
 $\mathrm{S}_1(-1/2)=\mathrm{S}_1(1/2)$, so that
$$
\mathrm{S}_1^2(1/2):= \mathrm{S}_1(-1/2)\mathrm{S}_1(1/2)=P^2=e^{\partial_x},
$$
which is equivalent to the unity operator, since it is the operator of
shifting by the period 1.

If we remove the constraint $t=\pm q^{-\ell_q-1/2}p^{-\ell_p-1/2}$
in \eqref{Bact1} (which is not legitimate in our procedure since
we would need to restore an integral part in $B_{mod}$)
and set formally $\ell_q,\ell_p\to\infty$, we would obtain
the double infinite series operator which sharply
differs from the univariate infinite series operator used in \cite{DKK}.
This fact shows the principle difference of the rigorously defined
integral operator $B_{mod}(t)$ of \cite{spi:bailey}
from the formal infinite elliptic hypergeometric series realization
of the $\mathrm{S}_1$-operator of \cite{DKK}.

From \eqref{Bact1} we see that for special quantized values of $t$
the $B_{mod}(t)$-operator has an almost factorized form, the only non-factorizable
pieces being the multiplier $(pqt^2)^{-2jk}$ and the action on the
function itself $f(tq^kp^jw)$. Suppose now that we work
in the space of functions of the form
\begin{equation}
f(z)=\theta^+_{4\ell_q}(z;p)\theta^+_{4\ell_p}(z;q),
\label{ansatz}\end{equation}
where $\theta^+_{4\ell}(z;q)$ is an arbitrary $A_1$-symmetric
theta-function of order $4\ell\geq 0$ with the modular parameter $q$,
i.e. a holomorphic function of $z\in\C^*$ satisfying the properties
\begin{equation}
\theta^+_{4\ell}(z^{-1};q)=\theta^+_{4\ell}(z;q),\qquad
\theta^+_{4\ell}(qz;q)=\frac{1}{(qz^2)^{2\ell}}\theta^+_{4\ell}(z;q).
\label{theta4l}\end{equation}
As mentioned already, such a consideration is inspired by our
starting intertwining relation \eqref{inter1} and its partner \eqref{inter2}
with $\mathrm{S}_1\sim B(t)$,
which show that the space of zero modes of the $B$-operator
forms an invariant space for two Sklyanin algebra generators.
It is known that the standard Sklyanin algebra generators leave invariant the
space formed by theta functions $\theta^+_{4\ell}(z;q)$ \cite{Sklyanin2}.
But our operator  $B_{mod}(t)$ is symmetric in $p$ and $q$.
Therefore it is natural to consider the above ansatz for $f(z)$ to have
an explicit realization of the automorphism for the elliptic modular
double permuting two Sklyanin algebras \cite{AA2008}.

Substitute now expression \eqref{ansatz} into formula \eqref{Bact1}, use
second relation in \eqref{theta4l}, explicitly substitute
the value of $t$, and remove where possible extra powers
of $p$ or $q$ from theta-function arguments. Then
for integer $\ell_p$ and $\ell_q$ we obtain
\begin{eqnarray}\nonumber &&
[B_{mod}f](w)
=\frac{(-1)^{(2\ell_q+1)(2\ell_p+1)}
p^{\ell_p(1-2\ell_p\ell_q+2\ell_p-2\ell_q)}
q^{\ell_q(1-2\ell_p\ell_q+2\ell_q-2\ell_p)} }
{w^{2(2\ell_p+2\ell_q+1)}\prod_{b=0}^{2\ell_q}\theta(w^{-2}q^b;p)\prod_{a=0}^{2\ell_p}\theta(w^{-2}p^a;q)}
\\  \nonumber &&\makebox[2em]{} \times
\left(\sum_{k=0}^{2\ell_q+1}q^k\frac{\theta(w^2q^{2k-2\ell_q-1};p)}
{\theta(w^2q^{-2\ell_q-1};p)}
\prod_{b=0}^{k-1}\frac{\theta(q^{b-2\ell_q-1},w^2q^{b-2\ell_q-1};p)}
{\theta(q^{b+1},w^2q^{b+1};p)}\theta^+_{4\ell_q}(p^{-\frac{1}{2}}q^{k-\ell_q-\frac{1}{2}}w;p)\right)
\\ && \makebox[2em]{}  \times
\left(\sum_{j=0}^{2\ell_p+1} p^j
\frac{\theta(w^2p^{2j-2\ell_p-1};q)}{\theta(w^2p^{-2\ell_p-1};q)}
\prod_{a=0}^{j-1}\frac{\theta(p^{a-2\ell_p-1},w^2p^{a-2\ell_p-1};q)}
{\theta(p^{a+1},w^2p^{a+1};q)}
\theta^+_{4\ell_p}(q^{-\frac{1}{2}}p^{j-\ell_p-\frac{1}{2}}w;q)\right).
\label{Bact2}\end{eqnarray}
Similar expressions are found when one of the parameters $\ell_p$
or $\ell_q$ (or both) is a half-integer.
We see a complete factorization of the action of our operator to the
proper subspaces $\theta^+_{4\ell_q}(z;p)$ and $\theta^+_{4\ell_p}(z;q)$,
so that each of its factors is independent on the other one.
This means that after finding zero modes for one of the factors
other zero modes are obtained simply by the interchange $p \leftrightarrow q$
and $\ell_p \leftrightarrow \ell_q$.

Note that it is not legitimate to choose $\ell_p=-1/2$ (or $\ell_q=-1/2$)
in formula \eqref{Bact2}, since the second relation in \eqref{theta4l}
is valid only for $\ell\geq 0$. For $2\ell_p+1=0$ (or $2\ell_q+1=0$),
when the second sum is absent, we obtain the previously considered
operator \eqref{Bdiscr}.
We have found thus a space of zero modes of the elliptic Fourier transformation
operator $B_{mod}(t)$ of dimension $d_{zm}=(2\ell_p+1)(2\ell_q+1)$
for $\ell_p, \ell_q\geq 0$. For $2\ell_p+1=0$ one has $d_{zm}=2\ell_q+1$
and, vice versa, for $2\ell_q+1=0$ one has $d_{zm}=2\ell_p+1$.
For $2\ell_p+1=2\ell_q+1=0$ (i.e., $t=\pm1$) there are no zero modes since
$B_{mod}(1)=1$ and $B_{mod}(-1)=P$, the parity operator.
This space forms a nontrivial finite-dimensional invariant subspace of
the $\mathrm{R}$-operator $\R_{12}$ which we plan to investigate in detail
in the future.

It would be interesting to characterize the full space of zero modes of the
integral operator $B_{mod}(t)$. We conjecture that for holomorphic functions
$f(z),\, z\in\C^*$, satisfying the property $f(z^{-1})=f(z)$,
and generic values of bases $p$ and $q$ such zero modes
exist only for $t^2= q^{N}p^{M}$ with $N,M\in\Z,\, (N,M)\neq (0,0)$.
The $A_1$-symmetric products of theta functions of moduli $p$ and $q$
described above are conjectured to form the full finite-dimensional subspace
of these zero modes. Respectively, we conjecture that the latter space
describes all finite-dimensional modules for our R-operator  $\R_{12}$
bases on holomorphic functions.
In this respect it would be interesting to understand how expression
\eqref{Rexplicit} is reduced to Baxter's R-matrix \eqref{Baxter},
Sklyanin's $L$-operator \cite{Sklyanin1}, and
how it is related to Felderhof's solution of YBE \cite{Fel}.

Examples of the meromorphic zero modes for $B_{mod}(t)$ with generic
continuous values of the parameter $t$ are found from the elliptic beta integral
evaluation \eqref{ell_beta}. Indeed, let us rewrite this formula in the form
\begin{eqnarray*} &&
[B_{mod}(t)f](w)=g(w), \qquad
f(z)=\prod_{j=1}^4 \Gamma(t_jz^{\pm1};p,q), \quad
\\ && \makebox[2em]{}
g(w)=\prod_{j=1}^4 \Gamma(tw^{\pm1}t_j;p,q)\prod_{1\leq j<k\leq4}\Gamma(t_jt_k;p,q),
\end{eqnarray*}
where $t^2\prod_{j=1}^4t_j=pq$ and the integration contour $C$
in $B_{mod}(t)$ is chosen in an appropriate way.
Take now, for instance, $t_3t_4=pq$ and assume that parameters
$t_1$ and $t_2$
do not depend on $w$ and $t_1t_2p^jq^k\neq 1,\, j,k \in\Z_{\geq 0}$
(or $t^2\neq p^{-j}q^{-k}$ since we have the constraint $t^2t_1t_2=1$).
Then $\Gamma(t_3t_4;p,q)=0$ and no singularities emerge from
other elliptic gamma functions in $g(w)$. Therefore
one obtains $g(w)=0$, i.e. $f(z)= \Gamma(t_1z^{\pm1},t_2z^{\pm1};p,q)$
is a meromorphic zero mode of the integral operator $B_{mod}(t)$.
If the values of $t_1$ and $t_2$ depend on $w$ then instead of
vanishing the function $g(w)$ may diverge, which is illustrated by the
inversion relation $B_{mod}(t)B_{mod}(t^{-1})=\II$.

\section{Solutions for $\mathrm{Im}(\eta)<0$ and $\mathrm{Im}(\eta)=0$}
\setcounter{equation}{0}

Suppose now that Im$(\eta)<0$. Then a particular solution of the starting equation
for $\Phi(z)$ in \eqref{ab} has the form
$$
\Phi(z) =  \frac{\Gamma(z+b-2\eta|\tau, -2\eta)}{\Gamma(z+a-2\eta|\tau, -2\eta)}.
$$
Note the flip of gamma functions and a shift of $z$ by $2\eta$ with respect
to the Im$(\eta)>0$ case. Using this fact it is not difficult to find
solutions of equations \eqref{pm}, \eqref{eqxfin},
\eqref{eqzfin}, and \eqref{eqxz} for Im$(\eta)<0$.
For instance, instead of \eqref{Delta} one finds
$$
\Delta(z,x) = \mathrm{e}^{- \frac{\pi \textup{i}}{\eta} (x^2-z^2)}
\,\frac{\Gamma(\pm z \pm x -\eta+s|\tau,-2\eta)}
{\Gamma(\pm 2x|\tau,-2\eta)}\varphi(z,x),
$$
where $\varphi(z,x)$ has the same properties as before. Now $\mathrm{S}_k$-operators
should act on functions of the form
$\mathrm{e}^{\frac{\pi \textup{i}}{\eta} x^2}\Psi(x)$
where $\Psi(x+1)=\Psi(-x)=\Psi(x)$. Passing to the space of functions $\Psi(x)$
we come to the following result.

\begin{theorem}
Let $\mathrm{Im}(\eta)<0$ (or $|q|>1$) and $V$ be the space of even and periodic
functions of two complex variables $\Psi(z_1,z_2)$ with the period 1
which do not have simple poles in the domains
$\text{Im}(\eta)\leq \text{Im}(z_1), \text{Im}(z_2)\leq -\text{Im}(\eta)$.
Define three operators
\begin{equation}
[\mathrm{S}_2(u_2-v_1)\Psi](z_1,z_2)=
\Gamma\left(\pm z_1\pm z_2 - u_2  + v_1 - \eta+\textstyle{\frac{\tau}{2}}|\tau,-2\eta\right)
\varphi_2(z_1,z_2)\cdot\Psi(z_1,z_2),
\label{S2fin-}\end{equation}
where $|\sqrt{p/q}e^{2\pi\textup{i}(v_1-u_2)}|<|q|^{-1}$,
\begin{equation}
[\mathrm{S}_1(u_1-u_2)\Psi](z_1,z_2)=\frac{\kappa '}{\Gamma(2u_1-2u_2|\tau,-2\eta)}
\int_0^1\frac{\Gamma(\pm z_1 \pm x +u_1-u_2|\tau,-2\eta)}
{\Gamma(\pm 2x|\tau,-2\eta)}\varphi_1(z_1,x)\cdot \Psi(x,z_2)dx,
\label{S1fin-}\end{equation}
where $|e^{2\pi\textup{i}(u_1  - u_2\pm z_1)}|<|q|^{-1/2}$ and
$\kappa '=(p;p)_\infty(q^{-1};q^{-1})_\infty/2$,
\begin{equation}
[ \mathrm{S}_3(v_1-v_2)\Psi](z_1,z_2)=\frac{\kappa '}{\Gamma(2v_1-2v_2|\tau,-2\eta)}
\int_0^1\frac{\Gamma(\pm z_2 \pm x +v_1-v_2|\tau,-2\eta)}
{\Gamma(\pm 2x|\tau,-2\eta)}\varphi_3(z_2,x)\cdot \Psi(z_1,x)dx,
\label{S3fin-}\end{equation}
where $|e^{2\pi\textup{i}(v_1  - v_2\pm z_2)}|<|q|^{-1/2}$. Here $\varphi_k(z,x)$-functions
have the same properties as in the case Im$(\eta)>0$.

Then the operators $\mathrm{S}_k(\bf{u})$ map the space $V$ onto itself and satisfy
the defining intertwining relations \eqref{RLL13}, \eqref{RLL2}, and \eqref{RLL4},
provided in the corresponding $\mathrm{L}$-operator \eqref{L_op} one uses the Sklyanin
algebra generators realization
$$
\mathbf{S}^a_{mod}= e^{\pi\textup{i}\eta}\frac{(\textup{i})^{\delta_{a,2}}
\theta_{a+1}(\eta)}{\theta_1(2 z) } \Bigl[\,\theta_{a+1} \left(2
z-2\eta\ell\right)\cdot e^{2\pi\textup{i}z}\cdot \mathrm{e}^{\eta \partial} - \theta_{j+1}
\left(-2z-2\eta\ell\right)\cdot e^{-2\pi\textup{i}z}\cdot  \mathrm{e}^{-\eta \partial}\, \Bigl].
$$
They satisfy also the Coxeter relations (\ref{def1}),~(\ref{def2}), and~(\ref{def3})
as a consequence of the elliptic beta
integral evaluation formula with $q$ replaced by $q^{-1}$.
\end{theorem}

For the choice $\varphi_k=1$ the $\mathrm{R}$-matrix has the following explicit form
\begin{eqnarray}\label{Rexplicit-} &&
[\mathbb{R}_{12}(\mathbf{u})f](z_1,z_2)
=\frac{(p;p)_\infty^2(q^{-1};q^{-1})_\infty^2}{(4\pi \textup{i})^2}
\Gamma(\sqrt{p/q}z_1^{\pm1}z_2^{\pm 1}e^{2\pi\textup{i}(v_2-u_1)};p,q^{-1})
\\ && \makebox[2em]{} \times
\int_{\mathbb{T}^2}
\frac{
\Gamma(e^{2\pi\textup{i}(u_1-v_1)}z_2^{\pm1}x^{\pm1},
e^{2\pi\textup{i}(u_2-v_2)}z_1^{\pm1}y^{\pm1},
\sqrt{p/q}e^{2\pi\textup{i}(v_1-u_2)}x^{\pm1}y^{\pm1};p,q^{-1})}
{\Gamma(e^{4\pi\textup{i}(u_1-v_1)},e^{4\pi\textup{i}(u_2-v_2)},
x^{\pm2},y^{\pm2};p,q^{-1})} f(x,y)\frac{dx}{x}\frac{dy}{y}.
\nonumber\end{eqnarray}
Evidently, this expression is symmetric in $p$ and $q^{-1}$, i.e. there
exists the second $\mathrm{RLL}$-intertwining relation obtained from \eqref{RLL}
simply by permuting $-2\eta$ and $\tau$ (or $p$ and $q^{-1}$). The demand
of existence of this modular double forces the functions $\varphi_k$
to be constants independent on $z_1$ and $z_2$.
Note that the operator \eqref{Rexplicit-} is formally obtained from
\eqref{Rexplicit} simply by the changes $u-v \to v-u$ and $\eta \to -\eta$.
However, in difference from the Baxter $\mathrm{R}$-matrix \eqref{Baxter},
this is not a symmetry transformation since both expressions are
defined only for a particularly fixed sign of Im$(\eta)$.

Consideration of the regime Im$(\eta)=0$, or $|q|=1$, is substantially more complicated.
One has to use the modified elliptic gamma function \cite{S2}.
Before passing to corresponding considerations, we would like to consider
the situation when $|q|<1$ and the Sklyanin algebra generators have the form
\begin{equation}\label{autoSklyan}
\mathbf{S}^a =\frac{(\textup{i})^{\delta_{a,2}}
\theta_{a+1}(\eta)}{\theta_1(2 z) } \Bigl[\,\theta_{a+1} \left(2
z-2\eta\ell+\frac{1}{2}\right)\cdot e^{\eta\partial_z}- \theta_{a+1}
\left(-2z-2\eta\ell+\frac{1}{2}\right)\cdot e^{-\eta\partial_z}\,\Bigr] ,
\end{equation}
which differs from \eqref{Sklyan} by the addition of $1/2$ to arguments
of theta functions depending on the spin $\ell$.
These operators represent a particular automorphism of the algebra \cite{Sklyanin2}
with the Casimir operators changed to
$$
\mathbf{K}_0 = 4\,\theta_2^2\bigl((2\ell+1)\,\eta\bigr)\ ;\qquad
\mathbf{K}_2 = 4\,\theta_2\bigl(2(\ell+1)\,\eta\bigr)\,\theta_2(2\ell\,\eta)\,.
$$

One can check that in this case the factorization \eqref{factL}
has the same form with the replacements $u_1\to u_1-1/4$ and $u_2\to u_2+1/4$.
Similar shifts  $v_1\to v_1-1/4$ and $v_2\to v_2+1/4$ take place in
the second L-operator entering intertwining relation \eqref{RLL2}.
Substituting these shifts in the appropriate places of derivation of
the $\mathrm{S}_2$-operator, this time we come to the following equations
\begin{eqnarray}\nonumber &&
\theta_3(z_1+z_2 +u_2-v_1)\, \mathrm{S}(z_1-\eta,z_2) =
\theta_3(z_1+z_2 +v_1-u_2)\, \mathrm{S}(z_1,z_2+\eta)\ ,
\\ \nonumber &&
\theta_3(z_1+z_2 -u_2 +v_1)\, \mathrm{S}(z_1 + \eta ,z_2) =
\theta_3(z_1+z_2 -v_1 +u_2)\, \mathrm{S}(z_1,z_2-\eta)\ ,
\\ \nonumber &&
\theta_3(z_1-z_2 + u_2 - v_1)\, \mathrm{S}(z_1 - \eta,z_2) =
\theta_3(z_1-z_2 +v_1 -u_2)\, \mathrm{S}(z_1,z_2-\eta)\ ,
\\  &&
\theta_3(z_1-z_2 -u_2 +v_1)\, \mathrm{S}(z_1 + \eta,z_2) =
\theta_3(z_1-z_2 -v_1 +u_2)\, \mathrm{S}(z_1,z_2+\eta)\
\label{modS2eqs}\end{eqnarray}
with the general solution for $\mathrm{S}_2$-operator
\begin{equation}
\mathrm{S}_2(a) =
\Gamma\left(\pm z_1\pm z_2 + a+ \textstyle{\frac{1}{2}}
+\eta+\textstyle{\frac{\tau}{2}|\tau,2\eta}\right)\varphi_2(z_1,z_2),
\quad a=u_2  - v_1.
\label{S2mod}\end{equation}
For $\varphi_2=1$, one still has $\mathrm{S}_2(-a)\mathrm{S}_2(a)=\II$,
as needed. Similar picture holds for $|q|>1$ regime as well.

As to the operators $\mathrm{S}_1$ and $\mathrm{S}_3$, they do not
change their form at all. Indeed, the intertwining relations \eqref{inter}
lead to equations \eqref{W-sys} with the replacements $s\to s-1/2$
in the first row theta functions and $s\to s+1/2$ in the second row.
As a result of the latter inhomogeneity, it happens that equation \eqref{eqx1}
does not change apart of the overall sign for all terms,
equation \eqref{eqx2} does not change at all. As a result, the
final equations \eqref{eqxfin}, \eqref{eqzfin}, and \eqref{eqxz}
do not change at all. Therefore, the shape of the $\mathrm{S}_1$-operator
does not change. Validity of the cubic Coxeter relation is
guaranteed again by the elliptic beta integral with the replacement
of corresponding parameters $g_3\to g_3+1/2$ and $g_4\to g_4+1/2$,
which does not spoil the balancing condition \eqref{balance}
defined modulo $\Z$. As a result the R-operator \eqref{Rexplicit}
is slightly changed --- it is necessary to replace in
the kernel of this operator $+\sqrt{pq}$ by $-\sqrt{pq}$.
The modular double exists as well with the partner Sklyanin
algebra generators being obtained from \eqref{mod_doub2}
after the replacement $g \to g-1/2$.

Now it is straightforward to build a solution of equations \eqref{modS2eqs}
which is well defined for Im$(\eta)=0$.
First we set
$2\eta=\omega_1/\omega_2, \tau=\omega_3/\omega_2$ and renormalize
all other variables in Sklyanin's L-operator
$$
z_{1}\to \frac{z_{1}}{\omega_2}, \quad z_{2}\to \frac{z_{2}}{\omega_2},
\quad u\to  \frac{u}{\omega_2}, \quad v\to  \frac{v}{\omega_2}.
$$
Then equation \eqref{ab} takes the form
$$
\Phi(z+\omega_1) = \mathrm{e}^{\pi \textup{i} \frac{a-b}{\omega_2}}\,
\frac{\theta_1(\frac{z+a}{\omega_2})}{\theta_1(\frac{z+b}{\omega_2})} \cdot \Phi(z)\; ,
$$
which has a particular solution of the form
$$
\Phi(z) =  \frac{G(z+a;{\bf \omega})}{G(z+b;{\bf \omega})},
$$
where $G(z;{\bf \omega})$ is the modified elliptic gamma function \eqref{MEGF1}
well defined for $|q|\leq 1$ and satisfying the same key equation
as $\Gamma(e^{2\pi\textup{i}u/\omega_2};p,q)$ \eqref{G_eq}. Using these facts,
we can immediately write out the final general expression for $\mathrm{S}_2$-operator
following from equations \eqref{modS2eqs} and valid for
Im$(\tau)>0,$ Im$(\eta/\tau)<0$ (which admits Im$(\eta)=0$):
$$
\mathrm{S}_2(a)=G(\pm z_1\pm z_2 + a + \textstyle{\frac{1}{2}
\sum_{k=1}^3\omega_k}\, ;{\bf \omega})
\varphi_2(z_1,z_2), \quad a= u_2-v_1,
$$
where $\varphi_2(z_1+\omega_1,z_2)=\varphi_2(z_1,z_2+\omega_1)=
\varphi_2(z_1+\omega_1/2,z_2+\omega_1/2)$. Because of the inversion formula
for $G(z;{\bf \omega})$-function,
for $\varphi_2=1$ one has $\mathrm{S}_2(-a)\mathrm{S}_2(a)=\II$.

A solution of equations  \eqref{eqxfin}, \eqref{eqzfin}, and \eqref{eqxz}
for the $\Delta$-kernel valid for Im$(\eta)=0$ has the form
\begin{equation}\label{Delta3}
\Delta(z,x) = \mathrm{e}^{\frac{2\pi\textup{i}}{\omega_1\omega_2} (x^2-z^2)}
\,\frac{G(\pm x\pm z -u_1+u_2;\mathbf{\omega})}
{G(\pm 2x;\mathbf{\omega})}\varphi_1(z,x),
\end{equation}
where $\varphi_1$-function has the same periodicity properties as $\varphi_2$.

Substitute now the second form of $G(x;\mathbf{\omega})$-function \eqref{MEGF2}
into these expressions. Then we can write
$$
\mathrm{S}_2(a)=
e^{-\frac{4\pi\textup{i} a}{\omega_1\omega_2\omega_3}(z_1^2+z_2^2)}
\cdot \Gamma\big(-\textstyle{\frac{1}{\omega_3}}(\pm z_1\pm z_2 + a +
\textstyle{\frac{1}{2}\sum_{k=1}^3\omega_k } )\, \big|
-\frac{\omega_2}{\omega_3}, -\frac{\omega_1}{\omega_3}\big)
\cdot\varphi_2'(z_1,z_2),
$$
where
$$
\varphi_2'(z_1,z_2)=\mathrm{e}^{-\frac{4\pi\textup{i}}{3}
B_{3,3}(a+\frac{1}{2}\sum_{k=1}^3\omega_k;{\bf \omega})}\cdot \varphi_2(z_1,z_2).
$$
Analogously,
\begin{eqnarray*}
\Delta(z,x) = \mathrm{e}^{\frac{2\pi \textup{i}}{\omega_1\omega_2} (x^2-z^2)}
\mathrm{e}^{\frac{4\pi \textup{i}}{\omega_1\omega_2\omega_3}
[x^2(b-\frac{1}{2}\sum_{k=1}^3\omega_k)+z^2(b+\frac{1}{2}\sum_{k=1}^3\omega_k)]}
\\
\times \,\frac{\Gamma\big(-\textstyle{\frac{1}{\omega_3}}(\pm x\pm z -b)\big|
-\frac{\omega_2}{\omega_3}, -\frac{\omega_1}{\omega_3}\big)}
{\Gamma\big(\pm \frac{2x}{\omega_3}\big| -\frac{\omega_2}{\omega_3},
-\frac{\omega_1}{\omega_3}\big)}\cdot\varphi_1'(z,x)
\end{eqnarray*}
where $b=u_1-u_2$ and
$$
\varphi_1'(z,x)=\mathrm{e}^{\frac{2\pi \textup{i}}{3}
(B_{3,3}(0;{\bf \omega})- 2B_{3,3}(-b;{\bf \omega}))}\cdot \varphi_1(z,x).
$$

After derivation of the $\Delta$-kernel we have to fix the integration
interval $[\alpha,\beta]$ and the space of functions for which the
$\mathrm{S}_1$-operator really satisfies the intertwining relations.
First, as it was done earlier, we pass to the modified Sklyanin algebra
generators \eqref{modSklgen} with additional shift by $1/2$ in the
arguments of $\ell$-dependent theta functions and conjugate
similarly $\mathrm{S}_k$-operators. This does not change the operator
$\mathrm{S}_2$, but removes the exponential
$\mathrm{e}^{\frac{2\pi \textup{i}}{\omega_1\omega_2}(x^2-z^2)}$ from $\Delta(z,x)$.
Then we note that the ratio of elliptic gamma functions in $\Delta$ is a periodic
function of $z$ and $x$ with the period $\omega_3$. Therefore
we set $\alpha=0$ and $\beta=\omega_3$ and demand that the modified
operator $\mathrm{S}_1$ acts on functions $\Phi(x)$ such that
$\Psi(x):=\mathrm{e}^{\frac{4\pi \textup{i}}{\omega_1\omega_2\omega_3}
x^2(b-\frac{1}{2}\sum_{k=1}^3\omega_k)}\Phi(x)$ is an even
$\omega_3$-periodic function of $x$, $\Psi(-x)=\Psi(x+\omega_3)=\Psi(x)$.
Finally, equations for $\Delta$-kernel are true provided $\Delta(z,x)$ has no
poles in the parallelogram $x\in [-\omega_1, \omega_3-\omega_1, \omega_3+\omega_1,
\omega_1]$ which, by complete analogy with the previous cases, is
guaranteed for $|e^{2\pi\textup{i}(b\pm z)/\omega_3}|<|e^{\pi\textup{i}\omega_1/\omega_3}|$.

So, we have found operators $\mathrm{S}_1$, $\mathrm{S}_2$, and $\mathrm{S}_3$
(it differs from $\mathrm{S}_1$ only by the space where it acts). Returning back
to the original notation, i.e. renormalizing back $x\to x\omega_2, u\to u\omega_2$, etc,
we come to the following theorem.

\begin{theorem}
Let Im$(\tau)>0$ (i.e., $|p|<1$) and $\mathrm{Im}(\eta/\tau)<0$
(for $\mathrm{Im}(\eta)=0$, i.e. $|q|=1$, this assumes $\mathrm{Re}(\eta)>0)$.
Denote $\varphi_k'(z,x), \, k=1,2,3,$
arbitrary even elliptic functions of $z$ and $x$ with periods $\tau$ and
$2\eta$ satisfying additional constraints
$$
\varphi_k'(z+\eta,x+\eta)=\varphi_k'(z,x), \quad k=1,2,3,
$$
and not having simple poles in the domains
$\text{Im}(\eta/\tau)\leq \text{Im}(z/\tau), \text{Im}(x/\tau)\leq -\text{Im}(\eta/\tau)$.

Define the operators
\begin{equation}
\mathrm{S}_2(u_2-v_1)=
e^{\frac{2\pi\textup{i}(v_1-u_2)}{\eta\tau}(z_1^2+z_2^2)}
\cdot \Gamma\big(\frac{\pm z_1\pm z_2 + v_1-u_2}{\tau} -
\textstyle{\frac{1}{2}}-\frac{\eta}{\tau}- \textstyle{\frac{1}{2\tau}}\, |
-\frac{1}{\tau}, -\frac{2\eta}{\tau}\big)
\cdot\varphi_2'(z_1,z_2),
\label{S2finm}\end{equation}
where $|\sqrt{\tilde p\tilde r}e^{2\pi\textup{i}(v_1-u_2)/\tau}|<|\tilde r|$
with $\tilde p=e^{-2\pi\textup{i}/\tau}$ and
$\tilde r= e^{-4\pi\textup{i}\eta/\tau}$,
and
\begin{eqnarray}\nonumber
[\mathrm{S}_1(a)\Phi ](z_1,z_2) =\tilde\kappa \int_0^{\tau}\mathrm{e}^{\frac{2\pi \textup{i}}{\eta\tau}
[x^2(a-\frac{1}{2}-\eta-\frac{\tau}{2})+z_1^2(a+\frac{1}{2}+\eta+\frac{\tau}{2})]}
\frac{ \Gamma\big(\frac{1}{\tau} (\pm x\pm z_1 +a)|
-\frac{1}{\tau}, -\frac{2\eta}{\tau}\big) }
{ \Gamma\big(\frac{2a}{\tau},\pm \frac{2x}{\tau}| -\frac{1}{\tau}, -\frac{2\eta}{\tau}\big) }
\cdot\varphi_1'(z_1,x)\Phi(x,z_2)\frac{dx}{\tau},
\label{S1finm}\end{eqnarray}
where $a=u_1-u_2$, $|e^{2\pi\textup{i}(a\pm z_1)/\tau}|<|\tilde r|^{1/2}$ and
$\tilde\kappa = (\tilde p;\tilde q)_\infty\,(\tilde r;\tilde r)_\infty/2,$
\begin{equation}
[\mathrm{S}_3(b)\Phi ](z_1,z_2) =\tilde\kappa \int_0^{\tau}\mathrm{e}^{\frac{2\pi \textup{i}}{\eta\tau}
[x^2(b-\frac{1}{2}-\eta-\frac{\tau}{2})+z_2^2(b+\frac{1}{2}+\eta+\frac{\tau}{2})]}
\frac{ \Gamma\big( \frac{1}{\tau} (\pm x\pm z_2 +b)|
-\frac{1}{\tau}, -\frac{2\eta}{\tau}\big) }
{ \Gamma\big(\frac{2b}{\tau},\pm \frac{2x}{\tau}| -\frac{1}{\tau}, -\frac{2\eta}{\tau}\big) }
\cdot\varphi_3'(z_2,x)\Phi(z_1,x)\frac{dx}{\tau},
\label{S3finm}\end{equation}
where $b=v_1-v_2$, $|e^{2\pi\textup{i}(b\pm z_2)/\tau}|<|\tilde r|^{1/2}$.

Denote as $V_{a,b}$ the space of functions
of two complex variables $\Phi(z_1,z_2)$ such that the products
$$
\mathrm{e}^{\frac{2\pi \textup{i}}{\eta\tau}
z_1^2(a-\frac{1}{2}-\eta- \frac{\tau}{2})}
\mathrm{e}^{\frac{2\pi \textup{i}}{\eta\tau}
z_2^2(b-\frac{1}{2}-\eta- \frac{\tau}{2})}\Phi(z_1,z_2)
$$
are even and periodic in $z_1$ and $z_2$
with the period $\tau$ and which do not have simple poles in the domains
$\text{Im}(\eta/\tau)\leq \text{Im}(z_1/\tau), \text{Im}(z_2/\tau)\leq -\text{Im}(\eta/\tau)$.
Then the operators $\mathrm{S}_1$,  $\mathrm{S}_2$ and $\mathrm{S}_3$
map the space  $V_{a,b}$ for $a=u_1-u_2$ and $b=v_1-v_2$ onto itself and they satisfy
the defining intertwining relations \eqref{RLL13}, \eqref{RLL2} and  \eqref{RLL4}
provided in the corresponding $\mathrm{L}$-operator \eqref{L_op} one uses the Sklyanin
algebra generators of the form
\begin{eqnarray}\nonumber && \makebox[-2em]{}
\mathbf{S}^a_{mod}=  e^{-\pi\textup{i} \eta}
\frac{(\textup{i})^{\delta_{a,2}}
\theta_{a+1}(\eta)}{\theta_1(2 z) } \Bigl[\,\theta_{a+1} \left(2
z-2\eta\ell+\textstyle{\frac{1}{2}}\right)\cdot e^{-2\pi\textup{i}z}\cdot
 \mathrm{e}^{\eta \partial_z}
\\ && \makebox[10em]{}
- \theta_{a+1} \left(-2z-2\eta\ell+\textstyle{\frac{1}{2}}\right)\cdot
e^{2\pi\textup{i}z}\cdot  \mathrm{e}^{-\eta \partial_z}\, \Bigl].
\label{modSklgen'}\end{eqnarray}
\end{theorem}

It remains to confirm the Coxeter relations for the choice $\varphi_k'=1$.
The relation $\mathrm{S}_2(a)\mathrm{S}_2(-a)=\II$ is evident.
It is not difficult to check the cubic Coxeter relation \eqref{key2}
as it leads to a solution of the star-triangle relation for Im$(\eta)=0$
considered in \cite{spi:conm}. We shall not present corresponding
details -- although it is neater than before, all the exponential factors
cancel and the identity is reduced again to the computation of the
elliptic beta integral  in a particular parameterization \cite{spi:umn}.
Also, it follows from the integral analogue of the Bailey
lemma formulated in terms of the $G(z;{\bf \omega})$-function.
In a similar way,  equalities
$\mathrm{S}_1(a)\mathrm{S}_1(-a)=\mathrm{S}_3(a)\mathrm{S}_3(-a)=\II$
are reduced to the previously considered inversion relations in
a different parameterization because of the
cancellation of exponential factors. In general,
in the arbitrary product $\ldots \mathrm{S}_i\mathrm{S}_j\mathrm{S}_k\ldots$
one never violates the restrictions on space of functions $\Phi(z_1,z_2)$
needed for operators $\mathrm{S}_1$ and  $\mathrm{S}_3$, i.e. all the exponential factors
can be pulled out to the far left and far right.

We could write out the explicit form of the R-operator
using its factorized form, but it is skipped since this is a straightforward procedure
leading to a somewhat cumbersome expression. After dropping
the exponential factors from this expression one would come
to the R-operator which is obtained from \eqref{Rexplicit}
(with $\sqrt{pq}$ replaced by $-\sqrt{pq}$) by a simple modular
transformation $(\omega_2,\omega_3) \to (-\omega_3,\omega_2)$.

As to the elliptic modular double, the R-operator written in terms of
the $G(u;\mathbf{\omega})$-functions (i.e., the form obtained
after scalings $z\to z/\omega_2,$ etc) is symmetric
with respect to the permutation $\omega_1\leftrightarrow\omega_2$.
In the original notation $z,g:=\eta(2\ell+1), \eta,\tau$,
the permutation of these quasiperiods is equivalent to the
changes $\eta\to 1/4\eta, \, \tau\to \tau/2\eta, z\to z/2\eta, g\to g/2\eta$
(here $g$ is considered as an independent variable). Therefore
the derived R-operator also has second RLL-relation, where $\mathrm{L}^{doub}$
is composed of a new Sklyanin algebra generators of the form
\begin{eqnarray}\nonumber && \makebox[-2em]{}
\mathbf{\tilde S}^a_{mod}=  e^{-\frac{\pi\textup{i}}{4 \eta}}
\frac{(\textup{i})^{\delta_{a,2}}
\theta_{a+1}(\textstyle{\frac{1}{4\eta}}|\frac{\tau}{2\eta})}
{\theta_1(\frac{z}{\eta}|\frac{\tau}{2\eta})) } \Bigl[\,\theta_{a+1} \left(
\frac{2z-g+1}{2\eta} \Big|\frac{\tau}{2\eta}\right)\cdot
e^{-\pi\textup{i}z/\eta}\cdot
 \mathrm{e}^{\frac{1}{4\eta} \partial_z}
\\ && \makebox[4em]{}
- \theta_{a+1} \left(\frac{-2z-g+1}{2\eta} \Big|\frac{\tau}{2\eta}\right)\cdot
e^{\pi\textup{i}z/\eta}\cdot  \mathrm{e}^{-\frac{1}{4\eta} \partial_z}\, \Bigl].
\label{modSklgendoub}\end{eqnarray}
This elliptic modular double has been introduced  in \cite{AA2008} as well.
Thus we have found solutions of the Yang-Baxter equation \eqref{YB}
for all possible regions of the key complex parameter $\eta$.

\section{Conclusion}
\setcounter{equation}{0}

In this paper we have merged two constructions from the theory of
quantum integrable systems and the theory of special functions.
One construction is a specific approach to building YBE solutions
developed in \cite{SD,DM0,DM,DKK}. It is based on  the
realization of the permutation group generators by various
operators acting in the functional spaces and it directly leads to
the factorized form of the R-matrices as products of elementary
transposition operators. Another construction is the elliptic beta
integral evaluation \cite{spi:umn} and its various consequences formulated
as an elliptic Fourier transformation and integral Bailey lemma
\cite{spi:bailey}. This result formed a basis for developing
the theory of a principally new and very powerful class of special
functions ---  elliptic hypergeometric integrals \cite{S2,spi:essays}.
As a result of our considerations, both fields have benefited
and mutually enriched each other. The most complicated known R-matrix
at the elliptic level appeared to be defined by an integral operator with
an elliptic hypergeometric kernel and algebraic properties of
the integral Bailey lemma ingredients got a natural interpretation
as Coxeter relations for the permutation group generators.
Moreover, the key integral operator defining the elliptic Fourier
transformation appeared to be an intertwining operator for the
Sklyanin algebra. The general construction shows that
YBE is a simple consequence of a particular word identity in
the group algebra for the braid group $\mathfrak{B}_6$ or, in our case,
of the symmetric group $\mathfrak{S}_6$,
whose generators are realized as integral operators.

Our results can be applied to all known forms of YBE \cite{PA}. In particular,
a generalization of our construction to root systems
is relatively straightforward due to the abundance of
elliptic beta integrals on root systems \cite{spi:essays}
and corresponding elliptic Fourier transformations \cite{spi-war:inversions}.
In the rational case the most general known R-operator for $A_n$-root
system was constructed in \cite{DM0,DM}.
Star-triangle and star-star type relations for the
root systems following from the elliptic hypergeometric integral
identities were considered in \cite{spi:conm,BS2}.

In \cite{VF}, Faddeev and Volkov constructed a solution of YBE
at the $q$-hypergeometric level with the help of the
pentagon relation for noncompact quantum dilogarithms.
A generalization of this model has been found in \cite{spi:conm}
and a question was posed --- is it still related
to the pentagon relation and does there exist an elliptic analogue
of the latter ? As shown in \cite{DKK}, the method used in the present paper
works at the compact $q$-hypergeometric level using $q$-exponential functions.
The noncompact situation can be treated as well after
appropriate replacement of $q$-exponentials by the noncompact quantum dilogarithms.
However, it is easy to degenerate our elliptic results to both
compact and non-compact $q$-levels. From our analysis of the elliptic
hypergeometric constructions it is not clear which relation can be taken as
a direct elliptic analogue of the pentagon relation. Instead of a
potential five-term relation, the key role is played by the hexagon
relation \eqref{STR} emerging in the theory of elliptic
Fourier transformation \cite{spi:bailey} and defining the
Coxeter relation for permutation operators (or the star-triangle
relation in integrable models of statistical mechanics).

We would like to stress that the spin variable $\ell$ in our analysis
takes continuous values. Therefore, strictly speaking, we deal not with the
discrete Ising-type models, but with two-dimensional quantum field theories.
In this context the Yang-Baxter equation can be interpreted as a condition
of factorizing the $N$-body $S$-matrix to the product
of two-body scattering matrices \cite{PA}.

Let us discuss briefly an application of our results to four-dimensional ($4d$)
supersymmetric gauge field theories. The key discovery of \cite{DO}
consists in the fact that superconformal indices of these theories
are described by the elliptic hypergeometric integrals.
For example, the elliptic beta integral evaluation formula gets a remarkable
interpretation as a direct indication on the confinement phenomenon in
the simplest $4d$ supersymmetric quantum chromodynamics.

A relation between $4d$  Nekrasov instanton
partition function and conformal blocks in $2d$ Liouville
field theory was empirically discovered in \cite{AGT}.
In \cite{GPRR}, superconformal indices of $4d$ $\mathcal{N}=2$ supersymmetric
field theories were tied to correlation functions of
$2d$ topological field theories. Our results are relevant to a different type
of $4d/2d$ correspondence discovered in \cite{spi:conm},
where $4d$ superconformal indices coincide with partition functions of
integrable models of $2d$ spin systems.

Seiberg duality is a special electric-magnetic duality of
$4d$ supersymmetric non-abelian gauge field theories.
In the language of elliptic hypergeometric integrals there
are two qualitatively different situations. When the dual theory confines
corresponding superconformal index is identical to some elliptic
beta integral on a root system or, from the statistical mechanics
point of view, to star-triangle relation for multicomponent
spin systems. When the dual theory is a
nontrivial interacting field theory, one deals with symmetry
transformations for integrals equivalent to the star-star
relations in statistical mechanics \cite{spi:conm}.

Since superconformal indices for simple gauge groups coincide with
the statistical sums of elementary cells for spin systems on the plane,
the Seiberg duality transformations represent
the Kramers-Wannier type duality transformations for
corresponding new $2d$ integrable models.
As indicated in \cite{SV} sequential integral transformations
following from the Bailey lemma define superconformal indices of particular quiver
gauge theories. This procedure corresponds to building full two-dimensional
lattice partition functions as prescribed in the theory of quantum integrable systems.

\section*{Acknowledgments}

The authors are deeply indebted to L. D. Faddeev and A. M. Vershik for
useful discussions and general support.
The work of V. S. is supported by RFBR grant
no. 11-01-00980 and NRU HSE scientific fund grant no. 12-09-0064.
The work of S. D. is supported by RFBR grants 11-01-00570, 11-01-12037,
12-02-91052 and Deutsche Forschungsgemeinschaft (KI 623/8-1).

\section{Appendix}
\setcounter{equation}{0}

In this Appendix we collect some useful formulae. The standard
infinite $q$-product is defined as
\begin{equation}
(x;q)_\infty = \prod_{k=0}^{+\infty}(1-q^k\cdot x)
\ ;\quad  q\in \C \ ,\quad  |q|<1.
\label{qinf}\end{equation}
The general theta-function with characteristics has the form
$$
\theta_{a,b}(z|\tau) = \sum_{n\in\mathbb{Z}} \mathrm{e}^{\pi \textup{i}
(n+\frac{a}{2})^2\tau}\cdot \mathrm{e}^{2\pi \textup{i} (n+\frac{a}{2})(z+\frac{b}{2})}.
$$
We use four standard theta-functions
$$
\theta_{1}(z|\tau) = -\theta_{1,1}(z|\tau) = -\sum_{n\in\mathbb{Z}}
\mathrm{e}^{\pi \textup{i} \left(n+\frac{1}{2}\right)^2\tau}\cdot
\mathrm{e}^{2\pi \textup{i}
\left(n+\frac{1}{2}\right)\left(z+\frac{1}{2}\right)}
$$
\begin{equation}
=\textup{i}p^{1/8} e^{-\pi \textup{i}z}\: (p;p)_\infty\: \theta(e^{2\pi \textup{i}z};p),
\label{theta1}\end{equation}
where $p=e^{2\pi \textup{i} \tau}$ and
\begin{equation}
\theta(t;p)=(t;p)_\infty(pt^{-1};p)_\infty,
\label{theta_p}\end{equation}
\be
\theta_{2}(z|\tau) = \theta_{1,0}(z|\tau) = \sum_{n\in\mathbb{Z}}
\mathrm{e}^{\pi \textup{i} \left(n+\frac{1}{2}\right)^2\tau}\cdot
\mathrm{e}^{2\pi \textup{i} \left(n+\frac{1}{2}\right) z }
\ee
\be
\theta_{3}(z|\tau) = \theta_{0,0}(z|\tau) = \sum_{n\in\mathbb{Z}}
\mathrm{e}^{\pi \textup{i} n^2\tau}\cdot \mathrm{e}^{2\pi \textup{i} n z }
\ee
\be
\theta_{4}(z|\tau) = \theta_{0,1}(z|\tau) = \sum_{n\in\mathbb{Z}}
\mathrm{e}^{\pi \textup{i} n^2 \tau}\cdot \mathrm{e}^{2\pi \textup{i}
n\left(z+\frac{1}{2}\right)}.
\ee
The following identities are used to factorize the $\mathrm{L}$-operator and
to derive defining equations for the operators
$\mathrm{S}_1$, $\mathrm{S}_2$, and $\mathrm{S}_3$:
\begin{eqnarray}\label{11} &&
2\,\theta_1(x+y)\,\theta_1(x-y) = \bar\theta_4(x)\,\bar\theta_3(y)
-\bar\theta_4(y)\,\bar\theta_3(x),
\\ \label{22} &&
2\,\theta_2(x+y)\,\theta_2(x-y) = \bar\theta_3(x)\,\bar\theta_3(y)
-\bar\theta_4(y)\,\bar\theta_4(x),
\\ \label{33} &&
  2\,\theta_3(x+y)\,\theta_3(x-y)
= \bar\theta_3(x)\,\bar\theta_3(y) +\bar\theta_4(y)\,\bar\theta_4(x),
\\ \label{44} &&
2\, \theta_4(x+y)\,\theta_4(x-y)=\bar\theta_4(x)\,\bar\theta_3(y)
+\bar\theta_4(y)\,\bar\theta_3(x),
\\ \label{14} &&
2\, \theta_4(x+y)\,\theta_1(x-y)=\bar\theta_1(x)\,\bar\theta_2(y)
-\bar\theta_1(y)\,\bar\theta_2(x),
\\ \label{12} &&
\bar\theta_1(x-y)\,\bar \theta_2(x+y)=\theta_1(2x)\,\theta_4(2y)
-\theta_1(2y)\,\theta_4(2x),
\end{eqnarray}
where $\bar\theta_a(z) \equiv \theta_a\left(z|\frac{\tau}{2}\right)$.
We need also the duplication formula
\begin{equation}
\theta_1(2x|2\tau)=\frac{(-p;p)_\infty}{(p;p)_\infty}
\theta_1(x|\tau)\theta_2(x|\tau),\qquad p=e^{2\pi\textup{i}\tau}.
\label{thetadup}\end{equation}

For Im$(\tau)>0$, Im$(\eta)>0$ the elliptic gamma function is defined by the
double infinite product
\be \Gamma(z|\tau,2\eta) \equiv
\prod_{n,m=0}^{\infty} \frac{1-\mathrm{e}^{2\pi
\textup{i}\left(\tau(n+1)+2\eta(m+1)-z\right)}}{1-\mathrm{e}^{2\pi
\textup{i}\left(\tau n+2\eta m+z\right)}}.
\label{egamma_a}\ee
It is symmetric in its modular parameters $\Gamma(z|\tau,2\eta)=
\Gamma(z|2\eta,\tau)$ and satisfies
equations
\be
\Gamma(z+1|\tau,2\eta)=\Gamma(z|\tau,2\eta),
\ee
 \be \Gamma(z+\tau|\tau,2\eta)
= \theta(e^{2\pi \textup{i} z};e^{4\pi \textup{i}\eta})\cdot\Gamma(z|\tau,2\eta),
\ee
\be \Gamma(z+2\eta|\tau,2\eta)
= \theta(e^{2\pi \textup{i} z};e^{2\pi \textup{i}\tau})\cdot\Gamma(z|\tau,2\eta),
\label{eq_gamma3}\ee
and the normalization condition $\Gamma(\eta+\tau/2|\tau,2\eta)=1$.
One can evidently replace in these equations
$$
\theta(e^{2\pi \textup{i} z};e^{2\pi \textup{i}\tau})=\mathrm{R}(\tau)\cdot \mathrm{e}^{\pi i
z}\theta_1(z|\tau), \quad
\theta(e^{2\pi \textup{i} z};e^{4\pi \textup{i}\eta})=\mathrm{R}(2\eta)\cdot \mathrm{e}^{\pi i
z}\theta_1(z|2\eta),
$$
where the constant $\mathrm{R}(\tau)$ does not depend on $z$:
$\mathrm{R}(\tau) = -\textup{i}\mathrm{e}^{-\frac{\pi \textup{i} \tau}{4}} \cdot
\left(\mathrm{e}^{2\pi \textup{i} \tau};\mathrm{e}^{2\pi \textup{i} \tau}\right)^{-1}_\infty$.

Zeros of $\Gamma(z|\tau,2\eta)$ are located at $z=\Z
+\tau \Z_{>0}+2\eta\Z_{>0}$ and poles at $z=\Z+\tau\Z_{\leq0}+2\eta\Z_{\leq0}$.
The reflection equation for this function has the form
\begin{equation}
\Gamma(z|\tau,2\eta)\Gamma(-z+2\eta+\tau|\tau,2\eta)=1.
\label{refl}\end{equation}
In the multiplicative notation one has
\begin{equation}
\Gamma(t;p,q)=\prod_{j,k=0}^\infty\frac{1-t^{-1}p^{j+1}q^{k+1}}{1-tp^jq^k},
\quad |p|, |q|<1,
\label{egamma_m}\end{equation}
so that $\Gamma(t;p,q)\Gamma(pq/t;p,q)=1$ and
$$
\Gamma(qt;p,q)=\theta(t;p)\Gamma(t;p,q), \quad
\Gamma(pt;p,q)=\theta(t;q)\Gamma(t;p,q).
$$

For incommensurate $\omega_1,\omega_2,\omega_3\in\C$ define three base variables,
\begin{eqnarray*}
&& q= e^{2\pi\textup{i}\frac{\omega_1}{\omega_2}}, \quad
p=e^{2\pi\textup{i}\frac{\omega_3}{\omega_2}}, \quad  r=e^{2\pi\textup{i}\frac{\omega_3}{\omega_1}},
\\ &&
\tilde q= e^{-2\pi\textup{i}\frac{\omega_2}{\omega_1}}, \quad
\tilde p=e^{-2\pi\textup{i}\frac{\omega_2}{\omega_3}},   \quad
\tilde r=e^{-2\pi\textup{i}\frac{\omega_1}{\omega_3}},
\end{eqnarray*}
where $\tilde q,\tilde p,\tilde r$ denote particular modular transformed bases.
The condition that $\sum_{k=1}^3n_k\omega_k\neq 0, \; n_k\in\Z,$ implies
that none of $p$, $q$, and $r$ is a root unity.

For $|q|, |p|<1$ (which assumes $|r|<1$) the modified elliptic gamma function
is defined as
\begin{equation} \label{MEGF1}
G(u;{\bf \omega}) =  \Gamma(e^{2 \pi \textup{i}
u/\omega_2};p,q) \Gamma(r e^{-2 \pi \textup{i}
u/\omega_1};\widetilde{q},r) = \frac{\Gamma(e^{2 \pi \textup{i}
u/\omega_2};p,q)}{\Gamma(\widetilde{q} e^{2 \pi \textup{i} u/\omega_1};\widetilde{q},r)}.
\end{equation}
This is a meromorphic function of $u$ even for $\omega_1/\omega_2>0$,
when $|q|=1$, which is easily seen from its another representation
 \begin{equation} \label{MEGF2}
G(u;{\bf \omega}) \ = \ e^{-\frac{\pi
\textup{i}}{3} B_{3,3}(u;\mathbb{\omega})} \Gamma(e^{-2 \pi \textup{i}
u/\omega_3};\widetilde{r},\widetilde{p}),
\end{equation}
where $B_{3,3}$ is a Bernoulli polynomial of the third order
\begin{eqnarray} \label{B33} && \makebox[-2em]{}
B_{3,3}(u;\mathbb{\omega}) = \frac{1}{\omega_1\omega_2 \omega_3}
\Bigl(u-\frac12\sum_{k=1}^3\omega_k\Bigr)\Bigl((u-\frac12\sum_{k=1}^3\omega_k )^2
-\frac14 \sum_{k=1}^3\omega_k^2\Bigr).
\end{eqnarray}
Multiple Bernoulli polynomials are defined in the theory of Barnes
multiple zeta-function from the following expansion
$$
\frac{x^m e^{xu}}{\prod_{k=1}^m(e^{\omega_k x}-1)}
=\sum_{n=0}^\infty B_{m,n}(u;\omega_1,\ldots,\omega_m)\frac{x^n}{n!}.
$$

This function satisfies the equations
\begin{equation}
G(u+\omega_1)=\theta(e^{2\pi\textup{i} u/\omega_2};p)G(u),
\label{G_eq}\end{equation}
$$
G(u+\omega_2) =\theta(e^{2\pi\textup{i} u/\omega_1};r) G(u),
\qquad   G(u+\omega_3) =e^{-\pi\textup{i}B_{2,2}(u;\mathbf{\omega})} G(u).
$$
and the normalization condition $G(\sum_{m=1}^3\omega_m/2)=1$.
Here
$$
B_{2,2}(u;\mathbf{\omega})=\frac{u^2}{\omega_1\omega_2}
-\frac{u}{\omega_1}-\frac{u}{\omega_2}+
\frac{\omega_1}{6\omega_2}+\frac{\omega_2}{6\omega_1}+\frac{1}{2}
$$
is the second order Bernoulli polynomial appearing in the modular transformation
law for the theta function
\begin{equation}
\theta\left(e^{-2\pi\textup{i}\frac{u}{\omega_1}};
e^{-2\pi\textup{i}\frac{\omega_2}{\omega_1}}\right)
=e^{\pi\textup{i}B_{2,2}(u;\mathbf{\omega})} \theta\left(e^{2\pi\textup{i}\frac{u}{\omega_2}};
e^{2\pi\textup{i}\frac{\omega_1}{\omega_2}}\right).
\label{mod-theta}\end{equation}
The reflection equation for $G(u)$ has the form
$$
G(a,b;{\bf \omega}):=G(a;{\bf \omega})G(b;{\bf \omega})=1,
\qquad a+b=\sum_{k=1}^3\omega_k.
$$


\begin{thebibliography}{99}

\bibitem{Baxter}
R.~J.~Baxter, {\em Exactly Solved Models in Statistical
Mechanics}, Academic Press, London, 1982.

\bibitem{FT}
L. D. Faddeev and L. A. Takhtadzhan, {\it The quantum method of
inverse problem and the Heisenberg XYZ model}, Uspekhi Mat. Nauk {\bf 34} (5),
13--63 (Russian Math. Surveys {\bf 34} (5) (1979),  11--68).

\bibitem{KS1} P. P. Kulish and E. K. Sklyanin, {\it On the solutions
of the Yang-Baxter equation},
Zap. Nauchn. Sem. LOMI {\bf 95} (1980), 129--160.

\bibitem{KS2} P. P. Kulish and E. K. Sklyanin, {\it Quantum spectral
transform method. Recent developments},
Lect. Notes in Physics, vol. {\bf 151} (1981), 61--119.

\bibitem{Jimbo}
M. Jimbo (ed), {\it Yang-Baxter equation in integrable systems},
Adv. Ser. Math. Phys., 10, World Scientific (Singapore), 1990.

\bibitem{Faddeev}
L. D. Faddeev, {\it How algebraic Bethe ansatz works
for integrable model},  Quantum Symmetries,
Proc. Les-Houches summer school, LXIV, North-Holland, 1998, 149--211.

\bibitem{PA} J. H. H. Perk and H. Au-Yang, {\it Yang-Baxter equations},
{\tt math-ph/0606053}.

\bibitem{Baxter1}
R.~J.~Baxter,~{\it Partition function of the eight-vertex lattice model},
  Ann. Phys.  {\bf 70} (1972), 193--228 (reprinted in  Ann. Phys. {\bf 281}
(2000), 187--222).

\bibitem{Sklyanin1} E. K. Sklyanin, {\it On some algebraic structures
related to Yang-Baxter equation},
Funkz. Analiz i ego Pril. {\bf 16} (1982),  27--34.

\bibitem{Sklyanin2} E. K. Sklyanin, {\it On some algebraic structures
related to Yang-Baxter equation: representations of the quantum algebra},
Funkz. Analiz i ego Pril. {\bf 17} (1983), 34--48.

\bibitem{spi:essays}
V. P. Spiridonov,
{\em Essays on the theory of elliptic
hypergeometric functions}, Uspekhi Mat. Nauk {\bf 63} (3)
(2008), 3--72 (Russian Math. Surveys {\bf 63} (3) (2008), 405--472),
{\tt arXiv:0805.3135 [math.CA]}.

\bibitem{DKK0} S. E. Derkachov, D. Karakhanyan, and R. Kirschner,
{\em Universal R-matrix as integral operator}, Nucl. Phys.
{\bf B618} (2001), 589--616.

\bibitem{spi:umn}
V. P. Spiridonov,
\textit{On the elliptic beta function}, Uspekhi Mat. Nauk {\bf 56} (1)
(2001), 181--182 (Russian Math. Surveys {\bf 56} (1) (2001),
185--186).

\bibitem{S2} V. P. Spiridonov, \textit{Theta hypergeometric
integrals}, Algebra i Analiz {\bf 15} (6) (2003), 161--215 (St.
Petersburg Math. J. {\bf 15} (6) (2004), 929--967), {\tt math.CA/0303205}.

\bibitem{FTur}
I. Frenkel and V. Turaev, {\it Elliptic solutions
of the Yang-Baxter equation and modular
hypergeometric functions}, The Arnold-Gelfand Mathematical
Seminars (Cambridge, MA:Birkhauser Boston) (1997), pp. 171--204.

\bibitem{aar}
G. E. Andrews, R. Askey, and R. Roy, \textit{ Special Functions},
Encyclopedia of Math. Appl. {\bf 71}, Cambridge Univ. Press,
Cambridge, 1999.

\bibitem{spi:bailey}
V. P. Spiridonov, \textit{A Bailey tree for
integrals}, Teor. Mat. Fiz. {\bf 139} (2004), 104--111 (Theor. Math.
Phys. {\bf 139} (2004), 536--541), {\tt math.CA/0312502}.

\bibitem
{fad:mod} L. D. Faddeev,
{\em Modular double of a  quantum group}, Conf. Mosh\'e Flato 1999,
vol. I, Math. Phys. Stud. {\bf 21}, Kluwer, Dordrecht, 2000, pp. 149--156.

\bibitem{AA2008}
V. P. Spiridonov,
{\em Continuous biorthogonality of the elliptic hypergeometric function},
Algebra i Analiz {\bf 20} (5) (2008), 155--185
(St. Petersburg Math. J. {\bf 20} (5) (2009) 791--812), {\tt arXiv:0801.4137 [math.CA]}.

\bibitem{SD}
S. E. Derkachov,  {\em Factorization of the $R$-matrix. I.} 
Zap. Nauchn. Sem. POMI {\bf 335} (2006), 134--163
(J. Math. Sciences {\bf 143} (1) (2007), 2773-2790), {\tt math.QA/0503396}.

\bibitem{DM0}
S. E. Derkachov and A. N. Manashov, {\em $R$-matrix and Baxter $Q$-operators
for the noncompact $\mathrm{SL}(n,\mathbb{C})$
invariant spin chain}, SIGMA {\bf 2} (2006), 084.

\bibitem{DM}
S. Derkachov and A. Manashov, \textit{General solution of the Yang-Baxter
equation with the symmetry group $\mathrm{SL}(n,\mathbb{C})$},
Algebra i Analiz {\bf 21} (4) (2009), 1--94
(St. Petersburg Math. J. {\bf 21} (2010), 513--577).

\bibitem{DKK}
S. Derkachov, D. Karakhanyan, and R. Kirschner, \textit{Yang-Baxter
${\mathcal R}$-operators and parameter permutations},
Nucl. Phys.  {\bf B785} (2007), 263--285.

\bibitem{Fel}
B. Felderhof, \textit{Diagonalization of the transfer matrix of the free-fermion model},
Physica {\bf 66} (1973), no. 2, 279--198.

\bibitem{Kri}
I. M. Krichever, \textit{Baxter's equations and algebraic geometry},
Funct. An. and its Appl. {\bf  15} (1981), no. 2, 22--35.

\bibitem{Z1}
A. Zabrodin, {\it On the spectral curve of the difference Lame operator},
Int. Math. Research Notices, no. 11 (1999), 589--614.

\bibitem{Z2}
A. Zabrodin, {\it Commuting difference operators with elliptic
coefficients from Baxter's vacuum vectors},  J. Phys. A:
Math. Gen.  {\bf 33}  (2000), 3825--3850.

\bibitem{KRS}
P. P. Kulish, N. Yu. Reshetikhin, and E. K. Sklyanin, {\it Yang-Baxter
equation and representation theory}, Lett. Math. Phys. {\bf 5}
(1981), 393--403.

\bibitem{spi-war:inversions}
V. P. Spiridonov and S. O. Warnaar, \textit{
Inversions of integral operators and elliptic beta integrals on root
systems}, Adv. Math. {\bf 207} (2006), 91--132.

\bibitem{spi:conm} V. P. Spiridonov, {\em Elliptic beta integrals
and solvable models of statistical mechanics}, Contemp. Math.
{\bf 563} (2012), 181--211, {\tt arXiv:1011.3798 [hep-th]}.

\bibitem{BS}
V. V. Bazhanov and S. M. Sergeev, {\em A master solution of the quantum
Yang-Baxter equation and classical discrete integrable equations},
{\tt arXiv:1006.0651 [math-ph]}.

\bibitem{BS2}
V. V. Bazhanov and S. M. Sergeev, {\em Elliptic gamma-function and
multi-spin solutions of the Yang-Baxter equation},
{\tt arXiv:1106.5874 [math-ph]}.

\bibitem{VF}
A. Yu. Volkov and L. D. Faddeev, {\em Yang-Baxterization of the
quantum dilogarithm}, Zapiski POMI {\bf 224} (1995), 146--154
(J. Math. Sciences {\bf 88} (2) (1998), 202--207).

\bibitem{DO}
F.~A.~Dolan and H.~Osborn, \textit{Applications of the
superconformal index for protected operators and $q$-hypergeometric
identities to $\mathcal{N}=1$ dual theories}, Nucl. Phys. {\bf B818}
(2009), 137--178.

\bibitem{AGT}
L.~F.~Alday, D.~Gaiotto, and Y.~Tachikawa,
\textit{Liouville correlation functions from four-dimensional gauge theories},
Lett. Math. Phys. {\bf 91} (2010), 167--197.

\bibitem{GPRR} A. Gadde, E. Pomoni, L. Rastelli, and S. S. Razamat,
\textit{$S$-duality and $2d$ topological QFT},
J. High Energy Phys. {\bf 03} (2010), 032.

\bibitem{SV}
V.~P.~Spiridonov and G.~S.~Vartanov, \textit{Elliptic hypergeometry
of supersymmetric dualities}, Commun. Math. Phys. {\bf 304} (2011), 797--874.

\end{thebibliography}
\end{document}